\documentclass[12pt,a4paper]{article}
\usepackage{amssymb}
\usepackage{pstricks}
\usepackage{amsmath}
\usepackage{amsfonts}
\usepackage[onehalfspacing,nodisplayskipstretch]{setspace}
\usepackage[a4paper,nohead,left=1in,right=1in,top=1in,bottom=1in]{geometry}

\newtheorem{theorema}{Theorem}

\newtheorem{theorem}{Theorem}[section]

\newtheorem{corollary}[theorem]{Corollary}

\newtheorem{example}[theorem]{Example}

\newtheorem{lemma}[theorem]{Lemma}

\newtheorem{proposition}[theorem]{Proposition}
\newtheorem{remark}[theorem]{Remark}

\newenvironment{proof}[1][Proof]{\noindent\textbf{#1.} }{\ \rule{0.5em}{0.5em}}
\let\oldexample\example
\renewcommand{\example}{\oldexample\normalfont}
\let\oldremark\remark
\renewcommand{\remark}{\oldremark\normalfont}
\input{tcilatex}
\def\pQ{p_{\scriptscriptstyle Q}}
\begin{document}

\title{Monotonic Mechanisms for Selling Multiple Goods\thanks{%
First version: November 2021. Partially supported by the European Research
Council (ERC) under the European Union Horizon 2020 research and innovation
programme (grant agreement No 740282, PI: Noam Nisan).}}
\author{Ran Ben Moshe\thanks{%
The Hebrew University of Jerusalem (Department of Mathematics). Research
conducted towards an M.Sc. degree under the supervision of Sergiu Hart; see
Ben Moshe (2022).\quad \emph{E-mail}: \texttt{ran.ben-moshe1@mail.huji.ac.il}%
} \and Sergiu Hart\thanks{%
The Hebrew University of Jerusalem (Federmann Center for the Study of
Rationality, Department of Economics, and Institute of Mathematics).\quad 
\emph{E-mail}: \texttt{hart@huji.ac.il} \quad \emph{Web site}: \texttt{%
http://www.ma.huji.ac.il/hart}} \and Noam Nisan\thanks{%
The Hebrew University of Jerusalem (Federmann Center for the Study of
Rationality, and School of Computer Science and Engineering). \emph{E-mail}: 
\texttt{noam@cs.huji.ac.il} \quad \emph{Web site}: \texttt{%
http://www.cs.huji.ac.il/\symbol{126}noam}}}
\maketitle

\begin{abstract}
Maximizing the revenue from selling two or more goods has been shown to
require the use of \emph{nonmonotonic} mechanisms, where a higher-valuation
buyer may pay less than a lower-valuation one. Here we show that the
restriction to \emph{monotonic} mechanisms may not just lower the revenue,
but may in fact yield only a \emph{negligible fraction} of the maximal
revenue; more precisely, the revenue from monotonic mechanisms is no more
than $k$ times the simple revenue obtainable by selling the goods
separately, or bundled (where $k$ is the number of goods), whereas the
maximal revenue may be arbitrarily larger. We then study the class of
monotonic mechanisms and its subclass of allocation-monotonic mechanisms,
and obtain useful characterizations and revenue bounds.
\end{abstract}

\tableofcontents

%TCIMACRO{%
%\TeXButton{References Without Numbers}{\def\@biblabel#1{#1\hfill}
%\def\thebibliography#1{\section*{References}
%\addcontentsline{toc}{section}{References}
%\list
%{}{
%\labelwidth 0pt
%\leftmargin 1.8em
%\itemindent -1.8em
%\usecounter{enumi}}
%\def\newblock{\hskip .11em plus .33em minus .07em}
%\sloppy\clubpenalty4000\widowpenalty4000
%\sfcode`\.=1000\relax\def\baselinestretch{1}\large \normalsize}
%\let\endthebibliography=\endlist}}%
%BeginExpansion
\def\@biblabel#1{#1\hfill}
\def\thebibliography#1{\section*{References}
\addcontentsline{toc}{section}{References}
\list
{}{
\labelwidth 0pt
\leftmargin 1.8em
\itemindent -1.8em
\usecounter{enumi}}
\def\newblock{\hskip .11em plus .33em minus .07em}
\sloppy\clubpenalty4000\widowpenalty4000
\sfcode`\.=1000\relax\def\baselinestretch{1}\large \normalsize}
\let\endthebibliography=\endlist%
%EndExpansion

\section{Introduction\label{s:intro}}

Consider the basic problem of a single seller---a monopolist---who is
maximizing his revenue when selling \emph{multiple} goods to a single buyer,
in the standard Bayesian setup where the seller knows only the distribution
of the buyer's valuation of the goods (with arbitrary correlation between
the goods). Even in the simplest case where the buyer's valuation is
additive over sets of goods, the seller has neither cost nor value for the
goods, and both seller and buyer have quasilinear risk-neutral utilities,
the problem turns out to be extremely complex---unlike the single-good case,
where it is optimal to simply set up a price for the good (a
\textquotedblleft take-it-or-leave-it offer," Myerson 1981). For two or more
goods, optimal mechanisms are structurally and conceptually complex (McAfee
and McMillan 1988; Thanassoulis 2004; Maneli and Vincent 2006; Alaei, Fu,
Haghpanah, Hartline, and Malekian 2012; Daskalakis, Deckelbaum, and Tzamos
2017); they may require the use of randomizations, i.e., the buyer is
offered to buy lottery tickets\footnote{%
Despite the risk neutrality of both seller and buyer; in the risk-averse
case, lotteries are needed already in the single-good case (Siedner 2019).}
(Manelli and Vincent 2007, 2012; Chawla, Malec, and Sivan 2010; Daskalakis,
Deckelbaum, and Tzamos 2014); there may be arbitrarily many, even infinitely
many, outcomes (Daskalakis, Deckelbaum, and Tzamos 2013); and no simple
mechanism can guarantee a positive fraction of the optimal revenue (Briest,
Chawla, Kleinberg, and Weinberg 2010/2015;\footnote{%
By \textquotedblleft 2010/2015" we mean \textquotedblleft conference
proceeding in 2010 and journal publication in 2015."} Hart and Nisan
2013/2019; Dughmi, Han, and Nisan 2014).

One surprising feature is the \emph{nonmonotonicity of the maximal revenue}
(Hart and Reny 2015), where the buyer is willing to pay more for the goods,
yet ends up paying less! That is, there are situations where the buyer's
valuation of each good (his \textquotedblleft willingness to pay") goes up,%
\footnote{%
In the first-order dominance sense; i.e., the probability of higher
valuations increases.} and yet the seller's optimal revenue goes down. This
cannot happen when there is a single good, where incentive compatibility
implies that a buyer with a higher valuation must get more of the good%
\footnote{%
I.e., a higher allocation---which means a higher probability of getting the
good in the indivisible case, and a higher fraction of the good in the
divisible case.} and pay more for it. For two or more goods, this is however
no longer true. Consider the following example, from Hart and Reny (2015). 
%TCIMACRO{\TeXButton{BeginFigure}{\begin{figure}[ptb] \centering}}%
%BeginExpansion
\begin{figure}[ptb] \centering%
%EndExpansion
%
%
%
%
%
%
%
%
%
%
%
%
%
%
%
%
%
%
%
%
%
%
%
%
%
%
%
%
%
%
%
%
%
%\input{c:/0data/swp/monot-mech/m-m-fig.tex}
\begin{pspicture}(0,-12.5)(8,6)
     \rput(-3,0){
     \psset{unit=1.5cm}
     \psline[linewidth=2pt]{->}(0,0)(4.2,0)
     \rput(4.4,0){$x_1$}
     \psline[linewidth=2pt]{->}(0,-0.03)(0,4.2)
     \rput(0,4.35){$x_2$}
    \psline[linewidth=1pt]{-}(1,0)(1,2)
     \psline[linewidth=1pt]{-}(0,2)(1,2)
     \psline[linewidth=1pt]{-}(1,2)(2,3)
     \psline[linewidth=1pt]{-}(2,3)(2,4)
     \psline[linewidth=1pt]{-}(2,3)(4,3)
	%\psline[linewidth=0.5pt]{-}(2,-0.05)(2,0.05)
	%\rput(2,-0.2){{\tiny $2$}}
	%\psline[linewidth=0.5pt]{-}(1,-0.05)(1,0)
	%\rput(1,-0.2){{\tiny $1$}}
	%\psline[linewidth=0.5pt]{-}(3,-0.05)(3,0.05)
	%\rput(3,-0.2){{\tiny $3$}}
	%\psline[linewidth=0.5pt]{-}(-0.05,1)(0.05,1)
	%\rput(-0.15,1){{\tiny $1$}}
	%\psline[linewidth=0.5pt]{-}(-0.05,2)(0,2)
	%\rput(-0.15,2){{\tiny $2$}}
	%\psline[linewidth=0.5pt]{-}(-0.05,3)(0.05,3)
	%\rput(-0.15,3){{\tiny $3$}}
     \rput*[l](0.2,1){$q=(0,0)$}
     \rput[l](0.2,0.6){$s=0$}
     \rput[l](2.2,1.8){$q=(1,0)$}
     \rput[l](2.2,1.4){$s=1$} 
     \rput[l](0.5,3.4){$q=(0,1)$}
     \rput[l](0.5,3){$s=2$}
     \rput[l](2.5,3.9){$q=(1,1)$}
     \rput[l](2.5,3.5){$s=4$}
     %\psline[linewidth=4pt,linecolor=red]{->}(1,2.3)(2.1,2.6)
	\psline[linewidth=4pt,linecolor=red]{->}(1.08,2.36)(1.92,2.64)
     \pscircle*(1,2.333){0.05}
     \pscircle*(2,2.667){0.05}
     \rput(2,-0.6){\textbf{not} monotonic}
      \rput(2,-1){\textbf{not} allocation monotonic}
      \rput(2,-1.4){(supermodular non-symmetric)}
}
\rput(5,0){ \psset{unit=1.5cm}
     \psline[linewidth=2pt]{->}(0,0)(4.2,0)
     \rput(4.4,0){$x_1$}
     \psline[linewidth=2pt]{->}(0,-0.03)(0,4.2)
     \rput(0,4.35){$x_2$}
    \psline[linewidth=1pt]{-}(1,0)(1,1)
     \psline[linewidth=1pt]{-}(0,1)(1,1)
     \psline[linewidth=1pt]{-}(1,1)(3,3)
     \psline[linewidth=1pt]{-}(3,3)(3,4)
     \psline[linewidth=1pt]{-}(3,3)(4,3)
	%\psline[linewidth=0.5pt]{-}(2,-0.05)(2,0.05)
	%\rput(2,-0.2){{\tiny $2$}}
	%\psline[linewidth=0.5pt]{-}(1,-0.05)(1,0)
	%\rput(1,-0.2){{\tiny $1$}}
	%\psline[linewidth=0.5pt]{-}(3,-0.05)(3,0.05)
	%\rput(3,-0.2){{\tiny $3$}}
	%\psline[linewidth=0.5pt]{-}(-0.05,1)(0.05,1)
	%\rput(-0.15,1){{\tiny $1$}}
	%\psline[linewidth=0.5pt]{-}(-0.05,2)(0,2)
	%\rput(-0.15,2){{\tiny $2$}}
	%\psline[linewidth=0.5pt]{-}(-0.05,3)(0.05,3)
	%\rput(-0.15,3){{\tiny $3$}}
     \rput*[l](0.2,0.7){$q=(0,0)$}
     \rput[l](0.2,0.3){$s=0$}
     \rput[l](2.2,1.2){$q=(1,0)$}
     \rput[l](2.2,0.8){$s=1$} 
     \rput[l](0.5,2.4){$q=(0,1)$}
     \rput[l](0.5,2){$s=1$}
     \rput[l](3.5,3.9){$q=(1,1)$}
     \rput[l](3.5,3.5){$s=4$}
     %\psline[linewidth=4pt,linecolor=red]{->}(1,2.3)(2.1,2.6)
	\psline[linewidth=4pt,linecolor=red]{->}(1.08,1.36)(1.92,1.64)
     \pscircle*(1,1.333){0.05}
     \pscircle*(2,1.667){0.05}
\rput(2,-0.6){monotonic}
      \rput(2,-1){\textbf{not} allocation monotonic}
 \rput(2,-1.4){(supermodular symmetric)}
}
\rput(1,-10){ \psset{unit=1.5cm}
     \psline[linewidth=2pt]{->}(0,0)(4.2,0)
     \rput(4.4,0){$x_1$}
     \psline[linewidth=2pt]{->}(0,-0.03)(0,4.2)
     \rput(0,4.35){$x_2$}
    \psline[linewidth=1pt]{-}(1.5,0)(1.5,1.5)
     \psline[linewidth=1pt]{-}(0,2)(1,2)
     \psline[linewidth=1pt]{-}(1,2)(1.5,1.5)
     \psline[linewidth=1pt]{-}(1.5,1.5)(4,1.5)
     \psline[linewidth=1pt]{-}(1,2)(1,4)
	%\psline[linewidth=0.5pt]{-}(2,-0.05)(2,0.05)
	%\rput(2,-0.2){{\tiny $2$}}
	%\psline[linewidth=0.5pt]{-}(1,-0.05)(1,0)
	%\rput(1,-0.2){{\tiny $1$}}
	%\psline[linewidth=0.5pt]{-}(3,-0.05)(3,0.05)
	%\rput(3,-0.2){{\tiny $3$}}
	%\psline[linewidth=0.5pt]{-}(-0.05,1)(0.05,1)
	%\rput(-0.15,1){{\tiny $1$}}
	%\psline[linewidth=0.5pt]{-}(-0.05,2)(0,2)
	%\rput(-0.15,2){{\tiny $2$}}
	%\psline[linewidth=0.5pt]{-}(-0.05,3)(0.05,3)
	%\rput(-0.15,3){{\tiny $3$}}
     \rput[l](0.2,1.1){$q=(0,0)$}
     \rput[l](0.2,0.7){$s=0$}
     \rput[l](2.4,0.9){$q=(1,0)$}
     \rput[l](2.4,0.5){$s=1.5$} 
     \rput*[l](0.2,3){$q=(0,1)$}
     \rput[l](0.2,2.6){$s=2$}
     \rput[l](1.8,2.5){$q=(1,1)$}
     \rput[l](1.8,2.1){$s=3$}
     %\psline[linewidth=4pt,linecolor=red]{->}(1,2.3)(2.1,2.6)
	%\psline[linewidth=4pt,linecolor=red]{->}(1.08,1.36)(1.92,1.64)
     %\pscircle*(1,1.333){0.05}
     %\pscircle*(2,1.667){0.05}
\rput(2,-0.6){monotonic}
      \rput(2,-1){allocation monotonic}
 \rput(2,-1.4){(submodular)}
}
 \end{pspicture}%
\caption{Deterministic mechanisms for two goods ($q$ denotes 
the allocation of the two goods, and $s$  denotes the payment; see Section \ref{sus:model})\label{fig:mon}}%
%TCIMACRO{\TeXButton{EndFigure}{\end{figure}}}%
%BeginExpansion
\end{figure}%
%EndExpansion

\begin{example}
\emph{Nonmonotonic mechanism} (see Figure \ref{fig:mon}, top left). Consider
the deterministic mechanism for two goods where the price of the first good
is $1$, the price of the second good is $2$, and the price of the bundle of
both goods is $4$. A buyer who values the first good at $1$ and the second
good at $2.3$ will buy the second good and pay $2$ (he prefers getting $%
2.3-2=0.3$ from the second good to $1-1=0$ from the first good),\footnote{%
The other options are not better: they yield $0$ (when getting nothing) or $%
1+2.3-4=-0.7$ (when buying both goods); similarly for the other valuation $%
(2,2.7).$} whereas a buyer with the higher valuations of $2$ for the first
good and $2.7$ for the second will buy the first good instead and pay only $%
1 $ (he prefers $2-1=1$ to $2.7-2=0.7)$. This mechanism is \emph{nonmonotonic%
}: a higher valuation pays less than a lower valuation (here: $(1,2.3)$ pays 
$2$, and $(2,2.7)$ pays $1$).
\end{example}

The fact that such a mechanism exists is in itself not surprising. What is
surprising is that nonmonotonic mechanisms are \emph{needed} in order to
maximize the seller's revenue. As a consequence, in such cases the seller's
revenue goes down as the buyer's valuations go up (just move a small
probability mass from a low valuation that pays more to a high valuation
that pays less).

The first result of the present paper is that the restriction to using only
monotonic mechanisms may not just decrease the revenue, but may in fact
yield only a negligible portion of the optimal revenue. Indeed, we have (see
Theorem \ref{th:monrev} in Section \ref{s:mon-rev}):

\begin{quote}
$\bullet $ \textbf{Result A. }\emph{For every }$k$\emph{-good valuation }$X$%
\emph{, the maximal revenue }\textsc{MonRev}$(X)$ \emph{that is} \emph{%
obtained by monotonic mechanisms is no more than }$k$\emph{\ times the
revenue }\textsc{SRev}$(X)$\emph{\ that is obtained by selling the goods
separately, or the revenue }\textsc{BRev}$(X)$\emph{\ that is obtained by
selling the goods as a single bundle}: 
\begin{eqnarray*}
\text{\textsc{MonRev}}(X) &\leq &k\cdot \,\text{\textsc{SRev}}(X)\text{\ \ \
and} \\
\text{\textsc{MonRev}}(X) &\leq &k\cdot \,\text{\textsc{BRev}}(X).
\end{eqnarray*}
\end{quote}

Since the separate and bundled revenues (as well as the revenue from any
other simple mechanisms) can be arbitrarily small relative to the optimal
revenue (Briest et al. 2015 for $k\geq 3$, and Hart and Nisan 2019 for $%
k\geq 2$), from Result A we immediately get (see Theorem \ref{c:gfor(mon)}
in Section \ref{s:mon-rev}):

\begin{quote}
$\bullet $ \textbf{Result B.}\emph{\ For every number of goods }$k\geq 2$ 
\emph{there exist }$k$\emph{-good valuations }$X$\emph{\ whose optimal
revenue }\textsc{Rev}$(X)$ \emph{is infinite, yet all monotonic mechanisms
yield a bounded revenue}$\emph{:}$%
\begin{equation*}
\text{\textsc{Rev}}(X)=\infty \;\;\;\text{and}\;\;\;\text{\textsc{MonRev}}%
(X)=1
\end{equation*}
\end{quote}

\noindent (for bounded valuations, for every $\varepsilon >0$ there are $X$
with values in, say, $[0,1]^{k}$, such that \textsc{MonRev}$(X)<\varepsilon
\cdot $\thinspace \textsc{Rev}$(X)$).

The bound $k$ in Result A is tight relative to the bundled revenue \textsc{%
BRev} (since selling separately---which is a monotonic mechanism---may yield 
$k$ times \textsc{BRev}, already in the i.i.d. case), while relative to the
separate revenue \textsc{SRev} we do not know whether it is so: the largest
gap that we have obtained is only of the order of $\log k$.

We thus study the class of monotonic mechanisms.\footnote{%
See Rubinstein and Weinberg (2015) and Yao (2018) for some studies of
monotonic mechanisms.} Two classes of monotonic mechanisms have been
identified in Hart and Reny (2015): the class of \emph{symmetric
deterministic} mechanisms, and the class of \emph{submodular} mechanisms
(i.e., those where the price increase due to increasing the quantity of one
good is lower when the quantities of other goods are higher). The latter
mechanisms turn out to satisfy a stronger form of monotonicity, namely, 
\emph{allocation monotonicity}, which requires the allocation of goods to be
a nondecreasing function of the valuation. This is easily seen to imply the
monotonicity of the mechanism (i.e., the payment being a nondecreasing
function of the valuation; see Proposition \ref{p:amon -> mon}). An example
of such a mechanism:

\begin{example}
\emph{Monotonic (and thus allocation-monotonic) mechanism} (see Figure \ref%
{fig:mon}, bottom). Consider the submodular deterministic mechanism for two
goods where the price of good $1$ is $1.5$, the price of good $2$ is $2$,
and the price of the bundle is $3$ (which is less than $1.5+2$).
\end{example}

Allocation monotonicity is a strictly stronger requirement than
monotonicity; for example:

\begin{example}
\label{ex:mon-not-am}\emph{Monotonic but not allocation-monotonic mechanism }%
(see Figure \ref{fig:mon}, top right). Consider the symmetric deterministic
mechanism for two goods where the price of each single good is $1$ and the
price of the bundle is $4$ (which is more than $1+1$, and so the mechanism
is supermodular). When the valuation increases from, say, $(1,2.3)$ to $%
(2,1.7)$ the allocation changes from buying good $2$ to buying good $1$;
this contradicts allocation monotonicity, since the higher valuation $%
(2,1.7) $ gets less of good $2$ (it does not contradict monotonicity, since
the payment stays the same, at $1$).
\end{example}

Our next result is a characterization of allocation monotonicity for
deterministic mechanisms (Theorem \ref{th:am} (iv) in Section \ref{sus:a-mon
characterization}; note that characterization results require some
\textquotedblleft regularity" when breaking ties):

\begin{quote}
$\bullet $ \textbf{Result C1.}\emph{\ A deterministic mechanism is
allocation monotonic if and only if it is submodular.}
\end{quote}

\noindent This does not however hold for general, probabilistic mechanisms,
where submodularity implies allocation monotonicity (see Theorem \ref{th:am}
(ii) in Section \ref{sus:a-mon characterization}), but the converse is not
true: see Example \ref{ex:AM not subm}, which belongs to the interesting
class of \textquotedblleft quadratic mechanisms," introduced and studied in
Section \ref{s:quadratic}. In fact, allocation-monotonic mechanisms satisfy
a strict weakening of submodularity, namely, \emph{separable subadditivity}:
buying a bundle of goods, whether deterministic or probabilistic,\ costs no
more than buying the goods separately (one may refer to this as
\textquotedblleft subadditivity across goods"; it is in general weaker than
subadditivity, where the same requirement applies also to different
quantities of the \emph{same} good---see Section \ref{sus:s-modularity}). We
have (see Theorem \ref{th:am} (ii)--(iii) in Section \ref{sus:a-mon
characterization}):

\begin{quote}
$\bullet $ \textbf{Result C2.}\emph{\ Every submodular mechanism is
allocation monotonic, and every allocation-monotonic mechanism is separably
subadditive.}
\end{quote}

\noindent (Neither one of these implications is an equivalence when there
are multiple goods.)

The property of separable subadditivity is closely related to the notion of
\textquotedblleft sybil-proofness" of Chawla, Teng, and Tzamos (2022). We
appeal to their Theorem 1.3 (see our Appendix \ref{sus-a:chawla}) and obtain
(see Theorem \ref{th:amon-srev} in Section \ref{sus:a-mon rev}):

\begin{quote}
$\bullet $ \textbf{Result D. }\emph{For every }$k$\emph{-good valuation }$X$%
\emph{, the maximal revenue }\linebreak \textsc{AMonRev}$(X)$ \emph{that is} 
\emph{obtained by allocation-monotonic mechanisms is no more than }$O(\log
k) $\emph{\ times the revenue }\textsc{SRev}$(X)$\emph{\ that is obtained by
selling the goods separately}:\emph{\ }%
\begin{equation*}
\text{\textsc{AMonRev}}(X)\;\leq \;O(\log k)\cdot \,\text{\textsc{SRev}}(X).
\end{equation*}
\end{quote}

Next, we deal with the other class of monotonic mechanisms, namely, the
symmetric deterministic mechanisms, for which we obtain (see Theorem \ref%
{th:symdrev} in Section \ref{sus:sym det general}):

\begin{quote}
$\bullet $ \textbf{Result E. }\emph{For every }$k$\emph{-good valuation }$X$%
\emph{, the maximal revenue }\linebreak \textsc{SymDRev}$(X)$ \emph{that is} 
\emph{obtained by symmetric deterministic mechanisms (which are monotonic)
is no more than }$O(\log ^{2}k)$\emph{\ times the revenue }\textsc{SRev}$(X)$%
\emph{\ that is obtained by selling the goods separately}:\emph{\ }%
\begin{equation*}
\text{\textsc{SymDRev}}(X)\;\leq \;O(\log ^{2}k)\cdot \,\text{\textsc{SRev}}%
(X).
\end{equation*}
\end{quote}

We get this result by first studying the revenue of \emph{supermodular}
symmetric deterministic mechanisms, for which we obtain a tight bound of $%
\ln k$ relative to \textsc{SRev}, and then extend it to the entire class,
again by the Chawla, Teng, and Tzamos (2022) result. The above class of
supermodular mechanisms also provides a gap of $\ln k$ between the revenue
of allocation-monotonic mechanisms \textsc{AMonRev} and that of monotonic
mechanisms \textsc{MonRev} (see Example \ref{ex:harmonic} and Corollary \ref%
{c:amon vs mon}).

The paper is organized as follows. Section \ref{s:prelim} presents the basic
model, concepts, notation, and preliminaries, including in particular the
useful notion of the \emph{canonical} pricing function in Section \ref%
{susus:canonical p}. In Section \ref{s:mon-rev} we deal with the monotonic
revenue, and obtain Results A and B. Allocation-monotonic mechanisms---their
characterizations (Results C1 and C2) and revenue bound (Result D)---are
studied in Section \ref{s:a-mon}. The class of quadratic mechanisms is
introduced in Section \ref{s:quadratic}, providing, in Section \ref%
{sus:a-mon not submod}, an example of an allocation-monotonic mechanism that
is not submodular. Section \ref{s:sym det} is devoted to the class of
symmetric deterministic mechanisms---which are monotonic---and provides
Result E. Characterizations of monotonicity for general (non-symmetric)
deterministic mechanisms are presented in Section \ref{s:mon det}, and we
conclude with a number of problems that have remained open. The appendices
contain several complements and proofs; in particular, in Appendix \ref%
{s-a:det} we study general deterministic mechanisms (which need not be
monotonic), extending the analysis of the symmetric case of Section \ref%
{s:sym det}.

\section{Preliminaries\label{s:prelim}}

\subsection{The Model\label{sus:model}}

The notation follows Hart and Nisan (2017, 2019) and Hart and Reny (2015).

One seller (a monopolist) is selling a number $k\geq 1$ of indivisible
goods\ (or items, objects, etc.) to one buyer; let $K:=\{1,2,...,k\}$ denote
the set of goods. The goods have no cost or value to the seller, and their
values to the buyer are $x_{1},x_{2},...,x_{k}\geq 0$. The value of getting
a set of goods is \emph{additive}: each subset $I\subseteq K$ of goods is
worth $x(I):=\sum_{i\in I}x_{i}$ to the buyer (and so, in particular, the
buyer's demand is not restricted to one good only). The valuation of the
goods is given by a random variable $X=(X_{1},X_{2},...,X_{k})$ that takes
values in\footnote{%
For vectors $x=(x_{1},x_{2},...,x_{k})$ in $\mathbb{R}^{k},$ we write $x\geq
0$ when $x_{i}\geq 0$ for all $i,$ and $x\gg 0$ when $x_{i}>0$ for all $i.$
The nonnegative orthant is $\mathbb{R}_{+}^{k}=\{x\in \mathbb{R}^{k}:x\geq
0\},$ and $x\cdot y=\sum_{i=1}^{k}x_{i}y_{i}$ is the scalar product of $x$
and $y$ in $\mathbb{R}^{k}.$} $\mathbb{R}_{+}^{k}$; we will refer to $X$ as
a $k$\emph{-good} \emph{random valuation.} The realization $%
x=(x_{1},x_{2},...,x_{k})\in \mathbb{R}_{+}^{k}$ of $X$ is known to the
buyer, but not to the seller, who knows only the distribution $F$ of $X$
(which may be viewed as the seller's belief); we refer to a buyer with
valuation $x$ also as a buyer of \emph{type }$x$. The buyer and the seller
are assumed to be risk neutral and to have quasilinear utilities.

The objective is to \emph{maximize} the seller's (expected) \emph{revenue}.

As was well established by the so-called Revelation Principle\ (starting
with Myerson 1981; see for instance the book of Krishna 2010), we can
restrict ourselves to \textquotedblleft direct mechanisms" and
\textquotedblleft truthful equilibria.\textquotedblright\ A direct\emph{\
mechanism} $\mu $ consists of a pair of functions\footnote{%
All functions are assumed to be Borel-measurable; see Hart and Reny (2015),
footnotes 10 and 48.} $(q,s)$, where $q=(q_{1},q_{2},...,q_{k}):\mathbb{R}%
_{+}^{k}\rightarrow \lbrack 0,1]^{k}$ and $s:\mathbb{R}_{+}^{k}\rightarrow 
\mathbb{R}$, which prescribe the \emph{allocation} and the \emph{payment},
respectively. Specifically, if the buyer reports a valuation vector $x\in 
\mathbb{R}_{+}^{k}$, then $q_{i}(x)\in \lbrack 0,1]$ is the probability that
the buyer receives good\footnote{%
An alternative interpretation of the model: the goods are infinitely
divisible and the valuation is linear in quantity, in which case $q_{i}$ is
the quantity (i.e., fraction) of good $i$ that the buyer gets.} $i$ (for all 
$i=1,2,...,k$), and $s(x)$ is the payment that the seller receives from the
buyer. When the buyer of type $x$ reports his type truthfully, his payoff is 
$b(x)=\sum_{i=1}^{k}q_{i}(x)x_{i}-s(x)=q(x)\cdot x-s(x)$, and the seller's
payoff is $s(x).$

The mechanism $\mu =(q,s)$ satisfies \emph{individual rationality} (IR%
\textbf{)} if $b(x)\geq 0$ for every $x\in \mathbb{R}_{+}^{k}$; it satisfies 
\emph{incentive compatibility} (IC) if $b(x)\geq q(y)\cdot x-s(y)$ for every
alternative report $y\in \mathbb{R}_{+}^{k}$ of the buyer when his value is $%
x$, for every $x\in \mathbb{R}_{+}^{k}$.

The (expected) revenue of a mechanism $\mu =(q,s)$ from a buyer with random
valuation $X$, which we denote by $R(\mu ;X)$, is the expectation of the
payment received by the seller, i.e., $R(\mu ;X)=\mathbb{E}\left[ s(X)\right]
$. We now define

\begin{itemize}
\item \textsc{Rev}$(X)$, the \emph{optimal revenue}, is the maximal revenue
that can be obtained: \textsc{Rev}$(X)=\sup_{\mu }R(\mu ;X)$, where the
supremum is taken over the class of all IC and IR mechanisms $\mu .$
\end{itemize}

When there is only one good, i.e., when $k=1$, Myerson's (1981) result%
\footnote{%
See also Riley and Samuelson (1981) and Riley and Zeckhauser (1983).} is
that 
\begin{equation}
\text{\textsc{Rev}}(X)=\sup_{t\geq 0}t\cdot (1-F(t)),  \label{eq:one good}
\end{equation}%
where $F(t)=\mathbb{P}\left[ X\leq t\right] $ is the cumulative distribution
function of $X$. Thus, it is optimal for the seller to \textquotedblleft
post" a price $p$, and then the buyer buys the good for the price $p$
whenever his value is at least $p$; in other words, the seller makes the
buyer a \textquotedblleft take-it-or-leave-it" offer to buy the good at
price $p.$

Besides the maximal revenue \textsc{Rev}$(X)$, we are also interested in
what can be obtained from certain classes of mechanisms. For any class $%
\mathcal{N}$ of IC and IR mechanisms we denote

\begin{itemize}
\item $\mathcal{N}$-\textsc{Rev}$(X):=\sup_{\mu \in \mathcal{N}}R(\mu ;X)$,
the maximal revenue over the class $\mathcal{N}.$
\end{itemize}

\noindent In particular:

\begin{itemize}
\item \textsc{SRev}$(X)$, the \emph{separate revenue}, is the maximal
revenue that can be obtained by selling each good separately. Thus 
\begin{equation*}
\text{\textsc{SRev}}(X)=\text{\textsc{Rev}}(X_{1})+\text{\textsc{Rev}}%
(X_{2})+...+\text{\textsc{Rev}}(X_{k}).
\end{equation*}

\item \textsc{BRev}$(X)$, the \emph{bundling revenue}, is the maximal
revenue that can be obtained by selling all the goods together in one
\textquotedblleft bundle." Thus 
\begin{equation*}
\text{\textsc{BRev}}(X)=\text{\textsc{Rev}}(X_{1}+X_{2}+...+X_{k}).
\end{equation*}

\item \textsc{DRev}$(X)$, the \emph{deterministic revenue}, is the maximal
revenue that can be obtained by deterministic mechanisms; these are the
mechanisms in which every good $i=1,2,...,k$ is either fully allocated or
not at all, i.e., $q_{i}(x)\in \{0,1\}$ for all valuations $x\in \mathbb{R}%
_{+}^{k}$ (rather than $q_{i}(x)\in \lbrack 0,1])$.
\end{itemize}

\noindent As seen in Hart and Nisan (2017, Proposition 6), when maximizing
revenue we can limit ourselves without loss of generality to those IC and IR
mechanisms that satisfy in addition the \emph{no positive transfer} (NPT)
property, namely, $s(x)\geq 0$ for every $x\in \mathbb{R}_{+}^{k}$, from
which it follows that $s(\mathbf{0})=b(\mathbf{0})=0$, where $\mathbf{0}%
=(0,0,...,0)\in \mathbb{R}_{+}^{k}$.

From now on\emph{\ we will assume that all mechanisms }$\mu $\emph{\ are }$k$%
\emph{-good mechanisms that are given in direct form and satisfy IC, IR, and
NPT};\footnote{%
For some of the results only IC is needed.} thus,\emph{\ }$\mu =(q,s):%
\mathbb{R}_{+}^{k}\rightarrow \lbrack 0,1]^{k}\times \mathbb{R}_{+}$ and $b:%
\mathbb{R}_{+}^{k}\rightarrow \mathbb{R}_{+}$, and\emph{\ }$s(\mathbf{0})=b(%
\mathbf{0})=0.$

We conclude with a standard technical result. The function $b$ is \emph{%
nonexpansive} if $b(y)-b(x)\leq \sum_{i=1}^{k}(y_{i}-x_{i})$ for all $x\leq
y $ in $\mathbb{R}_{+}^{k}$; a vector $g\in \mathbb{R}^{k}$ is a \emph{%
subgradient} of the function $b$ at the point $x\in \mathbb{R}_{+}^{k}$ if $%
b(y)\geq b(x)+g\cdot (y-x)$ for every $y\in \mathbb{R}_{+}^{k}$; the set of
subgradients of $b$ at $x$ is denoted by $\partial b(x).$

\begin{proposition}
\label{p:b function}A function $b:\mathbb{R}_{+}^{k}\rightarrow \mathbb{R}%
_{+}$ is a buyer payoff function of some mechanism if and only if $b$ is
continuous, convex, nondecreasing, and nonexpansive. In this case, $b$ is
obtained from the mechanism $\mu =(q,s)$ if and only if $q(x)\in \partial
b(x)^{+}:=\{g\in \partial b(x):g\geq 0\}$ and $s(x)=q(x)\cdot x-b(x)$ for
every $x\in \mathbb{R}_{+}^{k}.$
\end{proposition}

See Appendix A.1 in Hart and Reny (2015) for details.\footnote{%
In Hart and Reny (2015) the function $b$ is extended to a convex function on
the entire space $\mathbb{R}^{k},$ and then $b$ is continuous and $\partial
b(x)\subseteq \lbrack 0,1]^{k}$ everywhere, and there is no need for $%
\partial b(x)^{+}.$ Here we have found it more natural to keep $\mathbb{R}%
_{+}^{k}$ as the domain of $b$ (and so $b(x)=\infty $ for $x\notin \mathbb{R}%
_{+}^{k}).$} The set $\partial b(x)^{+}$ differs from $\partial b(x)$ only
at boundary points $x$ of $\mathbb{\mathbb{R}}_{+}^{k}$: if $x_{i}=0$ then
the $i$-th coordinate $g_{i}$ of any subgradient $g\in \partial b(x)$ can be
lowered arbitrarily, and so be negative (see Theorem 25.6 in Rockafellar
1970).

\subsection{Pricing Functions\label{sus:pricing}}

An equivalent description of mechanisms is by means of pricing functions.

A \emph{pricing function} $p:[0,1]^{k}\rightarrow \mathbb{R}\cup \{\infty \}$
assigns to each allocation $g\in \lbrack 0,1]^{k}$ a price $p(g)$, which may
be infinite; it generates the set of choices (a \textquotedblleft potential
menu") $\mathcal{M}_{p}:=\{(g,p(g)):g\in \lbrack 0,1]^{k}\}$. Given a
mechanism $\mu =(q,s)$, we will say that $p$ is a \emph{pricing function} 
\emph{of }$\mu $ if for every buyer valuation $x\in \mathbb{R}_{+}^{k}$ the
choice $(q(x),s(x))$ of $\mu $ is an optimal choice from the set $\mathcal{M}%
_{p}$. This means, first, that $p(q(x))=s(x)$, and second, that $g\cdot
x-p(g)\leq q(x)\cdot x-p(q(x))=b(x)$ for all $g\in \lbrack 0,1]^{k}$, where $%
b$ is the buyer payoff function of $\mu $. Thus, the price of the allocation 
$q(x)$ equals the corresponding payment $s(x)$ (this is well defined since $%
q(x)=q(x^{\prime })$ implies that $s(x)=s(x^{\prime })$ by IC), and if one
can choose any allocation $g$ in $[0,1]^{k}$ for the corresponding price $%
p(g)$ then $q(x)$ is an optimal choice for a buyer with valuation $x$. Let $%
Q\equiv Q_{\mu }:=q(\mathbb{R}_{+}^{k})\equiv \{q(x):x\in \mathbb{R}%
_{+}^{k}\}\subseteq \lbrack 0,1]^{k}$ denote the \emph{range of allocations
of} $\mu $; while the price of any $g\in Q$ (i.e., any allocation $g$ that
is chosen at some valuation) is well defined, the prices of all $g\notin Q$
are not: they just need to be high enough so that these $g$ are never chosen
(which is easily achieved, for instance, by making the price of each $%
g\notin Q$ infinite\footnote{%
Infinite prices \emph{must} be allowed. For a simple example, take two goods
($k=2),$ and let $\mu $ sell the first good at price $1$ (i.e., $p(1,0)=1);$
any price function $p$ of $\mu $ must then put an infinite price on the
second good (i.e., $p(0,1)=\infty ).$}). To summarize: $p$ is a pricing
function of the mechanism $\mu $ with buyer payoff function $b$ if and only
if%
\begin{equation}
b(x)=\max_{g\in \lbrack 0,1]^{k}}\left( g\cdot x-p(g)\right) =q(x)\cdot
x-p(q(x))  \label{eq:b=p0*}
\end{equation}%
for every $x\in \mathbb{R}_{+}^{k}$; note that $b(\mathbf{0)}=0$ implies
that $p$ cannot have negative values, and so the range of $p$ is $\mathbb{R}%
_{+}\cup \{\infty \}=[0,\infty ].$

When $\mu $ is a deterministic mechanism, i.e., $Q\subseteq \{0,1\}^{k}$, it
is sometimes simpler and more convenient to consider \emph{deterministic}
pricing functions $p\equiv p^{\text{\textsc{D}}}$, which are defined only on
the set $\{0,1\}^{k}$ of deterministic allocations. In this case we identify 
$\{0,1\}^{k}$ with the set of subsets $2^{K}$ of $K=\{1,...,k\}$, and write $%
p(A)$ instead of $p(\mathbf{1}_{A})$ (where $\mathbf{1}_{A}$ denotes the
indicator vector of $A$, i.e., $(\mathbf{1}_{A})_{i}=1$ for $i\in A$ and $(%
\mathbf{1}_{A})_{i}=0$ for $i\notin A)$; thus, $p\equiv p^{\text{\textsc{D}}%
}:\{0,1\}^{k}\sim 2^{K}\rightarrow \lbrack 0,\infty ]$.

In general, each mechanism has at least one pricing function---perhaps more
than one, when not all allocations are chosen---and each pricing function
generates at least one mechanism---perhaps more than one, depending on the
way that the buyer breaks ties when indifferent. In the next two sections we
deal with these multiplicities and provide useful selections that simplify
the analysis.

\subsection{The Canonical Pricing Function\label{susus:canonical p}}

Among all pricing functions of a given mechanism there is a
\textquotedblleft canonical" one that turns out to be particularly
convenient and useful: the price of any unused allocation $g$ is set to be
the \emph{smallest} price that ensures that $g$ is never strictly optimal
(setting it to be the largest price, i.e., infinite, is unwieldy: for
instance, it may yield a pricing function that is at times decreasing).
Since $p(g)$ must satisfy the inequalities $g\cdot x-p(g)\geq b(x)$, i.e., $%
p(g)\geq g\cdot x-b(x)$, for all $x$, the smallest such price is $%
\sup_{x}(g\cdot x-b(x))$.

Following Hart and Reny (2015, Appendix A.2), we thus define the \emph{%
canonical} pricing function $p_{0}$ of the mechanism $\mu $, with buyer
payoff function $b$, by%
\begin{equation}
p_{0}(g)%
%TCIMACRO{\TeXButton{:=}{{\;:=\;}}}%
%BeginExpansion
{\;:=\;}%
%EndExpansion
\sup_{x\in \mathbb{R}_{+}^{k}}\left( g\cdot x-b(x)\right)  \label{eq:p0}
\end{equation}%
for every $g\in \lbrack 0,1]^{k}$. As seen above, $p$ is a pricing function
of $\mu $ if and only if $p_{0}(g)\leq p(g)\leq \infty $ for every $g\in
\lbrack 0,1]^{k}$, and $p(g)=p_{0}(g)$ for every $g\in Q$, and so $p_{0}$ is
the \emph{minimal} pricing function of $\mu $. The function $p_{0}$ is
nondecreasing, convex, and closed (i.e., lower semicontinuous\footnote{%
The function $p_{0}$ is \emph{lower semicontinuous} if $\underline{\lim }%
_{h\rightarrow g}p_{0}(h)\geq p_{0}(g)$ for every $g\in \lbrack 0,1]^{k}.$
At relatively interior points $g$ of $\mathrm{dom}\,p_{0}=\{g:p_{0}(g)<%
\infty \}$ the function $p_{0}$ is in fact continuous, since it is a convex;
see Rockafellar (1970), Theorems 9.4 and 10.2.}), because it is the supremum
of such functions; also, $p_{0}(\mathbf{0})=0$ (because $b(\mathbf{0})=0$).
In Proposition \ref{p:p0} in Appendix \ref{sus-a:canonical p} we will show
that the canonical pricing function is in fact the \emph{unique} pricing
function that is nondecreasing, convex, and closed. Moreover, the buyer
payoff function $b$ and the canonical pricing function $p_{0}$ are Fenchel
conjugates; see (\ref{eq:b=p0*}) for $p=p_{0}$ and (\ref{eq:p0}).

When $\mu $ is a deterministic mechanism, i.e., $Q\subseteq \{0,1\}^{k}$,
the restriction to $\{0,1\}^{k}$ of the canonical pricing function $p_{0}$
will be called the \emph{canonical deterministic} pricing function of $\mu $%
, and denoted by $p_{0}^{\text{\textsc{D}}}$; thus, $p_{0}^{\text{\textsc{D}}%
}:\{0,1\}^{k}\rightarrow \lbrack 0,\infty ]$ is given by (\ref{eq:p0}) for
every $g\in \{0,1\}^{k}$. When $\{0,1\}^{k}$ is identified with $2^{K}$, the
set of subsets of $K=\{1,...,k\}$, (\ref{eq:p0}) becomes%
\begin{equation*}
p_{0}^{\text{\textsc{D}}}(A):=\sup_{x\in \mathbb{R}_{+}^{k}}(x(A)-b(x)),
\end{equation*}%
for every $A\subseteq K$, where $x(A)=\sum_{i\in A}x_{i}$. From the
properties of $p_{0}$ on $[0,1]^{k}$ it follows that $p_{0}^{\text{\textsc{D}%
}}$ is a nondecreasing function, i.e., if $A\subseteq B\subseteq K$ then $%
p_{0}^{\text{\textsc{D}}}(A)\leq p_{0}^{\text{\textsc{D}}}(B)$ (convexity
and closedness are irrelevant here); in fact, as we will now show, $p_{0}^{%
\text{\textsc{D}}}$ is the \emph{unique} pricing function that is
nondecreasing. Let $%
%TCIMACRO{\TeXButton{p_Q}{p_{\scriptscriptstyle Q}}}%
%BeginExpansion
p_{\scriptscriptstyle Q}%
%EndExpansion
:Q\rightarrow \mathbb{R}_{+}$ denote the common restriction to $Q$ of all
pricing functions of $\mu $.

\begin{proposition}
\label{p:p0-det}Let $\mu $ be a deterministic mechanism. Then the canonical
deterministic pricing function $p_{0}^{\text{\textsc{D}}}$ of $\mu $ is the
unique deterministic pricing function of $\mu $ that is nondecreasing, and
it is given by%
\begin{equation}
p_{0}^{\text{\textsc{D}}}(A)=\inf \{%
%TCIMACRO{\TeXButton{p_Q}{\pQ}}%
%BeginExpansion
\pQ%
%EndExpansion
(B):B\in Q,~B\supseteq A\}  \label{eq:p0D}
\end{equation}%
for every $A\subseteq K.$
\end{proposition}

\begin{proof}
Let $p_{2}(A)$ denote the right-hand side of (\ref{eq:p0D}). When $%
A\subseteq B\in Q$ we get $p_{0}(A)\leq p_{0}(B)=%
%TCIMACRO{\TeXButton{p_Q}{\pQ}}%
%BeginExpansion
\pQ%
%EndExpansion
(B)$ (because $p_{0}$ is nondecreasing), and then the infimum over $B\in Q$
yields $p_{0}(A)\leq p_{2}(A)$, with equality when $A\in Q$ (because then $%
B=A$ is included). Thus $p_{2}$ is indeed a pricing function of $\mu .$

Assume by way of contradiction that $p_{0}(A)<p_{2}(A)$ for some $A\notin Q$%
. Consider the valuation $x=M\mathbf{1}_{A}$ for large $M>0$ (specifically, $%
M>\max_{B\in Q}%
%TCIMACRO{\TeXButton{p_Q}{\pQ}}%
%BeginExpansion
\pQ%
%EndExpansion
(B)$). We claim that $(A,p_{0}(A))$ is strictly better at $x$ than $(B,%
%TCIMACRO{\TeXButton{p_Q}{\pQ}}%
%BeginExpansion
\pQ%
%EndExpansion
(B))$ for any $B\in Q$; indeed, if $B\supseteq A$ then $p_{0}(A)<p_{2}(A)%
\leq 
%TCIMACRO{\TeXButton{p_Q}{\pQ}}%
%BeginExpansion
\pQ%
%EndExpansion
(B)$, and so $x(B)-%
%TCIMACRO{\TeXButton{p_Q}{\pQ}}%
%BeginExpansion
\pQ%
%EndExpansion
(B)=M\left\vert A\right\vert -%
%TCIMACRO{\TeXButton{p_Q}{\pQ}}%
%BeginExpansion
\pQ%
%EndExpansion
(B)<M\left\vert A\right\vert -p_{0}(A)=x(A)-p_{0}(A)$, and if $B\nsupseteq A$
then at least one coordinate of $A$ is not included in $B$, and so $%
x(A)-x(B)\geq M$, which, for large enough $M$, yields $x(B)-%
%TCIMACRO{\TeXButton{p_Q}{\pQ}}%
%BeginExpansion
\pQ%
%EndExpansion
(B)<x(A)-p_{0}(A)$. Thus $(A,p_{0}(A))$ is the unique optimal choice at $x$,
in contradiction to $A\notin Q.$

Uniqueness: let $p$ be any nondecreasing pricing function of $\mu $; since $%
p $ coincides with $%
%TCIMACRO{\TeXButton{p_Q}{\pQ}}%
%BeginExpansion
\pQ%
%EndExpansion
$ on $Q$ and is nondecreasing, it follows that $p\leq p_{2}$ (because $%
p(A)\leq p(B)=%
%TCIMACRO{\TeXButton{p_Q}{\pQ}}%
%BeginExpansion
\pQ%
%EndExpansion
(B)$ for every $A\subseteq B\in Q$, and so $p(A)\leq p_{2}(A))$; now $%
p_{2}=p_{0}$, the minimal pricing function (proved above), and so $p\leq
p_{2}=p_{0}$ yields $p=p_{0}.$
\end{proof}

\subsection{Favorable Tie Breaking\label{sus:tie-fav}}

{A pricing function $p$ does not determine the mechanism in case of ties,
i.e., when two or more allocations provide the same maximal payoff for some
buyer valuation:\ $b(x)=g\cdot x-p(g)=g^{\prime }\cdot x-p(g^{\prime })$}
for some $x$ and distinct $g$ and $g^{\prime }${. More generally, ties occur
at points }$x$ where the{\ buyer payoff function }$b$ (which is uniquely
determined by $p$; see (\ref{eq:b=p0*})) is not differentiable and there are
multiple subgradients in $\partial b(x)^{+}$ (see Proposition \ref{p:b
function}). {In order to fully specify the mechanism one needs to specify
how ties are broken in such cases. The \textquotedblleft favorable"
tie-breaking rules below---which maximize the payment and the
allocation---are particularly convenient, because they simplify the analysis
while affecting neither the maximization of revenue nor the relevant
mechanism properties (such as monotonicity and allocation monotonicity).
Moreover, these properties could well fail if one were to break ties
arbitrarily; see Hart and Reny (2015) (the reason that we need to go beyond
their seller favorability is that for allocation monotonicity the choice of
allocation matters as well).\footnote{%
What is needed is a certain consistency when breaking ties. Indeed, some of
the characterization results below hold under other tie-breaking rules, such
as always choosing a \emph{minimal}\textbf{\ }allocation; we focus on
tie-favorable mechanisms because these are the ones that maximize revenue.}}

We will say that a mechanism $\mu =(q,s)$, with buyer payoff function $b$, is

\begin{itemize}
\item \emph{seller favorable} (Hart and Reny 2015)\emph{\ }if for every $%
x\in \mathbb{R}_{+}^{k}$ the payment $s(x)$ is maximal; i.e., there is no $%
g\in \partial b(x)^{+}$ such that $g\cdot x-b(x)>s(x);$

\item $\emph{buyer}$ \emph{favorable} if for every $x\in \mathbb{R}_{+}^{k}$
the allocation $q(x)$ is maximal; i.e., there is no $g\in \partial b(x)^{+}$
such that $g\geq q(x)$ and $g\neq q(x);$

\item \emph{tie favorable} if it is both seller and buyer favorable.
\end{itemize}

We will at times refer to a mechanism $\mu ^{\prime }$ with the same buyer
payoff function $b$ as a (tie-breaking) \emph{version} of the mechanism $\mu 
$. Since the canonical pricing function $p_{0}$ is determined by $b$, all
tie-breaking versions of $\mu $\ have the same canonical pricing function.

The seller\emph{-}favorability condition is equivalent to the subgradient $%
q(x)\in \partial b(x)^{+}$ being maximal in the direction $x$ , i.e., $%
q(x)\in \partial b(x)_{x}^{+}:=\arg \max_{g\in \partial b(x)^{+}}g\cdot x$,
and then $s(x)=b^{\prime }(x;x)-b(x)$, where $b^{\prime }(y;z)$ denotes the
derivative of $b$ at $y$ in the direction $z$; see Hart and Reny 2015,
Appendix A.1. Tie favorability requires in addition that $q(x)$ be a maximal
element of the set $\partial b(x)_{x}^{+}$ (and thus of $\partial b(x)$ and $%
\partial b(x)^{+}$ as well; at interior points $x\gg 0$ any seller-favorable
choice is buyer favorable, and hence tie favorable). Since there are always
such choices (the sets $\partial b(x)^{+}$ and $\partial b(x)_{x}^{+}$ are
nonempty, compact, convex, and $\subseteq \lbrack 0,1]^{k}$), there always
exist seller-, buyer-, and tie-favorable versions of any mechanism.

The restriction to seller-favorable mechanisms is without loss of generality
when maximizing revenue, because these have the highest payments and
revenues; moreover, all seller-favorable versions of the mechanism with
buyer payoff function $b$ have identical payment functions $s$ (given by the
above formula), and so they always yield identical revenues. Therefore the
further restriction to the subclass of tie-favorable mechanisms is also
without loss of generality when maximizing revenue.

Finally, as we will show in Proposition \ref{p:s-fav monot} and Corollary %
\ref{c:am-fav}, monotonicity and allocation monotonicity are preserved when
one considers tie-favorable versions of such mechanisms. This of course
applies also to any property (such as submodularity; see Section \ref%
{sus:s-modularity}) that depends only on the buyer payoff function $b$ or
only on the canonical pricing function $p_{0},$ which are common to all
versions of a mechanism.

\subsection{Submodularity and Supermodularity\label{sus:s-modularity}}

Notions of submodularity and supermodularity will play an important role in
our analysis and results. We will use these notions for functions that are
defined not only on discrete domains (such as deterministic pricing
functions on $\{0,1\}^{k}\sim 2^{K}$) but also on continuous domains (such
as general pricing functions on $[0,1]^{k}$, and buyer payoff functions on $%
\mathbb{R}_{+}^{k}$).

Let $\mathcal{L}\subseteq \mathbb{R}_{+}^{k}$ be a lattice with respect to
coordinatewise maximum and minimum; i.e., for any $x=(x_{i})_{i=1,...,k}$
and $y=(y_{i})_{i=1,...,k}$ in $\mathcal{L}$ the vectors $x\vee y:=(\max
\{x_{i},y_{i}\})_{i=1,...,k}$ and $x\wedge y:=(\min
\{x_{i},y_{i}\})_{i=1,...,k}$ belong to $\mathcal{L}$ as well; in our use, $%
\mathcal{L}$ will be $\mathbb{R}_{+}^{k}$, $[0,1]^{k}$, or $\{0,1\}^{k}$.
The vectors $x,y\geq 0$ are \emph{orthogonal} if $x\cdot y=0$, which is the
same as $x\wedge y=\mathbf{0}$; i.e., $x$ and $y$ are positive on disjoint
sets of coordinates.

Let $f:\mathcal{L}\rightarrow \mathbb{R\cup \{}\infty \}$ be a function
(such as a buyer payoff function $b$, or a pricing function $p$; the value $%
\infty $ is thus allowed).

\begin{itemize}
\item $f$ is \emph{submodular} if for every $x,y\in \mathcal{L}$,%
\begin{equation}
f(x)+f(y)\geq f(x\vee y)+f(x\wedge y).  \label{eq:submod-f}
\end{equation}

\item $f$ is \emph{supermodular} if for every $x,y\in \mathcal{L}$,%
\begin{equation}
f(x)+f(y)\leq f(x\vee y)+f(x\wedge y)  \label{eq:supermod-f}
\end{equation}%
(for convex functions, supermodularity is equivalent to the stronger
requirement of \emph{ultramodularity}; see the discussion below).

\item $f$ is \emph{separably subadditive} if for every $x,y,x+y\in \mathcal{L%
}$ such that $x$ and $y$ are orthogonal,%
\begin{equation}
f(x+y)\leq f(x)+f(y).  \label{eq:subadd-f}
\end{equation}

\item $f$ is \emph{separably superadditive} if for every $x,y,x+y\in 
\mathcal{L}$ such that $x$ and $y$ are orthogonal,%
\begin{equation}
f(x+y)\geq f(x)+f(y).  \label{eq:superadd-f}
\end{equation}
\end{itemize}

A number of comments:

\emph{Submodularity and supermodularity}. The submodularity condition (\ref%
{eq:submod-f}) can be rewritten as follows: put $d^{1}:=x-x\wedge y\geq 0$
and $d^{2}:=y-x\wedge y\geq 0$; then $d^{1}$ and $d^{2}$ are orthogonal
(because $d_{i}^{1}>0$ when $x_{i}>y_{i}$, and $d_{i}^{2}>0$ when $%
y_{i}>x_{i}$), and replacing $x\wedge y$ with $x$ yields $%
f(x+d^{1})+f(x+d^{2})\geq f(x+d^{1}+d^{2})+f(x)$. Thus: $f$ is submodular if
and only if%
\begin{equation}
f(x+d^{2})-f(x)\geq f(x+d^{1}+d^{2})-f(x+d^{1})  \label{eq:subm-d}
\end{equation}%
for every $x\in \mathcal{L}$ and $d^{1},d^{2}\geq 0$ such that $d^{1}$ and $%
d^{2}$ are \emph{orthogonal} (assume that all vectors are in the domain
where $f$ is finite). The interpretation of (\ref{eq:subm-d}) is that the
contribution of $d^{2}$ to $f$ can only decrease as we increase $x$ by $%
d^{1} $; since $d^{1}$ and $d^{2}$ are orthogonal, this means that the
change in $f $ when we increase some coordinates can only be smaller when 
\emph{other} coordinates are larger. (This equivalent formulation shows that
(\ref{eq:submod-f}) is the appropriate definition of submodularity for
pricing functions; cf. Babaioff, Nisan, and Rubinstein 2018).

When $f$ is twice differentiable, it is thus submodular if and only if its
second-order partial derivatives satisfy%
\begin{equation}
\frac{\partial ^{2}f}{\partial x_{i}\partial x_{j}}(x)\leq 0\text{ for all }%
i\neq j;  \label{eq:d2p<=0}
\end{equation}%
i.e., all \emph{off-diagonal} elements of the Hessian matrix $\nabla
^{2}f(x) $ are nonpositive: the $i$-th partial derivative $\partial
f(x)/\partial x_{i}$ at $x$ can only decrease when some coordinate $x_{j}$
with $j$ \emph{different} from $i$ increases. We emphasize that
submodularity does \emph{not} entail (\ref{eq:subm-d}) for \emph{all} $%
d^{1},d^{2}\geq 0$, but only for orthogonal $d^{1},d^{2}\geq 0$; thus, there
is no requirement on the diagonal elements of the Hessian matrix. In fact,
the function $f$ may well be convex,\footnote{%
For example, the function $f(x)=(x_{1}-x_{2})^{2}$ is submodular and convex.}
in which case the opposite inequality holds for $d^{1}=d^{2}$, and the
diagonal elements of the Hessian matrix are nonnegative, i.e., $\partial
^{2}f(x)/\partial x_{i}^{2}\geq 0$ for all $i$.

Similarly, supermodularity is equivalent to%
\begin{equation}
f(x+d^{2})-f(x)\leq f(x+d^{1}+d^{2})-f(x+d^{1})  \label{eq:super-d}
\end{equation}%
for \emph{orthogonal} $d^{1},d^{2}\geq 0$, which means that the marginals of 
$f$ can only increase when \emph{other} coordinates increase (and so $%
\partial ^{2}f(x)/\partial x_{i}\partial x_{j}\geq 0$ for all $i\neq j$ in
the twice-differentiable case).

\emph{Ultramodularity}. A function $f$ is \emph{ultramodular} if (\ref%
{eq:super-d}) holds for \emph{all} $d^{1},d^{2}\geq 0$ (whether or not they
are orthogonal); when $f$ is twice differentiable, this translates to all
the elements of the Hessian matrix $\nabla ^{2}f(x)$ being nonnegative: $%
\partial ^{2}f(x)/\partial x_{i}\partial x_{j}\geq 0$ for \emph{all}\textbf{%
\ }$i,j$. For convex functions, supermodularity is equivalent to the
stronger condition of ultramodularity (because (\ref{eq:super-d}) holds for $%
i=j$ by convexity; see Corollary 4.1 in Marinacci and Montrucchio 2005).

\emph{Separable subadditivity and superadditivity.} First, we note that
separable subadditivity is weaker than subadditivity, which requires $%
f(x+y)\leq f(x)+f(y)$ to hold not only for orthogonal $x$ and $y,$ but for 
\emph{all} $x$ and $y$ (the two properties are equivalent when the domain $%
\mathcal{L}$ is $\{0,1\}^{k},$ because then $x+y\in \mathcal{L}$ if and only
if $x$ and $y$ are orthogonal). The same applies to separable
superadditivity vs. superadditivity. For functions $f$ that satisfy $f(%
\mathbf{0})=0$ supermodularity implies separable superadditivity, and
submodularity implies separable subadditivity (because for orthogonal
vectors $x,y$ we have $x\vee y=x+y$ and $x\wedge y=\mathbf{0}$).

Next, the separable subadditivity condition (\ref{eq:subadd-f}) can be
equivalently expressed as follows. Given a partition of $K$ into two
disjoint sets $K^{\prime }$ and $K^{\prime \prime }$, every vector $x\in 
\mathbb{R}_{+}^{k}$ can be expressed as the sum $x=x^{\prime }+x^{\prime
\prime }$ of the orthogonal vectors $x^{\prime }$ and $x^{\prime \prime }$
with supports $K^{\prime }$ and $K^{\prime \prime }$, respectively (put $%
x_{i}^{\prime }=x_{i}$ and $x_{i}^{\prime \prime }=0$ for $i\in K^{\prime }$%
, and $x_{i}^{\prime }=0$ and $x_{i}^{\prime \prime }=x_{i}$ for $i\in
K^{\prime \prime }$).\footnote{%
When $\mathcal{L}$ is $\mathbb{R}_{+}^{k},$ $[0,1]^{k},$ or $\{0,1\}^{k},$
if $x\in \mathcal{L}$ then $x^{\prime },x^{\prime \prime }\in \mathcal{L}.$}
Then $f$ is separably subadditive if and only if%
\begin{equation*}
f(x)\leq f(x^{\prime })+f(x^{\prime \prime })
\end{equation*}%
for every $x$ and every partition $K=K^{\prime }\cup K^{\prime \prime }$.
Applying this repeatedly yields 
\begin{equation*}
f(x)\leq \sum_{i=1}^{k}f(x_{i}e^{i}),
\end{equation*}%
where $e^{i}\in \mathbb{R}_{+}^{k}$ denotes the $i$-th unit vector (and so $%
x=\sum_{i}x_{i}e^{i}$). The same applies to separable superadditivity, with
all inequalities reversed.

Finally, if $f$ is both submodular and supermodular, i.e., (\ref{eq:submod-f}%
) and (\ref{eq:supermod-f}) hold as equalities, then (\ref{eq:subadd-f}) and
(\ref{eq:superadd-f}) hold as equalities, and $f$ is \emph{separably additive%
} (or \textquotedblleft additive across coordinates"), i.e., $%
f(x)=\sum_{i=1}^{k}f(x_{i}e^{i})$ (where $e^{i}$ is the $i$-th unit vector).

Turning now to mechanisms, we define:

\begin{itemize}
\item A mechanism $\mu $ is \emph{submodular }/\emph{\ supermodular }/\emph{%
\ separably subadditive} /\emph{\ separably superadditive }if its canonical
pricing function $p:[0,1]^{k}\rightarrow \lbrack 0,\infty ]$ is,
respectively, supermodular / submodular / separably subadditive /\ separably
superadditive\emph{.}
\end{itemize}

\noindent Thus, $\mu $ is submodular if%
\begin{equation}
p(g)+p(h)\geq p(g\vee h)+p(g\wedge h)  \label{eq:sub-m}
\end{equation}%
holds for its canonical pricing function $p$ for every $g,h\in \lbrack
0,1]^{k}$; i.e., (see (\ref{eq:subm-d}) and (\ref{eq:d2p<=0})), the marginal
price of good $i$ can only decrease as the quantity of good $j$ increases.
For deterministic mechanisms it suffices that (\ref{eq:sub-m}) holds for
deterministic allocations, i.e., for sets of goods: 
\begin{equation}
p(A)+p(B)\geq p(A\cup B)+p(A\cap B)  \label{eq:subm-sets}
\end{equation}%
for all $A,B\subseteq K$ (which is easily seen to be equivalent to%
\begin{equation}
p(A\cup \{i\})-p(A)\geq p(A\cup \{i,j\})-p(A\cup \{j\})
\label{eq:submod-i-j}
\end{equation}%
for all $A\subset K$ and $i\neq j$ not in $A$); Proposition \ref{p:subm-p0}
in Appendix \ref{sus-a:submod p} will show that it suffices that (\ref%
{eq:subm-sets}) holds \emph{only} for $A$ and $B$ in the range of
allocations $Q$ of the mechanism and, moreover, for \emph{some} pricing
function $p$, not necessarily the canonical one.

\subsection{Monotonicity and Allocation Monotonicity\label{sus:mon and amon}}

Let $\mu =(q,s)$ be a mechanism.

\begin{itemize}
\item $\mu $ is \emph{monotonic} (Hart and Reny 2015) if its payment
function $s$ is nondecreasing, i.e., $s(x)\leq s(y)$ for every two
valuations $x\leq y$ in $\mathbb{R}_{+}^{k}.$

\item $\mu $ is \emph{allocation monotonic} if its allocation function $q$
is nondecreasing,\footnote{%
Adding the IC inequalities $q(x)\cdot x-s(x)\geq q(y)\cdot x-s(y)$ (at $x)$
and $q(y)\cdot y-s(y)\geq q(x)\cdot y-s(x)$ (at $y)$ gives $(q(y)-q(x))\cdot
(y-x)\geq 0,$ which does \emph{not} yield $q(y)\geq q(x)$ when $y\geq x$
unless we are in the one-dimensional case of $k=1$ (see Hart and Reny 2015).}
i.e., $q(x)\leq q(y)$ for every two valuations $x\leq y$ in $\mathbb{R}%
_{+}^{k}.$
\end{itemize}

Let \textsc{MonRev}$(X)$ and \textsc{AMonRev}$(X)$ denote the maximal
revenue that can be achieved by monotonic and allocation-monotonic
mechanisms, respectively, for a $k$-good random valuation $X$. In the case
of one good, i.e., when $k=1$, every IC mechanism is monotonic and
allocation monotonic (this follows from the IC conditions: see the proof of
Proposition 2 in Hart and Reny 2015), and so \textsc{AMonRev}$(X)=\,$\textsc{%
MonRev}$(X)=\,$\textsc{Rev}$(X).$

Two immediate observations: allocation monotonicity is a stronger
requirement than monotonicity, and both \textsc{MonRev} and \textsc{AMonRev}%
---unlike \textsc{Rev}---are nondecreasing with respect to first-order
stochastic dominance: they can only increase when the buyer valuations,
i.e., his \textquotedblleft willingness to pay," increase.

\begin{proposition}
\label{p:amon -> mon}If the mechanism $\mu $ is allocation monotonic then it
is monotonic.
\end{proposition}

\begin{proof}
If $q(y)\geq q(x)$ but $s(y)<s(x)$ then $(q(y),s(y))$ is always strictly
better than $(q(x),s(x))$ (a higher allocation at a strictly cheaper
payment), and so the latter cannot be the optimal choice at $x.$
\end{proof}

\begin{proposition}
\label{p:rev mon}Let $X$ and $Y$ be random valuations. If $Y$ first-order
stochastically dominates $X$ then%
\begin{equation*}
\text{\textsc{MonRev}}(Y)\geq \text{\textsc{MonRev}}(X)\text{\ \ \ and\ \ \ 
\textsc{AMonRev}}(Y)\geq \text{\textsc{AMonRev}}(X).
\end{equation*}
\end{proposition}

\begin{proof}
Without loss of generality assume that the random variables $X$ and $Y$ are
\textquotedblleft coupled," i.e., that they are defined on the same
probability space and $X\leq Y$ pointwise. Then, for every
(allocation-)monotonic mechanism $\mu =(q,s)$ we have $s(X)\leq s(Y)$
(because $s$ is nondecreasing), and so $R(\mu ;X)=\mathbb{E}\left[ s(X)%
\right] \leq \mathbb{E}\left[ s(Y)\right] =R(\mu ;Y)$. Taking the supremum
over all (allocation-)monotonic mechanisms $\mu $ yields the result.
\end{proof}

\bigskip

Two classes of monotonic mechanisms have been identified in Hart and Reny
(2015): the class of deterministic symmetric mechanisms, and the class of
submodular mechanisms (Propositions 5 and 8 there). We will consider these
two classes in Sections \ref{s:sym det} and \ref{s:a-mon} below. While the
mechanisms in the first class need not be allocation monotonic---see Example %
\ref{ex:mon-not-am} in the Introduction---those in the second class are
allocation monotonic (the argument that proves Proposition 8 in Hart and
Reny 2015 can be easily adapted to show this; see also Section \ref{s:a-mon}
below).

\section{Revenue of Monotonic Mechanisms\label{s:mon-rev}}

In this section we prove that monotonic mechanisms for $k$ goods \emph{cannot%
} extract more than $k$ times the separate revenue, and also no more than $k$
times the bundled revenue (Theorem A), which is very different from
general---and hence necessarily nonmonotonic---mechanisms for $k\geq 2$,
which can extract an arbitrarily large multiple of these (and any other
simple) revenues (Theorem B).

\begin{theorema}
\label{th:monrev}Let $X$ be a $k$-good random valuation. Then%
\begin{equation*}
\text{\textsc{MonRev}}(X)\leq k\cdot \min \{\text{\textsc{SRev}}(X),\,\text{%
\textsc{BRev}}(X)\}.
\end{equation*}
\end{theorema}

Since selling separately and selling the bundle of all goods are monotonic
mechanisms, we trivially have 
\begin{equation}
\text{\textsc{MonRev}}(X)\geq \max \{\text{\textsc{SRev}}(X),\,\text{\textsc{%
BRev}}(X)\}.  \label{eq:SBrev}
\end{equation}

An immediate implication of Theorem \ref{th:monrev} and the results of Hart
and Nisan (2019)\footnote{%
See Briest et al. (2015) for $k\geq 3$.} is that for two or more goods \emph{%
monotonic mechanisms cannot guarantee any positive fraction of the optimal
revenue. }The Guaranteed Fraction of Optimal Revenue (\textsc{GFOR}) for
monotonic mechanisms is thus zero, the same as that of any class of
\textquotedblleft simple mechanisms," and the revenue is mostly obtained
from \emph{nonmonotonic} mechanisms. Theorem \ref{th:monrev} implies that
monotonic mechanisms are no better, in terms of revenue, than mechanisms
with a menu size that is at most the number of goods. Formally:

\begin{theorema}
\label{c:gfor(mon)}For every $k\geq 2$:

\begin{description}
\item[(i)] There exists a $k$-good random valuation $X$ such that%
\begin{equation*}
\text{\textsc{MonRev}}(X)=1\text{\ \ and\ \ \textsc{Rev}}(X)=\infty .
\end{equation*}

\item[(ii)] For every $\varepsilon >0$ there exists a $k$-good random
valuation with bounded values (in, say, $[0,1]^{k}$) such that%
\begin{equation*}
\text{\textsc{MonRev}}(X)<\varepsilon \cdot \text{\textsc{Rev}}(X).
\end{equation*}

\item[(iii)] There exists a $k$-good random valuation $X$ such that%
\begin{equation*}
\text{\textsc{MonRev}}(X)\leq \frac{k^{2}}{2^{k}-1}\cdot \text{\textsc{DRev}}%
(X).
\end{equation*}
\end{description}
\end{theorema}

\begin{proof}
Use \textsc{MonRev}$(X)\leq k\cdot $\textsc{BRev}$(X)$ and Theorems A and D
in Hart and Nisan (2019).
\end{proof}

\bigskip

To prove Theorem \ref{th:monrev}, we start with two preliminary results. A $%
k $-good valuation $X$ all of whose coordinates are equal, i.e., with
support included in the set $\{ye:y\geq 0\}\subset \mathbb{R}_{+}^{k}$,
where $e:=(1,...,1)\in \mathbb{R}_{+}^{k}$, is called a \emph{diagonal}
valuation. A diagonal valuation is thus 
\begin{equation*}
(\underbrace{Y,...,Y}_{k})\equiv Ye,
\end{equation*}%
where $Y$ is a one-good random valuation$.$

\begin{lemma}
\label{l:Ye}Let $Ye$ be a $k$-good diagonal valuation. Then%
\begin{equation*}
\text{\textsc{Rev}}(Ye)=k\cdot \text{\textsc{Rev}}(Y).
\end{equation*}
\end{lemma}

\begin{proof}
The restriction of any (IC and IR) $k$-good mechanism $\mu =(q,s)$ to the
diagonal $\{ye:y\geq 0\}\subset $ $\mathbb{R}_{+}^{k}$ yields an (IC\ and
IR) one-good mechanism $\nu =(\hat{q},\hat{s})$ given by $\hat{q}%
(y):=(1/k)\sum_{i=1}^{k}q_{i}(ye)$ and $\hat{s}(y):=(1/k)s(ye)$ for every $%
y\geq 0$. Conversely, any (IC and IR) one-good mechanism $\nu =(\hat{q},\hat{%
s})$ yields an (IC and IR) $k$-good mechanism $\mu =(q,s)$ given by $q(x):=(%
\hat{q}(\bar{x}),...,\hat{q}(\bar{x}))\in \lbrack 0,1]^{k}$ and $s(x):=k\hat{%
s}(\bar{x})$, where $\bar{x}:=(1/k)\sum_{i=1}^{k}x_{i}$, for every $x\in 
\mathbb{R}_{+}^{k}$. For the revenue, we have $R(\mu ;Ye)=k\cdot R(\nu ;Y).$
\end{proof}

\bigskip

The following result may be of independent interest.

\begin{proposition}
\label{p:max}Let $X_{1},...,X_{n}$ be one-good random valuations.\footnote{%
With arbitrary correlation between the $X_{i}$-s.} Then%
\begin{equation*}
\text{\textsc{Rev}}\left( \max_{1\leq i\leq n}X_{i}\right) \leq \min \left\{
\sum_{i=1}^{n}\text{\textsc{Rev}}(X_{i}),\,\text{\textsc{Rev}}\left(
\sum_{i=1}^{n}X_{i}\right) \right\} .
\end{equation*}
\end{proposition}

\begin{proof}
For \textsc{SRev}: for every $t\geq 0$ we have%
\begin{equation*}
\mathbb{P}\left[ \max_{1\leq i\leq n}X_{i}\geq t\right] =\mathbb{P}\left[
\bigcup\nolimits_{1\leq i\leq n}\{X_{i}\geq t\}\right] \leq \sum_{i=1}^{n}%
\mathbb{P}\left[ X_{i}\geq t\right] ,
\end{equation*}%
and thus%
\begin{equation*}
t\cdot \mathbb{P}\left[ \max_{1\leq i\leq n}X_{i}\geq t\right] \leq
\sum_{i=1}^{n}t\cdot \mathbb{P}\left[ X_{i}\geq t\right] \leq \sum_{i=1}^{n}%
\text{\textsc{Rev}}(X_{i}).
\end{equation*}%
Taking the supremum over $t\geq 0$ yields \textsc{Rev}$\left(
\max_{i}X_{i}\right) \leq \sum_{i}$\textsc{Rev}$(X_{i})$ by the Myerson
result for one good (\ref{eq:one good}).

For \textsc{BRev}: since $\max_{i}X_{i}\leq \sum_{i}X_{i}$ and the revenue
is monotonic for one good (Proposition 11 in Hart and Nisan 2017) we get 
\textsc{Rev}$\left( \max_{i}X_{i}\right) \leq \,$\textsc{Rev}$%
(\sum_{i}X_{i}).$
\end{proof}

\bigskip

\noindent \textbf{Remarks. }\emph{(a) }The result of Proposition \ref{p:max}
is tight: in the following example we have \textsc{Rev}$(\max_{i}X_{i})=%
\sum_{i}\text{\textsc{Rev}}(X_{i})=\,\text{\textsc{Rev}}(\sum_{i}X_{i})$.
Let $(X_{1},...,X_{n})$ take as values the $n$ unit vectors in $\mathbb{R}%
_{+}^{n}$, with probability $1/n$ each; then $\max_{i}X_{i}=\sum_{i}X_{i}=1$
and so \textsc{Rev}$(\max_{i}X_{i})=\text{\textsc{Rev}}(\sum_{i}X_{i})=1$;
and, for each $i$ we have \textsc{Rev}$(X_{i})=1/n$ (obtained by selling at
price $1$), and so $\sum_{i}$\textsc{Rev}$(X_{i})=1.$

\emph{(b) }Let \textsc{SymSRev}$(X)$ denote the maximal revenue that is
obtained by \emph{symmetric} separate mechanisms; by Myerson's result for
one good (\ref{eq:one good}), this is achieved by selling each one of the
goods at the same price $t$, i.e., 
\begin{equation*}
\text{\textsc{SymSRev}}(X)=\sup_{t\geq 0}\sum_{i=1}^{n}t\cdot \mathbb{P}%
\left[ X_{i}\geq t\right]
\end{equation*}%
(whereas \textsc{SRev}$(X)=\sum_{i=1}^{n}\sup_{t\geq 0}t\cdot \mathbb{P}%
\left[ X_{i}\geq t\right] $). The proof of (i) above shows that%
\begin{equation*}
\text{\textsc{Rev}}\left( \max_{1\leq i\leq n}X_{i}\right) \leq \min \left\{ 
\text{\textsc{SymSRev}}(X_{1},...,X_{n}),\,\text{\textsc{BRev}}%
(X_{1},...,X_{n})\right\} .
\end{equation*}%
This tightening of \textsc{SRev} to \textsc{SymSRev} applies to Theorem \ref%
{th:monrev} as well.\footnote{%
Another instance where one can replace \textsc{SRev} with \textsc{SymSRev}
is in Theorem A.7 in Hart and Nisan (2019), whose proof shows that $\int
1/w(t)\,\mathrm{d}t$ is in fact a bound on the multiple of the \emph{%
symmetric} separate revenue.}

\bigskip

We can now prove Theorem \ref{th:monrev}. See Appendix \ref{s-a:monrev} for
an easy generalization.

\bigskip

\begin{proof}[Proof of Theorem \protect\ref{th:monrev}]
Let $X$ be a $k$-good random valuation, and put $Y:=\max_{1\leq i\leq
k}X_{i} $. Since $X\leq Ye$, for any monotonic mechanism $\mu =(q,s)$ we
have $s(X)\leq s(Ye)$, and so, taking expectation, $R(\mu ;X)\leq R(\mu ;Ye)$%
. Taking the supremum over all monotonic $\mu $ yields \textsc{MonRev}$%
(X)\leq \,$\textsc{Rev}$(Ye)$, which equals $k\cdot \,$\textsc{Rev}$%
(Y)=k\cdot \,$\textsc{Rev}$(\max_{i}X_{i})$ by Lemma \ref{l:Ye}; thus,%
\begin{equation}
\text{\textsc{MonRev}}(X)\leq k\cdot \text{\textsc{Rev}}\left( \max_{1\leq
i\leq n}X_{i}\right) .  \label{eq:k-max}
\end{equation}%
Proposition \ref{p:max} completes the proof.
\end{proof}

\bigskip

\noindent \textbf{Remarks. }\emph{(a) }The multiple $k$ is tight for \textsc{%
BRev} in Theorem \ref{th:monrev}: Example 27 in Hart and Nisan (2017) shows
that for every $\varepsilon >0$ there are $k$ independent goods such that 
\begin{equation*}
\text{\textsc{MonRev}}(X)\geq \,\text{\textsc{SRev}}(X)\geq (k-\varepsilon
)\cdot \,\text{\textsc{BRev}}(X)
\end{equation*}%
(the first inequality is due to selling separately being a monotonic
mechanism).

\emph{(b) }When $k=2$ the multiple $k$ is tight for \textsc{SRev} in Theorem %
\ref{th:monrev}: Example A.3 in Hart and Nisan (2019) yields \textsc{MonRev}$%
(X)=\,$\textsc{BRev}$(X)=2\cdot \,$\textsc{SRev}$(X).$

When $k\geq 3$ the best-known lower bound is of the order of $\log k$:
Proposition 25 in Hart and Nisan (2017) shows that for $k$ independent 
\textsc{ER} goods we have%
\begin{equation}
\text{\textsc{MonRev}}(X)\geq \,\text{\textsc{BRev}}(X)\geq \Omega (\log
k)\cdot \,\text{\textsc{SRev}}(X)  \label{eq:brev=logk*srev}
\end{equation}%
(the second inequality cannot be improved because \textsc{BRev}$(X)\leq
O(\log k)\cdot \,$\textsc{SRev}$(X)$ for any $X$ by Proposition A.4 in Hart
and Nisan 2019). See also Example \ref{ex:harmonic} below, where \textsc{%
MonRev}$(X)\geq \ln k\,\cdot \,$\textsc{SRev}$(X).$

\section{Allocation-Monotonic Mechanisms\label{s:a-mon}}

In this section we provide characterizations of allocation-monotonic
mechanisms (Theorem C), and then show that these mechanisms can yield at
most an $O(\log k)$ multiple of the separate revenue (Theorem D).

\subsection{Characterization of Allocation-Monotonic Mechanisms\label%
{sus:a-mon characterization}}

Hart and Reny (2015)\ proved that (seller-favorable) submodular mechanisms
are monotonic; the proof there shows that they are in fact allocation
monotonic. In this section we will show that for deterministic mechanisms
the converse is true, i.e., allocation monotonicity is equivalent to
submodularity; however, this equivalence does not hold for probabilistic
mechanisms (see Section \ref{sus:a-mon not submod}; the correct
characterization of allocation monotonicity is the supermodularity of the
Fenchel conjugate of the canonical pricing function, which is the buyer
payoff function). Since these results may not hold for arbitrary
tie-breaking rules, we consider buyer-favorable mechanisms, which suffice
when maximizing revenue (see Section \ref{sus:tie-fav}), and preserve
allocation monotonicity, as we show in Corollary \ref{c:am-fav} below.

We now state the characterization result.

\begin{theorema}
\label{th:am}Let $\mu $ be a buyer-favorable mechanism. Then:

\begin{description}
\item[(i)] $\mu $ is allocation monotonic if and only if the buyer payoff
function $b$ of $\mu $ is supermodular.

\item[(ii)] If $\mu $ is submodular then $\mu $ is allocation monotonic.

\item[(iii)] If $\mu $ is allocation monotonic then $\mu $ is separably
subadditive.

\item[(iv)] When $\mu $ is a \emph{deterministic} mechanism: $\mu $ is
allocation monotonic if\ and only if $\mu $ is submodular.
\end{description}
\end{theorema}

When there are multiple goods, the converse of (ii) is not true for general,
non-deterministic, mechanisms: see Example \ref{ex:AM not subm} in Section %
\ref{sus:a-mon not submod}; also, the converse of (iii) is not true: see
Appendix \ref{sus-a:sep-subadd}. Thus, the class of buyer-favorable
allocation-monotonic mechanisms lies between the strictly smaller class of
submodular mechanisms and the strictly larger class of separably subadditive
mechanisms.

The proof of Theorem \ref{th:am} consists of showing the connections between
properties of a mechanism $\mu $ and those of its buyer utility function $b$
and its canonical pricing function $p$. These are summarized in the two
tables below, one for general, probabilistic, mechanisms, and the other for
deterministic mechanisms. The vertical relations are proved in the three
Propositions \ref{p:AM2supM}, \ref{p:p-b-mod}, and \ref{p:p-b-add} that
follow, and the horizontal implications are immediate (because $b(\mathbf{0}%
)=0$ and $p(\mathbf{0})=0$; see Section \ref{sus:s-modularity}).%
\begin{equation*}
\begin{tabular}{|clc|}
\hline
\multicolumn{3}{|c|}{\textbf{General mechanisms}} \\ \hline\hline
\multicolumn{3}{|l|}{$\mu $ allocation monotonic} \\ 
%TCIMACRO{\TeXButton{Updownarrow}{{\large $\Updownarrow$}} }%
%BeginExpansion
{\large $\Updownarrow$}
%EndExpansion
[P\ref{p:AM2supM}] &  &  \\ 
$b$ supermodular & 
%TCIMACRO{\TeXButton{Rightarrow}{{\large $\Rightarrow$}} }%
%BeginExpansion
{\large $\Rightarrow$}
%EndExpansion
& $b$ separably superadditive \\ 
%TCIMACRO{\TeXButton{Uparrow}{{\large $\Uparrow$}} }%
%BeginExpansion
{\large $\Uparrow$}
%EndExpansion
[P\ref{p:p-b-mod}] &  & 
%TCIMACRO{\TeXButton{Updownarrow}{{\large $\Updownarrow$}} }%
%BeginExpansion
{\large $\Updownarrow$}
%EndExpansion
[P\ref{p:p-b-add}] \\ 
$p$ submodular & 
%TCIMACRO{\TeXButton{Rightarrow}{{\large $\Rightarrow$}} }%
%BeginExpansion
{\large $\Rightarrow$}
%EndExpansion
& $p$ separably subadditive \\ 
($\mu $ submodular) &  & ($\mu $ separably subadditive) \\ \hline
\end{tabular}%
\;\;\;\;\;\;%
\begin{tabular}{|c|}
\hline
\textbf{Deterministic mechanisms} \\ \hline\hline
$\mu $ allocation monotonic \\ 
%TCIMACRO{\TeXButton{Updownarrow}{{\large $\Updownarrow$}} }%
%BeginExpansion
{\large $\Updownarrow$}
%EndExpansion
[P\ref{p:AM2supM}] \\ 
$b$ supermodular \\ 
%TCIMACRO{\TeXButton{Updownarrow}{{\large $\Updownarrow$}} }%
%BeginExpansion
{\large $\Updownarrow$}
%EndExpansion
[P\ref{p:p-b-mod}] \\ 
$p$ submodular \\ 
($\mu $ submodular) \\ \hline
\end{tabular}%
\end{equation*}

The first proposition shows the equivalence between allocation monotonicity
and supermodularity of the buyer payoff function $b$. This is easy to see
when $b$ is twice differentiable: allocation monotonicity means that $%
q_{i}(x)=\partial b(x)/\partial x_{i}$ is nondecreasing in $x$ for all $i$,
i.e., $\partial ^{2}b(x)/\partial x_{i}\partial x_{j}\geq 0$ for all $i,j$;
this is equivalent to the ultramodularity of $b$, which, since $b$ is a
convex function, is the same as the supermodularity of $b$.

\begin{proposition}
\label{p:AM2supM}Let $\mu $ be a mechanism, with buyer payoff function $b$.
Then:

\begin{description}
\item[(i)] If $\mu $ is allocation monotonic then $b$ is supermodular.

\item[(ii)] If $b$ is supermodular then the unique buyer-favorable (and thus
tie-favorable) version $\tilde{\mu}$ of $\mu $ is allocation monotonic.
\end{description}
\end{proposition}

\begin{proof}
(i) Assume that the function $q$ is nondecreasing; we will show that $b$ is
ultramodular, i.e., $b(x+d^{2})-b(x)\leq b(x+d^{1}+d^{2})-b(x+d^{1})$ for
every $x$ in $\mathbb{R}_{+}^{k}$ and every $d^{1},d^{2}\geq 0$.

First, we claim that for every $x\in \mathbb{R}_{+}^{k}$ and $d\geq 0$ we
have%
\begin{equation}
b(x+d)-b(x)=d\cdot \int_{0}^{1}q(x+td)\,\mathrm{d}t.  \label{eq:b-d}
\end{equation}%
Indeed, let $\beta (t):=b(x+td)$ for $t\geq 0$; then $\beta $ is a convex
function of $t$, and the IC inequality\footnote{%
Use $b(x)=q(x)\cdot x-s(x)$ and $b(x+\delta d)\geq q(x)\cdot (x+\delta
d)-s(x).$}%
\begin{equation*}
b(x+(t+\delta )d)-b(x+td)\geq q(x+td)\cdot \delta d
\end{equation*}%
yields, after we divide by $\delta $ and take the limit as $\delta
\rightarrow 0^{-}$ and $\delta \rightarrow 0^{+}$,%
\begin{equation*}
\beta _{-}^{\prime }(t)\leq q(x+td)\cdot d\leq \beta _{+}^{\prime }(t)
\end{equation*}%
for every $t>0$; applying Theorem 24.2 and Corollary 24.2.1 in Rockafellar
(1970) proves (\ref{eq:b-d}).

Next, for every $x\in \mathbb{R}_{+}^{k}$ and every $d^{1},d^{2}\geq 0$, we
have%
\begin{eqnarray*}
b(x+d^{2})-b(x) &=&d^{2}\cdot \int_{0}^{1}q(x+td^{2})\,\mathrm{d}t \\
&\leq &d^{2}\cdot \int_{0}^{1}q(x+d^{1}+td^{2})\,\mathrm{d}%
t=b(x+d^{1}+d^{2})-b(x+d^{1}),
\end{eqnarray*}%
where the inequality is by $d^{1}\geq 0$ (since $q$ is nondecreasing) and $%
d^{2}\geq 0$, and the equalities are by (\ref{eq:b-d}). Therefore $b$ is
ultramodular.

(ii) If $b$ is ultramodular then $\nabla ^{+}b(x):=(b^{\prime
}(x;e^{i}))_{i=1,...,k}$ is the maximal subgradient of $b$ at $x$, for every 
$x\in \mathbb{R}_{+}^{k}$, and $\nabla ^{+}b(x)$ is nondecreasing in $x$ (by
Lemma 5.1 and Theorem 5.4 (iii) in Marinacci and Motrucchio 2005).\footnote{%
Marinacci and Montrucchio (2005) prove these claims only for interior points 
$x$; the proof extends to boundary points as well, as shown in Proposition %
\ref{p:delta+} in Appendix \ref{sus-a:utramodular}.} Thus $\tilde{q}%
(x)=\nabla ^{+}b(x)$ is the unique buyer-favorable choice at $x$, and $%
\tilde{q}$ is nondecreasing.
\end{proof}

\bigskip

This immediately implies that buyer favorability, and hence tie
favorability, preserves allocation monotonicity; when maximizing revenue,
the restriction to tie-favorable versions of allocation-monotonic mechanisms
is thus without loss of generality.

\begin{corollary}
\label{c:am-fav}Let $\mu $ be a mechanism, with buyer payoff function $b$,
and let $\tilde{\mu}$ be a buyer-favorable mechanism with the same buyer
payoff function $b$. If $\mu $ is allocation monotonic then $\tilde{\mu}$ is
allocation monotonic.
\end{corollary}

\begin{proof}
If $\mu $ is allocation monotonic then $b$ is supermodular by Proposition %
\ref{p:AM2supM} (i), and then the buyer-favorable mechanism $\tilde{\mu}$
with the same $b$ is allocation monotonic by (ii) of the same proposition.
\end{proof}

\bigskip

The next proposition provides the relation between the supermodularity of $b$
and the submodularity of its Fenchel conjugate $p$.

\begin{proposition}
\label{p:p-b-mod}Let $\mu $ be a mechanism, with buyer payoff function $b$
and canonical pricing function $p$.

\begin{description}
\item[(i)] If $p$ is submodular then $b$ is supermodular.

\item[(ii)] When $\mu $ is a \emph{deterministic} mechanism: $p$ is
submodular if and only if $b$ is supermodular.
\end{description}
\end{proposition}

The converse of (i) does not hold in the general, non-deterministic case:
see Section \ref{sus:a-mon not submod} below. Assuming differentiability
(and ignoring technicalities), $p$ is submodular if and only if the
off-diagonal entries of the Hessian matrix $\nabla ^{2}p$ are $\leq 0$; this
implies that the off-diagonal entries of its inverse $(\nabla ^{2}p)^{-1},$
which is the Hessian matrix $\nabla ^{2}b$ of $p$'s Fenchel conjugate%
\footnote{%
See Crouzeix (1977).} $b$, are $\geq 0$, i.e., $b$ is supermodular; the
converse is however not true in general (cf. Proposition \ref{p:quad} and
Example \ref{ex:AM not subm}).

\bigskip

\begin{proof}
(i) Since for any $x,y\in \mathbb{R}_{+}^{k}$ and $g,h\in Q$ we have%
\begin{equation*}
g\cdot x+h\cdot y\leq (g\vee h)\cdot (x\vee y)+(g\wedge h)\cdot (x\wedge y)
\end{equation*}%
(because\footnote{%
This is the ultramodularity of the function $(x,g)\mapsto g\cdot
x=\sum_{i}g_{i}x_{i}$ (all its second-order partial derivatives are either $%
1 $ or $0$).} the difference between the right-hand side and the left-hand
side is $(g-g\wedge h)\cdot (y-x\wedge y)+(h-g\wedge h)\cdot (x-x\wedge
y)\geq 0$), subtracting the submodularity inequality (\ref{eq:sub-m}), i.e., 
$p(g)+p(h)\geq p(g\vee h)+p(g\wedge h)$, gives%
\begin{eqnarray*}
\lbrack g\cdot x-p(g)]+[h\cdot y-p(h)] &\leq &[(g\vee h)\cdot (x\vee
y)-p(g\vee h)]+[(g\wedge h)\cdot (x\wedge y)-p(g\wedge h)] \\
&\leq &b(x\vee y)+b(x\wedge y)
\end{eqnarray*}%
(the inequality on the second line because $p$ is a pricing function; see (%
\ref{eq:b=p0*}). Taking the supremum over $g\in Q$ and $h\in Q$ yields,
again by (\ref{eq:b=p0*}), $b(x)+b(y)\leq b(x\vee y)+b(x\wedge y)$, and so $%
b $ is supermodular.

(ii) Let $\mu $ be a deterministic mechanism, with range of allocations $%
Q\subseteq 2^{K}$, and let $p:2^{K}\rightarrow \lbrack 0,\mathbb{\infty ]}$
be its canonical deterministic pricing function. Thus, $b(x)=\max_{A%
\subseteq K}(x(A)-p(A))=\max_{A\in Q}(x(A)-p(A))$ for every $x\in \mathbb{R}%
_{+}^{k}$, and $p(A)=\sup_{x\in \mathbb{R}_{+}^{k}}\left( x(A)-b(x)\right) $
for every $A\subseteq K.$

Fix a large $M>0$ (specifically, $M\geq \max_{A\in Q}p(A)$); for $x=M\mathbf{%
1}_{A}$ (with $A\subseteq K$) we have%
\begin{equation*}
b(x)\geq x(A)-p(A)=M\left\vert A\right\vert -p(A).
\end{equation*}%
Moreover, if $A\in Q$ then 
\begin{equation*}
b(x)=x(A)-p(A)=M\left\vert A\right\vert -p(A)
\end{equation*}%
(i.e., $A$ is optimal at $x$). Indeed, $x(A)-p(A)\geq x(B)-p(B)$ for every $%
B\subseteq K$, because if $B\supseteq A$ then $x(B)=x(A)$ and $p(A)\leq p(B)$%
, and if $B\nsupseteq A$ then $x(B)=x(A)-M\left\vert A\backslash
B\right\vert \leq x(A)-M$ (since $A\backslash B\neq \emptyset )$ and $%
p(A)-p(B)\leq p(A)\leq M$ (since $A\in Q$).

Let $A,B\in Q$. Take $x=M\mathbf{1}_{A}$ and $y=M\mathbf{1}_{B}$; then $%
x\vee y=M\mathbf{1}_{A\cup B}$ and $x\wedge y=M\mathbf{1}_{A\cap B}$, which
yields%
\begin{eqnarray*}
b(x) &=&M\left\vert A\right\vert -p(A), \\
b(y) &=&M\left\vert B\right\vert -p(B), \\
b(x\vee y) &\geq &M\left\vert A\cup B\right\vert -p(A\cup B), \\
b(x\wedge y) &\geq &M\left\vert A\cap B\right\vert -p(A\cap B)
\end{eqnarray*}%
(the last two are inequalities because $A\cup B$ and $A\cap B$ need not be
in $Q$). Since $\left\vert A\right\vert +\left\vert B\right\vert =\left\vert
A\cup B\right\vert +\left\vert A\cap B\right\vert $, we get%
\begin{equation*}
b(x)+b(y)-b(x\vee y)-b(x\wedge y)\leq p(A\cup B)+p(A\cap B)-p(A)-p(B).
\end{equation*}%
If $b$ is supermodular then the left-hand side is $\geq 0$, and so the
right-hand side is $\geq 0$, which yields the submodularity inequality (\ref%
{eq:subm-sets}) for all $A,B\in Q$; this extends to all $A,B\subseteq K$ by
Proposition \ref{p:subm-p0} in Appendix \ref{sus-a:submod p}.
\end{proof}

\bigskip

The third proposition shows the equivalence between the separable
superadditivity of $b$ and the separable subadditivity of its Fenchel
conjugate $p$ (unlike super/submodularity, where only one direction holds in
general, and there is equivalence only in the deterministic case; see
Proposition \ref{p:p-b-mod}).

\begin{proposition}
\label{p:p-b-add}Let $\mu $ be a mechanism, with buyer payoff function $b$
and canonical pricing function $p$. Then $b$ is separably superadditive if
and only if $p$ is separably subadditive.
\end{proposition}

\begin{proof}
(i) Assume that $b$ is separably superadditive. Take a partition $%
K=K^{\prime }\cup K^{\prime \prime }$. For every $g\in \lbrack 0,1]^{k}$ and 
$x\in \mathbb{R}_{+}^{k}$ we then have $g\cdot x=g^{\prime }\cdot x^{\prime
}+g^{\prime \prime }\cdot x^{\prime \prime }$ (where $x=x^{\prime
}+x^{\prime \prime }$ and $g=g^{\prime }+g^{\prime \prime }$ are the
corresponding decompositions, as in Section \ref{sus:s-modularity}), and $%
b(x)\geq b(x^{\prime })+b(x^{\prime \prime })$ (by the separable
superadditivity of $b$), which yields%
\begin{eqnarray*}
g\cdot x-b(x) &\leq &\left( g^{\prime }\cdot x^{\prime }+g^{\prime \prime
}\cdot x^{\prime \prime }\right) -\left( b(x^{\prime })+b(x^{\prime \prime
})\right) \\
&=&g^{\prime }\cdot x^{\prime }-b(x^{\prime })+g^{\prime \prime }\cdot
x^{\prime \prime }-b(x^{\prime \prime })\leq p(g^{\prime })+p(g^{\prime
\prime }).
\end{eqnarray*}%
Taking the supremum over $x$ gives $p(g)\leq p(g^{\prime })+p(g^{\prime
\prime }).$

(ii) Assume that $p$ is separably subadditive. Let $x\in \mathbb{R}_{+}^{k}$
and take a partition $K=K^{\prime }\cup K^{\prime \prime }$, resulting in
vectors $x^{\prime }$ and $x^{\prime \prime }$ as above. Let $g^{\prime
}:=q(x^{\prime })\wedge \mathbf{1}_{K^{\prime }}$ and $g^{\prime \prime
}:=q(x^{\prime \prime })\wedge \mathbf{1}_{K^{\prime \prime }}$. Then $%
g^{\prime }\cdot x^{\prime }=q(x^{\prime })\cdot x^{\prime }$, which implies
that $p(g^{\prime })\geq p(q(x^{\prime }))$ by IC at $x^{\prime }$; but $%
g^{\prime }\leq q(x^{\prime })$ and $p$ is nondecreasing, and so we have
equality, i.e., $p(g^{\prime })=p(q(x^{\prime }))$, which yields $%
b(x^{\prime })=g^{\prime }\cdot x^{\prime }-p(g^{\prime })$. Similarly, $%
b(x^{\prime \prime })=g^{\prime \prime }\cdot x^{\prime \prime }-p(g^{\prime
\prime })$. Put $g:=g^{\prime }+g^{\prime \prime }$' then $g\in \lbrack
0,1]^{k}$ (because $g^{\prime }$ and $g^{\prime \prime }$ are orthogonal)
and $g\cdot x=g^{\prime }\cdot x^{\prime }+g^{\prime \prime }\cdot x^{\prime
\prime }$, and so we get%
\begin{eqnarray*}
b(x) &\geq &g\cdot x-p(g)\geq \left( g^{\prime }\cdot x^{\prime }+g^{\prime
\prime }\cdot x^{\prime \prime }\right) -\left( p(g^{\prime })+p(g^{\prime
\prime })\right) \\
&=&(g^{\prime }\cdot x^{\prime }-p(g^{\prime }))+(g^{\prime \prime }\cdot
x^{\prime \prime }-p(g^{\prime \prime }))=b(x^{\prime })+b(x^{\prime \prime
}),
\end{eqnarray*}%
where the second inequality is by the separable subadditivity of $p.$
\end{proof}

\bigskip

The three Propositions \ref{p:AM2supM}, \ref{p:p-b-mod}, and \ref{p:p-b-add}
prove Theorem \ref{th:am} (see the implications in the above tables).

\subsection{Revenue of Allocation-Monotonic Mechanisms\label{sus:a-mon rev}}

In this section we prove that allocation-monotonic mechanisms yield a
multiple of at most $O(\log k)$ of the separate revenue (Theorem D); this
follows from the separable subadditivity property of allocation-monotonic
mechanisms (Theorem \ref{th:am} (iii)), by applying a result of Chawla,
Teng, and Tzamos (2022).

\begin{theorema}
\label{th:amon-srev}Let $X$ be a $k$-good random valuation. Then%
\begin{equation*}
\text{\textsc{AMonRev}}(X)\leq 2\ln (2k)\cdot \text{\textsc{SRev}}(X).
\end{equation*}
\end{theorema}

\begin{proof}
Let $\mu $ be a tie-favorable allocation-monotonic mechanism (recall
Corollary \ref{c:am-fav}), with buyer payoff function $b$ and canonical
pricing function $p$. Define a new pricing function $p^{\prime
}:[0,1]^{k}\rightarrow \lbrack 0,\mathbb{\infty ]}$ by%
\begin{equation*}
p^{\prime }(g)%
%TCIMACRO{\TeXButton{:=}{{\;:=\;}}}%
%BeginExpansion
{\;:=\;}%
%EndExpansion
\sum_{i=1}^{k}p(g_{i}e^{i});
\end{equation*}%
thus $p^{\prime }$ is a separably additive function, and thus it yields a
separable mechanism (namely, take for each good $i$ the one-good mechanism
with pricing function $p_{i}(g_{i})%
%TCIMACRO{\TeXButton{:=}{{\;:=\;}}}%
%BeginExpansion
{\;:=\;}%
%EndExpansion
p(g_{i}e^{i})$).

The function $p$ is separably subadditive by Theorem \ref{th:am} (iii), and
so%
\begin{equation}
p(g)\leq p^{\prime }(g)  \label{eq:p<=p'}
\end{equation}%
for every $g$ (see Section \ref{sus:s-modularity}). Because $p$ is
nondecreasing we have $p(g)\geq p(g_{i}e^{i})$ for all $i$, and so $p(g)$ is
at least as large as their average, i.e., 
\begin{equation}
p(g)\geq \frac{1}{k}p^{\prime }(g).  \label{eq:p>=(1/k)p'}
\end{equation}

Since $p$ is a nondecreasing and closed function, so is the derived function 
$p^{\prime }$. We can therefore apply the result of Chawla, Teng, and Tzamos
(2022) (see Theorem \ref{th:chawla} in Appendix \ref{sus-a:chawla}): \textsc{%
AMonRev}$(X)$, which is the revenue obtainable from any pricing function $p$
as above, is, by (\ref{eq:p<=p'}) and (\ref{eq:p>=(1/k)p'}), at most $2\ln
(2k)$ times the revenue obtainable from any derived pricing function $%
p^{\prime }$; the pricing functions $p^{\prime }$ are separable, and so
their revenue is at most \textsc{SRev}$(X).$
\end{proof}

\section{Quadratic Mechanisms\label{s:quadratic}}

In this section we introduce a useful class of mechanisms that have a simple
representation, and are thus amenable to an easier analysis: mechanisms
where the relevant functions---specifically, the payment function $s$, the
buyer payoff function $b$, and the pricing function $p$---are all convex
quadratic functions (in appropriate domains). In particular, this will allow
the construction in Section \ref{sus:a-mon not submod} of an
allocation-monotonic mechanism that does not have a submodular pricing
function.

In the single good case, let $a>0$ be a parameter, and consider a mechanism $%
\mu $ with $q(x)=ax$ and $s(x)=%
%TCIMACRO{\U{bd}}%
%BeginExpansion
{\frac12}%
%EndExpansion
ax^{2}$, and thus $b(x)=ax\cdot x-%
%TCIMACRO{\U{bd}}%
%BeginExpansion
{\frac12}%
%EndExpansion
ax^{2}=%
%TCIMACRO{\U{bd}}%
%BeginExpansion
{\frac12}%
%EndExpansion
ax^{2}$, for $x$ in a set $V\subseteq \lbrack 0,a^{-1}]$ (so that $q(x)\leq
1 $; the IC conditions hold because $b^{\prime }(x)=ax=q(x)$). The
corresponding pricing function is then $p(g)=s(a^{-1}g)=%
%TCIMACRO{\U{bd}}%
%BeginExpansion
{\frac12}%
%EndExpansion
a^{-1}g^{2}$ for $g=ax$ (with $x\in V$), and so the functions $s,b$, and $p$
are convex quadratics (in certain domains).

In the $k$-good case, let $A$ be a positive definite (and thus symmetric and
invertible) $k\times k$ matrix, and let $V\subseteq \{x\in \mathbb{R}%
_{+}^{k}:Ax\in \lbrack 0,1]^{k}\}$ be a domain of valuations (in this
section all vectors should be understood as column vectors, even when
written as, say, $x=(x_{1},...,x_{k})$; the transpose of $x$, which is a row
vector, is denoted by $x^{\top }$). For convenience we assume that $V$ is a
closed convex set with a nonempty interior that contains $\mathbf{0}$. We
define:

\begin{itemize}
\item A mechanism $\mu =(q,s)$ is \emph{quadratic} with matrix $A$ and
domain $V$ as above if%
\begin{equation}
q(x)=Ax\text{\ \ and\ \ }s(x)=\frac{1}{2}x^{\top }Ax  \label{eq:quadratic}
\end{equation}%
for every $x$ in $V$
\end{itemize}

\noindent (thus, there are no restrictions outside $V$). For every $x$ in $V$
we then have $b(x)=q(x)\cdot x-s(x)=Ax\cdot x-%
%TCIMACRO{\U{bd}}%
%BeginExpansion
{\frac12}%
%EndExpansion
x^{\top }Ax=%
%TCIMACRO{\U{bd}}%
%BeginExpansion
{\frac12}%
%EndExpansion
x^{\top }Ax$, i.e.,%
\begin{equation}
b(x)=\frac{1}{2}x^{\top }Ax,  \label{eq:quad-b}
\end{equation}%
and so the function $b$ is a convex quadratic function on $V$, and $q(x)=Ax$
is a subgradient of $b$ at $x\in V$ (for $x$ in the interior of $V$, this is
the gradient); see Proposition \ref{p:b function}.

Do quadratic mechanisms exist? Yes: for every $A$ and $V$ let $\mu $ be the
IC mechanism with menu $\{(Ax,%
%TCIMACRO{\U{bd}}%
%BeginExpansion
{\frac12}%
%EndExpansion
x^{\top }Ax):x\in V\}$. Then for each $x\in V$ the choice $(Ax,%
%TCIMACRO{\U{bd}}%
%BeginExpansion
{\frac12}%
%EndExpansion
x^{\top }Ax)$ is optimal at $x$, because $Ax\cdot x-%
%TCIMACRO{\U{bd}}%
%BeginExpansion
{\frac12}%
%EndExpansion
x^{\top }Ax\geq Ay\cdot x-%
%TCIMACRO{\U{bd}}%
%BeginExpansion
{\frac12}%
%EndExpansion
y^{\top }Ay$ for all $y\in V$ (this inequality, which is equivalent to $%
f(y)\geq f(x)+(y-x)\cdot \nabla f(x)$ for the function $f(x)=%
%TCIMACRO{\U{bd}}%
%BeginExpansion
{\frac12}%
%EndExpansion
x^{\top }Ax$, holds for all $x,y\in \mathbb{R}^{k}$ by the convexity of $f$%
). Therefore $\mu $ satisfies (\ref{eq:quadratic}) (and $\mathbf{0}\in V$
gives $b(\mathbf{0})=s(\mathbf{0})=0$, i.e., IR and NPT).

Let $Q_{V}\subseteq Q$ be the range of allocations for valuations in $V$,
i.e., $Q_{V}:=\{q(x):x\in V\}\subseteq \lbrack 0,1]^{k}$. Each $g\in Q_{V}$
is obtained as $q(x)$ for a unique $x\in V$, namely, $x=A^{-1}g$, and so $%
Q_{V}=A^{-1}(V)$ is a closed convex set with a nonempty interior that
contains $\mathbf{0}$ (because $V$ is such a set, and $A^{-1}$ is regular).
The restriction to $Q_{V}$ of any pricing function $p$ of $\mu $ then
satisfies%
\begin{equation}
p(g)=s(A^{-1}g)=\frac{1}{2}(A^{-1}g)^{\top }A(A^{-1}g)=\frac{1}{2}g^{\top
}A^{-1}g  \label{eq:pq}
\end{equation}%
for every $g\in Q_{V}$, and so $p$ is a convex quadratic function on $Q_{V}$.

\subsection{Allocation Monotonicity without Submodularity\label{sus:a-mon
not submod}}

For quadratic mechanisms there are simple necessary conditions for
allocation monotonicity and for submodularity.

\begin{proposition}
\label{p:quad}Let $\mu $ be a quadratic mechanism with matrix $A$ and domain 
$V$. Then:

\begin{description}
\item[(i)] If $\mu $ is allocation monotonic then the off-diagonal entries
of $A$ are all nonnegative: $(A)_{ij}\geq 0$ for all $i\neq j.$

\item[(ii)] If $\mu $ is submodular then the off-diagonal entries of $A^{-1}$
are all nonpositive: $(A^{-1})_{ij}\leq 0$ for all $i\neq j.$
\end{description}
\end{proposition}

For the diagonal entries we have $(A)_{ii}\geq 0$ and $(A^{-1})_{ii}\geq 0$
for all $i$, by the convexity of the functions $b$ and $p$ (in fact, $%
(A)_{ii}>0$ and $(A^{-1})_{ii}>0$ by the positive definiteness of $A$ and $%
A^{-1}$).

\begin{proof}
(i) The function $q(x)=Ax$ is nondecreasing on the nonempty open set $%
\mathrm{int\,}V$ if and only if $(A)_{ij}\geq 0$ for all $i,j$
(alternatively: $b(x)=%
%TCIMACRO{\U{bd}}%
%BeginExpansion
{\frac12}%
%EndExpansion
x^{\top }Ax$ is supermodular if and only if $(A)_{ij}\geq 0$ for all $i\neq
j $; see Proposition \ref{p:AM2supM}).

(ii) The function $p(g)=%
%TCIMACRO{\U{bd}}%
%BeginExpansion
{\frac12}%
%EndExpansion
g^{\top }A^{-1}g$ is submodular on the nonempty open set $\mathrm{int\,}%
Q_{V} $ if and only if $(A^{-1})_{ij}\leq 0$ for all $i\neq j$; see (\ref%
{eq:d2p<=0}).
\end{proof}

\bigskip

For positive definite matrices $A$ and $A^{-1}$, the condition on $A^{-1}$
in (ii) of nonpositive nondiagonal entries \emph{implies} the condition on $%
A $ in (i) of nonnegative nondiagonal entries (see Plemmons 1977, Theorem 1,
C9 $\Longleftrightarrow $ F15, applied to $A^{-1}$, and Proposition \ref%
{p:p-b-mod} (i)). While the converse is easily shown to be true when $k=2$,
it is not true in general, when $k\geq 3$. This allows the construction of
an example of a quadratic mechanism that is allocation monotonic but \emph{%
not} submodular.

\begin{example}
\label{ex:AM not subm}Let $k=3$ and take%
\begin{equation*}
A=%
\begin{bmatrix}
6 & 3 & 1 \\ 
3 & 6 & 3 \\ 
1 & 3 & 6%
\end{bmatrix}%
;
\end{equation*}%
then $A$ is a positive definite matrix, and its inverse is the positive
definite matrix 
\begin{equation*}
A^{-1}=\frac{1}{120}%
\begin{bmatrix}
27 & -15 & 3 \\ 
-15 & 35 & -15 \\ 
3 & -15 & 27%
\end{bmatrix}%
.
\end{equation*}%
We will construct a quadratic mechanism $\mu $ with matrix $A$ and domain $V$
that is allocation monotonic (on the whole space $\mathbb{R}_{+}^{k}$, not
just on\footnote{%
The standard extension that uses the menu $\{(Ax,%
%TCIMACRO{\U{bd}}%
%BeginExpansion
{\frac12}%
%EndExpansion
x^{\top }Ax):x\in V\}$, as in the previous section, is not allocation
monotonic, and so we construct a different extension.} $V)$. However, $\mu $
is \emph{not} submodular, since $A^{-1}$ has positive off-diagonal entries,
namely, $(A^{-1})_{13}=(A^{-1})_{31}=3>0$, and so, by Proposition \ref%
{p:quad}, the pricing function $p$ is not submodular (already on $Q_{V}$).

To construct the mechanism $\mu $, take $v\gg 0$ such that $Av\leq e$ (e.g., 
$v=(1/15)e$), and put $V=\{x\in \mathbb{R}^{3}:0\leq x\leq v\}$; then $0\leq
Ax\leq Av\leq e$ for every $x\in V$ (because $A\geq 0$). For every $x\in 
\mathbb{R}_{+}^{3}$ denote $\tilde{x}:=x\wedge v$, and let%
\begin{equation*}
b(x)=%
%TCIMACRO{\U{bd}}%
%BeginExpansion
{\frac12}%
%EndExpansion
\tilde{x}^{\top }A\tilde{x}+e\cdot (x-\tilde{x})=%
%TCIMACRO{\U{bd}}%
%BeginExpansion
{\frac12}%
%EndExpansion
\tilde{x}^{\top }A\tilde{x}+\sum_{i=1}^{k}[x_{i}-v_{i}]_{+}
\end{equation*}%
(thus, $b(x)=%
%TCIMACRO{\U{bd}}%
%BeginExpansion
{\frac12}%
%EndExpansion
x^{\top }Ax$ for every $x\in V$). We claim that $b$ is a continuous,
nondecreasing, nonexpansive, convex, and supermodular function on $\mathbb{R}%
_{+}^{3}$. Indeed, the derivative of $b$ in direction $e^{i}$ (the $i$-th
unit vector) is $b^{\prime }(x;e^{i})=(A\tilde{x})_{i}$ if $x_{i}<v_{i}$,
and $b^{\prime }(x;e^{i})=1$ if $x_{i}\geq v_{i}$. Since $(A\tilde{x})_{i}$
is nondecreasing in $x$ (because the mapping $x\rightarrow \tilde{x}$ is
nondecreasing, and $A\geq 0)$ and $(A\tilde{x})_{i}\leq (Av)_{i}\leq 1$, the
function $b^{\prime }(x;e^{i})$ is nondecreasing in $x$. Therefore $b$ is an
ultramodular function (by Theorem 5.5 (i) in Marinacci and Montrucchio 2005);%
\footnote{%
Which yields ultramodularity of $b$ on the interior of $\mathbb{R}_{+}^{k};$
since $b$ is continuous, it extends to $\mathbb{R}_{+}^{k}$ (see also
Proposition \ref{p:delta+} in Appendix \ref{sus-a:utramodular}).} since $%
0\leq b^{\prime }(x;e^{i})\leq 1$, it is nondecreasing and nonexpansive.
Finally, to see that $b$ is convex, represent it as%
\begin{equation}
b(x)=\max_{I\subseteq K}\left\{ b_{0}(x_{K\backslash I},v_{I})+\sum_{i\in
I}(x_{i}-v_{i})\right\}   \label{eq:b0}
\end{equation}%
where $b_{0}(x):=%
%TCIMACRO{\U{bd}}%
%BeginExpansion
{\frac12}%
%EndExpansion
x^{\top }Ax.$ Indeed: (i) the nonexpansiveness of $b_{0}$ yields $%
b_{0}(y_{-i},x_{i})\leq b_{0}(y_{-i},v_{i})+x_{i}-v_{i}$ when $x_{i}\geq
v_{i},$ and $b_{0}(y_{-i},v_{i})+x_{i}-v_{i}\leq b_{0}(y_{-i},x_{i})$ when $%
x_{i}\leq v_{i}$, and so the maximum in (\ref{eq:b0}) is attained at $%
I=\{i:x_{i}\geq v_{i}\},$ i.e., when $(x_{K\backslash I},v_{I})=x\wedge v=%
\tilde{x}$ and $b(x)=b_{0}(\tilde{x})+e\cdot (x-\tilde{x});$ (ii) for each $I
$ the corresponding function in (\ref{eq:b0}) is convex, because $b_{0}$ is
convex, and so their maximum is convex as well.

Therefore $b$ is a buyer payoff function (by Proposition \ref{p:b function})
that is supermodular. By Proposition \ref{p:AM2supM}, the corresponding
buyer-favorable (and tie-favorable) mechanism $\mu =(q,s)$, whose allocation
function is $q_{i}(x)=b^{\prime }(x;e^{i})$ (i.e., $q_{i}(x)=(A\tilde{x}%
)_{i} $ if $x_{i}<v_{i}$ and $q_{i}(x)=1$ if $x_{i}\geq v_{i}$) for $%
i=1,...,k$, is allocation monotonic.
\end{example}

\section{Revenue of Symmetric Deterministic Mechanisms\label{s:sym det}}

In this section we study the class of symmetric deterministic mechanisms,
which are monotonic (by Theorem 4 and Proposition 5 in Hart and Reny 2015),
but in general not allocation monotonic (see Example \ref{ex:mon-not-am} in
the Introduction). Perhaps surprisingly, here we start with the \emph{%
supermodular} case, where we show that the revenue is at most an $O(\log k)$
multiple of the separate revenue (Theorem \ref{th:sd-super}), from which we
obtain an $O(\log ^{2}k)$ multiple in the general case (Theorem \ref%
{th:symdrev}).

Let $\mu $ be a symmetric deterministic mechanism. Its canonical
deterministic pricing function $p$ is then also symmetric: the price of a
set $A\subseteq K$ of goods depends only on the size $|A|$ of $A$. We thus
write $p(\left\vert A\right\vert )$ instead of $p(A)$, where $%
p:\{0,1,...,k\}\rightarrow \lbrack 0,\infty ]$ is a nondecreasing function
with $p(0)=0.$

\subsection{The Supermodular Case\label{sus:sym det supermod}}

The symmetric deterministic mechanism $\mu $ is supermodular if $%
p(|A|)+p(|B|)\leq p(|A\cup B|)+p(|A\cap B|)$ for every $A,B\subseteq K$.
This is easily seen to be equivalent to the sequence of price differences $%
d(m):=p(m)-p(m-1)$ being nondecreasing, i.e., 
\begin{equation*}
d(m)\geq d(m-1)
\end{equation*}%
for all\footnote{%
Such functions are the restriction to $\{0,1,...,k\}$ of convex functions on
the interval $[0,k].$} $1\leq m\leq k$ (when $p(m)=\infty $ put $d(m)=\infty 
$). Let \textsc{SupermodSymDRev} denote the maximal revenue achievable by
supermodular symmetric deterministic mechanisms.

For every $k\geq 1$ let%
\begin{equation*}
H(k)%
%TCIMACRO{\TeXButton{:=}{{\;:=\;}}}%
%BeginExpansion
{\;:=\;}%
%EndExpansion
1+\frac{1}{2}+...+\frac{1}{k}
\end{equation*}%
denote the harmonic sum up to $k$; thus, $\ln k\leq H(k)\leq \ln k+1$, and $%
H(k)-\ln k\rightarrow \gamma \equiv 0.577..$. as $k\rightarrow \infty .$

\begin{theorem}
\label{th:sd-super}Let $X$ be a $k$-good random valuation. Then%
\begin{equation*}
\text{\textsc{SupermodSymDRev}}(X)\leq H(k)\cdot \text{\textsc{SymSRev}}%
(X)\leq O(\log k)\cdot \text{\textsc{SRev}}(X).
\end{equation*}
\end{theorem}

\begin{proof}
Let $\mu =(q,s)$ be a symmetric deterministic mechanism, whose canonical
deterministic pricing function $p:\{0,1,...,k\}\rightarrow \lbrack 0,\infty
] $ is supermodular. Let $k_{0}\leq k$ be the maximal size for which the
price is finite, i.e., $p(m)<\infty $ if and only if $m\leq k_{0}$. We will
show that%
\begin{equation}
R(\mu ;X)\leq H(k_{0})\cdot \,\text{\textsc{SymSRev}}(X),
\label{eq:harmonic}
\end{equation}%
which yields the result.

Given a random valuation $X$, for each subset size $0\leq m\leq k$ let $%
\beta _{m}:=\mathbb{P}\left[ |q(X)|=m\right] $ be the probability that $\mu $
allocates exactly $m$ goods; then $\beta _{m}=0$ for $m>k_{0}$ (because the
price there is infinite), and 
\begin{equation}
R(\mu ;X)=\sum_{m=1}^{k_{0}}\beta _{m}p(m)  \label{eq:R}
\end{equation}%
(the sum starts at $1$ since $p(0)=0$).

We claim that for every $n=1,...,k_{0}$,%
\begin{equation}
\sum_{i=1}^{k_{0}}\mathbb{P}\left[ X_{i}\geq d(n)\right] \geq
\sum_{m=n}^{k_{0}}m\beta _{m}.  \label{eq:pn-pn-1}
\end{equation}%
Indeed, if $q(x)$ is a set of size at least $n$, say $q(x)=A$ with $%
|A|=m\geq n$, then for each one of the $m$ elements $i$ of $A$ we have $%
b(x)=x(A)-p(m)\geq x(A\backslash \{i\})-p(m-1)$ (by (\ref{eq:b=p0*})), and
so $x_{i}\geq p(m)-p(m-1)=d(m)\geq d(n)$ (by supermodularity). Therefore $A$
contributes to the left-hand sum at least $m$ times $\mathbb{P}\left[ q(X)=A%
\right] $. Summing over all $A$ with $|A|=m\geq n$ (the events $\{q(X)=A\}$
are disjoint for different sets $A$) yields the inequality (\ref{eq:pn-pn-1}%
).

Therefore%
\begin{equation*}
\text{\textsc{SymSRev}}(X)\geq d(n)\sum_{i=1}^{k_{0}}\mathbb{P}\left[
X_{i}\geq d(n)\right] \geq d(n)\sum_{m=n}^{k_{0}}m\beta _{m}\geq
d(n)n\sum_{m=n}^{k_{0}}\beta _{m}
\end{equation*}%
(the second inequality is (\ref{eq:pn-pn-1})), and so%
\begin{equation*}
d(n)\sum_{m=n}^{k_{0}}\beta _{m}\leq \frac{1}{n}\text{\textsc{SymSRev}}(X).
\end{equation*}%
Summing over $n=1,...,k_{0}$ yields (\ref{eq:harmonic}), since%
\begin{equation*}
\sum_{n=1}^{k_{0}}d(n)\sum_{m=n}^{k_{0}}\beta _{m}=\sum_{m=1}^{k_{0}}\beta
_{m}\sum_{n=1}^{m}d(n)=\sum_{m=1}^{k_{0}}\beta _{m}p(m),
\end{equation*}%
which is $R(\mu ;X)$ by (\ref{eq:R}).
\end{proof}

\bigskip

\noindent \textbf{Remarks. }\emph{(a)} The bound in (\ref{eq:harmonic}) is
tight for each $k\geq 1$: see Example \ref{ex:harmonic} below.

\emph{(b)} The bound in (\ref{eq:harmonic}) need not hold when $p$ is not
supermodular: see Example \ref{ex:non-convex p} below.

\bigskip

The following example shows the tightness of (\ref{eq:harmonic}); in
addition, it provides a gap of $\ln k$ between \textsc{AMonRev} and \textsc{%
MonRev}.

\begin{example}
\label{ex:harmonic}Let $M>1$ be large, and define for each $m\geq 1$ 
\begin{eqnarray*}
z_{m} &%
%TCIMACRO{\TeXButton{:=}{{\;:=\;}}}%
%BeginExpansion
{\;:=\;}%
%EndExpansion
&(\underbrace{M^{m},...,M^{m}}_{m},\underbrace{0,...,0}_{k-m})\in \mathbb{R}%
_{+}^{k}\text{\ \ \ and} \\
\beta _{m} &%
%TCIMACRO{\TeXButton{:=}{{\;:=\;}}}%
%BeginExpansion
{\;:=\;}%
%EndExpansion
&\frac{1}{mM^{m}},
\end{eqnarray*}%
and put $z_{0}:=\mathbf{0}$ and $\beta _{0}:=1-\sum_{m=1}^{k}\beta _{m}\geq
0 $. Let the random valuation $X$ take as values the $\binom{k}{m}$
permutations of (the coordinates of) $z_{m}$ with probability $\beta _{m}/%
\binom{k}{m}$ each, for all $m=0,...,k$; thus, for every $A\subseteq K$ we
have $\mathbb{P}\left[ X=M^{m}\mathbf{1}_{A}\right] =\beta _{m}/\binom{k}{m}$%
, where $m:=\left\vert A\right\vert .$

We claim that, as $M\rightarrow \infty ,$

\begin{enumerate}
\item \textsc{SupermodSymDRev}$(X)$ and \textsc{Rev}$(X)$ converge to $H(k)$%
, and thus so do \textsc{SymDRev}$(X),~$\textsc{DRev}$(X)$, and \textsc{%
MonRev}$(X)$ (which all lie between the above two revenues).

\item \textsc{SRev}$(X)$\ (which is the same as \textsc{SymSRev}$(X)$) and\ 
\textsc{AMonRev}$(X)$ converge to\footnote{\textsc{BRev}$(X)$ converges to $%
1 $ as well, because $mM^{m}\cdot \mathbb{P}\left[ \sum_{i}X_{i}\geq mM^{m}%
\right] \rightarrow 1$ for all $m\geq 1$ (cf. the proof of (ii)).} $1.$
\end{enumerate}

\noindent The proof is divided into four parts, (i)--(iv).

\begin{description}
\item[(i)] \ \ \ \textsc{SupermodSymDRev}$(X)\geq H(k).$
\end{description}

\noindent \emph{Proof.} Take the symmetric deterministic mechanism $\mu
=(q,s)$ with supermodular pricing function $p(m)=M^{m}$ for all $m=1,...,k$
(and $p(0)=0$); we will show that $R(\mu ;X)=H(k)$. Indeed, for each
permutation $z_{m}^{\prime }$ of $z_{m}$ we have $q(z_{m}^{\prime })=\{i\in
K:z_{m,i}^{\prime }=M^{m}\}$ (e.g., $q(z_{m})=\{1,...,m\}$) and $%
s(z_{m}^{\prime })=p(m)$. To see why this is so, consider $z_{m}$: the IC
and IR conditions hold at $z_{m}$ because $q(z_{m})\cdot
z_{m}-s(z_{m})=mM^{m}-p(m)=(m-1)M^{m}$, whereas for each $z_{n}^{\prime }$
(some permutation of $z_{n}$) that is different from $z_{m}$, if $n\leq m$
then $q(z_{n}^{\prime })\cdot z_{m}-s(z_{n}^{\prime })\leq
nM^{m}-M^{n}<(m-1)M^{m}$ (because the set $q(z_{n}^{\prime })$ contains at
most $n$ elements from the set $\{1,...,m\}$; when $n=m$ use $z_{n}^{\prime
}\neq z_{m}$), and if $n>m$ then $q(z_{n}^{\prime })\cdot
z_{m}-s(z_{n}^{\prime })\leq mM^{m}-M^{n}<(m-1)M^{m}$. The same holds for
any permutation $z_{m}^{\prime }$ of $z_{m}$. Therefore $R(\mu
;X)=\sum_{m=1}^{k}\beta _{m}p(m)=\sum_{m=1}^{k}1/m=H(k).\;\square $

\begin{description}
\item[(ii)] $\;\;\;$\textsc{Rev}$(X)\rightarrow H(k)$ as $M\rightarrow
\infty .$
\end{description}

\noindent \emph{Proof.} Let $\mu =(q,s)$ be a mechanism. Since $X$ is
symmetric we can assume without loss of generality that $\mu $ is symmetric,
in the sense that if the coordinates of $x^{\prime }$ are a permutation of
the coordinates of $x$ then $q(x^{\prime })$ is the corresponding
permutation $q(x)^{\prime }$ of the coordinates of $q(x)$, and $s(x^{\prime
})=s(x)$ (indeed, the average of all \textquotedblleft coordinate
permutations" of $\mu $ yields the same revenue from $X$ as the original $%
\mu )$. Therefore the first $m$ coordinates of $q(z_{m})$ are all equal,
say, $q_{i}(z_{m})=\gamma _{m}$ for some $0\leq \gamma _{m}\leq 1$, and $%
s(z_{m}^{\prime })=s(z_{m})$ for every permutation $z_{m}^{\prime }$ of $%
z_{m}$, which yields%
\begin{equation*}
R(\mu ;X)=\sum_{m=1}^{k}\frac{1}{mM^{m}}s(z_{m}).
\end{equation*}

For every $m\geq 1$, by IC at $z_{m}$ vs. $z_{m-1}$, we have%
\begin{equation*}
m\gamma _{m}M^{m}-s(z_{m})\geq (m-1)\gamma _{m-1}M^{m}-s(z_{m-1}).
\end{equation*}%
By IR at $z_{m-1}$ we have $s(z_{m-1})\leq (m-1)M^{m-1}\leq mM^{m-1}$, and
so dividing by $mM^{m}$ yields%
\begin{equation}
\frac{1}{mM^{m}}s(z_{m})\leq \gamma _{m}-\frac{m-1}{m}\gamma _{m-1}+\frac{1}{%
M}.  \label{eq:pm}
\end{equation}%
Summing (\ref{eq:pm}) over $m=1,...,k$ gives%
\begin{equation*}
R(\mu ;X)\leq \sum_{m=1}^{k}\frac{1}{m}\gamma _{m-1}+\frac{k}{M}\leq
\sum_{m=1}^{k}\frac{1}{m}+\frac{k}{M}\rightarrow H(k),
\end{equation*}%
where the second inequality is by $\gamma _{m}\leq 1$ for all $m$. Part (i)
completes the proof.$\;\square $

\begin{description}
\item[(iii)] \ \ \ \textsc{SRev}$(X)=\,\text{\textsc{SymSRev}}(X)\rightarrow
1$ as $M\rightarrow \infty .$
\end{description}

\noindent \emph{Proof.} $X_{i}$ takes the value $M^{m}$ in the fraction $m/k$
of the permutations of $z_{m}$, and so $\mathbb{P}\left[ X_{i}=M^{m}\right]
=(m/k)\beta _{m}=1/(kM^{m})$, which yields%
\begin{equation*}
M^{m}\cdot \mathbb{P}\left[ X_{i}\geq M^{m}\right] =M^{m}\cdot \sum_{n\geq m}%
\frac{1}{kM^{n}}\rightarrow \frac{1}{k}
\end{equation*}%
as $M\rightarrow \infty $. Therefore \textsc{Rev}$(X_{i})\rightarrow 1/k$
for each $i$, and so \textsc{SymSRev}$(X)=\,$\textsc{SRev}$(X)\rightarrow
1.\;\square $

\begin{description}
\item[(iv)] \ \ \ \textsc{AMonRev}$(X)\rightarrow 1$ as $M\rightarrow \infty
.$
\end{description}

\noindent \emph{Proof.} Let $\mu =(q,s)$ be an allocation-monotonic
mechanism; as in the the proof of (ii), we assume without loss of generality
that $\mu $ is symmetric, because the symmetrization preserves allocation
monotonicity (an average of allocation-monotonic mechanisms is clearly
allocation monotonic). Define $\gamma _{m}$ as in (ii).

For each $m\geq 1$, let $y_{m}$ have its first $m$ coordinates equal to $%
M^{m-1}$ and the rest equal to $0$; then $y_{m}\geq z_{m-1}$ and $y_{m}\geq
z_{m-1}^{\#}$, where $z_{m}^{\#}$ is the permutation of $z_{m-1}$ for which
coordinates $2,...,m$ are equal to $M^{m-1}$; by allocation monotonicity we
thus have $q(y_{m})\geq q(z_{m-1})$ and $q(y_{m})\geq
q(z_{m-1}^{\#})=q(z_{m-1})^{\#}$, and so each one of the first $m$
coordinates of $q(y_{m})$ is $\geq \gamma _{m-1}$. By IR at $y_{m}$ we have $%
s(y_{m})\leq q(y_{m})\cdot y_{m}\leq mM^{m-1}$, and then by IC at $z_{m}$
vs. $y_{m}$ we get%
\begin{equation*}
m\gamma _{m}M^{m}-s(z_{m})\geq q(y_{m})\cdot z_{m}-s(y_{m})\geq m\gamma
_{m-1}M^{m}-mM^{m-1}.
\end{equation*}%
Dividing by $mM^{m}$ yields%
\begin{equation*}
\frac{1}{mM^{m}}p_{m}\leq \gamma _{m}-\gamma _{m-1}+\frac{1}{M},
\end{equation*}%
and then summing over $m=1,...,k$ gives%
\begin{equation*}
R(\mu ;X)\leq \gamma _{k}+\frac{k}{M}\leq 1+\frac{k}{M}\rightarrow 1.
\end{equation*}%
Part (iii) completes the proof, since selling separately is an
allocation-monotonic mechanism.$\;\square $
\end{example}

\bigskip

We thus have a random valuation where monotonic mechanisms yield a revenue
that is $\ln k$ times higher than the revenue of allocation-monotonic
mechanisms (use $H(k)>\ln k$ for $k\geq 2$).

\begin{corollary}
\label{c:amon vs mon}For every $k\geq 2$ there exists a $k$-good random
valuation $X$ such that%
\begin{equation*}
\text{\textsc{MonRev}}(X)\geq \ln k\,\cdot \,\text{\textsc{AMonRev}}(X).
\end{equation*}
\end{corollary}

Theorem \ref{th:monrev} implies that \textsc{MonRev}$(X)\leq k\,\cdot \,$%
\textsc{AMonRev}$(X)$ (because selling separately, or bundled, is allocation
monotonic). We do not know what is the correct bound here (between $\ln k$
and $k$).

The next example shows that the bound in (\ref{eq:harmonic}) need not hold
for symmetric deterministic mechanisms that are \emph{not} supermodular.

\begin{example}
\label{ex:non-convex p}Let $k=3$, and take $p(0)=0,\;p(1)=p(2)=1$, and $%
p(3)=M$, where $M>1$ is a large number; then $p$ is \emph{not} supermodular,
because $p(2)-p(1)<p(1)-p(0)$. We will construct a random valuation $X$ such
that, by letting $B_{m}$ be the event that $\left\vert q(X)\right\vert =m$,
and $\beta _{m}:=\mathbb{P}\left[ B_{m}\right] $ its probability, we have $%
\beta _{1}=0$, $\beta _{2}=1/2$, $\beta _{3}=1/(12M)$, and $\beta
_{0}=1-1/2-1/(12M)$. Let the random valuation $X=(X_{1},X_{2},X_{3})$ be
exchangeable (i.e., permuting the coordinates does not change the
distribution of $X)$, as follows: the set $B_{2}$ consists of the valuations 
$(U,1-U,0)$ and their permutations, where $U$ is uniformly chosen from the
interval $(0,1)$; the set $B_{3}$ consists of the single valuation $(M,M,M)$%
; and $B_{0}$ consists of $\mathbf{0}=(0,0,0)$ (the set $B_{1}$ is empty).
The IC and IR conditions are easily checked. Then $R(\mu ;X)=\sum_{m}\beta
_{m}p(m)=1/2\cdot 1+1/(12M)\cdot M=7/12$, whereas for each $i$ we have%
\footnote{%
We write $a\sim b$ to mean $a/b\rightarrow 1$ as $M\rightarrow \infty .$} 
\textsc{Rev}$(X_{i})=\sup_{t}t\cdot \mathbb{P}\left[ X_{i}\geq t\right] \sim
1/12$ (attained at\footnote{%
For $t=M$ we have $M\cdot \beta _{3}=1/12,$ and for $0\leq t\leq 1$ we have $%
t\cdot ((1-t)(1/3)+1/(12M))\sim t(1-t)/3,$ whose maximum, at $t=1/2,$ equals 
$1/12.$} $t\sim 1/2)$, yielding \textsc{SymSRev}$(X)=\,$\textsc{SRev}$%
(X)\sim 1/4$ and thus $R(\mu ;X)/$\textsc{SRev}$(X)\sim
(7/12)/(1/4)=7/3>11/6=1+1/2+1/3=H(3).$
\end{example}

\subsection{The General Case\label{sus:sym det general}}

We can now bound the maximal revenue obtainable by symmetric deterministic
mechanisms, which we denote by \textsc{SymDRev}, relative to the (symmetric)
separate revenue.

\begin{theorema}
\label{th:symdrev}Let $X$ be a $k$-good random valuation. Then%
\begin{equation*}
\text{\textsc{SymDRev}}(X)\leq 2\ln (2k)H(k)\cdot \text{\textsc{SymSRev}}%
(X)\leq O(\log ^{2}k)\cdot \text{\textsc{SRev}}(X).
\end{equation*}
\end{theorema}

\begin{proof}
We will show that any symmetric deterministic $k$-good pricing function can
be bounded within a factor of $k$ by a supermodular one, which yields the
first $\log k$ factor; the second $\log k$ factor comes from Theorem \ref%
{th:sd-super}.

Let $\mu $ be a symmetric deterministic mechanism; let $p$ be its canonical
deterministic pricing function, with $p(0)=0$. Let $d(m):=p(m)-p(m-1)$ be
the price differences; then\footnote{%
Recall that we put $d(m)=\infty $ when $p(m)=\infty .$} $0\leq d(m)\leq p(m)$
(because $p$ is nondecreasing and $p(0)=0)$. Define $d^{\prime
}(m):=\max_{1\leq n\leq m}d(n)$ and $p^{\prime }(m):=\sum_{\ell
=1}^{m}d^{\prime }(\ell )$ for all $m\geq 1$, and $p^{\prime }(0):=p(0)=0$.
We have: $p^{\prime }\geq p$ (because $d^{\prime }\geq d$); $p^{\prime }$ is
nondecreasing (because $d^{\prime }\geq d\geq 0$); $p^{\prime }$ is
supermodular (because $d^{\prime }$ is nondecreasing); and, for each $m\geq
1 $, if $d^{\prime }(m)=d(n)$ for some $1\leq n\leq m$ then $p^{\prime
}(m)\leq md(n)$ (because $d^{\prime }(\ell )\leq d(n)$ for all $1\leq \ell
\leq m)$, and so $p^{\prime }(m)\leq kp(m)$ (because $d(n)\leq p(n)\leq p(m)$
and $m\leq k$). Altogether, we have%
\begin{equation*}
\frac{1}{k}p^{\prime }(m)\leq p(m)\leq p^{\prime }(m)
\end{equation*}%
for all $m$, with $p^{\prime }$ supermodular. By the result of Chawla, Teng,
and Tzamos (2022) (Theorem \ref{th:chawla}), it follows that%
\begin{equation*}
\text{\textsc{SymDRev}}(X)\leq 2\ln (2k)\cdot \text{\textsc{SupermodSymDRev}}%
(X);
\end{equation*}%
together with Theorem \ref{th:sd-super}, this yields the result.
\end{proof}

\section{Monotonicity of Deterministic Mechanisms\label{s:mon det}}

In this section we obtain conditions on a deterministic pricing function to
yield a monotonic mechanism. These conditions are more complex than the
submodularity condition for allocation monotonicity (see Theorem \ref{th:am}
(iv) in Section \ref{sus:a-mon characterization}). A mechanism $\mu =(q,s)$
is \emph{tie consistent }if the buyer breaks ties in a consistent manner,
i.e., in the same way at all valuations. This means that if the same two
distinct choices are optimal both at $x$ and at $y$, then $\mu $ cannot
choose one at $x$ and the other at $y$. For example, if the price of each
single good is $1$, we cannot have good $1$ allocated at $x=(2,2,0)$ and
good $2$ allocated at $y=(2,2,1)$. Formally, tie consistency requires that
if $q(x)$ is optimal at $y$ (which means that $q(y)\cdot
x-s(y)=b(x)=q(x)\cdot x-s(x)$) and $q(y)$ is optimal at $x$ then $q(x)=q(y)$%
. Choosing among the seller-favorable choices the set $A$ that maximizes,
say, $\sum_{i\in A}i$ yields a mechanism that is tie favorable and tie
consistent.

\begin{theorem}
\label{th:pm}Let $\mu $ be a deterministic buyer-favorable tie-consistent
mechanism, with a nondecreasing pricing function\footnote{%
Assume for simplicity that all prices are finite; $p$ need not be the
canonical pricing function.} $p:2^{K}\rightarrow \mathbb{R}_{+}$. A
necessary and sufficient condition for $\mu $ to be monotonic is: if%
\begin{equation*}
p(A)>p(B)
\end{equation*}%
for some $A,B\subseteq K$ then there is no $z\in \mathbb{R}^{A\backslash B}$
satisfying%
\begin{equation}
p(A)-p(A\backslash C)\leq \sum_{i\in C}z_{i}<p(B\cup C)-p(B)\text{\ for all }%
\emptyset \neq C\subseteq A\backslash B.  \label{eq:C}
\end{equation}
\end{theorem}

To clarify: such a $z$ would need to satisfy simultaneously \emph{all} $%
2\cdot (2^{\left\vert A\backslash B\right\vert }-1)$ inequalities in (\ref%
{eq:C}).

\begin{proof}
Let $x\in \mathbb{R}_{+}^{k}$, and let $A:=q(x)$ be the set of goods
allocated to $x$. If we increase coordinates in $A$ and decrease coordinates
outside $A$, i.e., if $x^{\prime }\in \mathbb{R}_{+}^{k}$ satisfies $%
x_{i}^{\prime }\geq x_{i}$ for all $i\in A$ and $x_{i}^{\prime }\leq x_{i}$
for all $i\notin A$, then $q(x^{\prime })=A$ as well. This is easily seen
since $[x(A)-p(A)]-[x(B)-p(B)]\geq \lbrack x^{\prime }(A)-p(A)]-[x^{\prime
}(B)-p(B)]$ for every $B$, which implies that, first, $A$ is optimal at $%
x^{\prime }$, and second, any $B$ that is optimal at $x^{\prime }$ is also
optimal at $x$; by tie consistency, $A$ must therefore be chosen at $%
x^{\prime }$ too.

Next, the monotonicity of $\mu $ is equivalent to: if $p(A)>p(B)$ then there
are no $x\leq y$ such that $q(x)=A$ and $q(y)=B$ (because then $%
s(x)=p(A)>p(B)=s(y))$. Applying the observation of the previous paragraph to
both $x$ and $y$, proceed as follows while keeping the inequality $x\leq y$:
decrease $x_{i}$ and $y_{i}$ to $0$ for $i\notin A\cup B$; increase $x_{i}$
and $y_{i}$ to a large number $M$ for $i\in A\cap B$; decrease $x_{i}$ to $0$
and increase $y_{i}$ to $M$ for $i\in B\backslash A$; and, finally, decrease 
$y_{i}$ to $x_{i}$ for $i\in A\backslash B$. All these changes therefore do
not affect the allocation, and so, by letting $z\in \mathbb{R}^{A\backslash
B}$ be the restriction of $x$ to $A\backslash B$, we have without loss of
generality that $q(x)=A$ and $q(y)=B$ for $x$ and $y$ of the form%
\begin{equation*}
\begin{tabular}{c|cccc}
& $A\backslash B$ & $A\cap B$ & $B\backslash A$ & $K\backslash (A\cup B)$ \\ 
\hline\hline
$x$ & $z$ & $(M,...,M)$ & $(0,...,0)$ & $(0,...,0)$ \\ 
$y$ & $z$ & $(M,...,M)$ & $(M,...,M)$ & $(0,...,0)$%
\end{tabular}%
\;.
\end{equation*}

For large $M$, the conditions that $q(x)=A$, namely $x(A)-p(A)\geq
x(A^{\prime })-p(A^{\prime })$ for all $A^{\prime }$, reduce to%
\begin{equation*}
z(C)\geq p(A)-p(A\backslash C)
\end{equation*}%
for all $C\subseteq A\backslash B$ (indeed, only sets $A^{\prime }$ such
that $A\cap B\subseteq A^{\prime }\varsubsetneq A$, i.e., $A^{\prime
}=A\backslash C$, matter). Similarly, the conditions for $q(y)=B$ reduce to%
\begin{equation*}
z(C)<p(B\cup C)-p(B)
\end{equation*}%
for all $C\subseteq A\backslash B$ (indeed, only sets $B^{\prime }$ such
that $B\varsubsetneq B^{\prime }\subseteq B\cup (A\backslash B)$, i.e., $%
B^{\prime }=B\cup C$, matter; the inequalities here are strict since
otherwise $B^{\prime }$ would be chosen at $y$ instead of $B$, by buyer
favorability).
\end{proof}

\bigskip

\noindent \textbf{Remark. }From the proof it follows that it suffices to
consider sets $A,B$ that are in the range of allocations $Q=q(\mathbb{R}%
_{+}^{k})$ of $\mu $, and, moreover, incomparable (i.e., $A\backslash B$ and 
$B\backslash A$ are both nonempty; indeed, if $A\subseteq B$ then\emph{\ }we
cannot have $p(A)>p(B)$, and if $B\varsubsetneqq A$ then condition (\ref%
{eq:C}) for $C=A\backslash B$ becomes $p(A)-p(B)\leq z(A\backslash
B)<p(A)-p(B)$, a contradiction that shows that there can be no such $z$).
This applies to the two corollaries below as well.

\bigskip

The nonexistence of $z$ satisfying (\ref{eq:C}) is a somewhat unwieldly
condition; we obtain from it two simpler conditions, one that is sufficient
and one that is necessary.

\begin{corollary}
\label{c:pm-subm}Let $\mu $ be a deterministic buyer-favorable
tie-consistent mechanism. Then $\mu $ is monotonic if 
\begin{equation}
p(A)+p(B)\geq p(A\cup B)+p(A\cap B)  \label{eq:submod-p}
\end{equation}%
for all $A,B\subseteq K$ with\footnote{%
Since (\ref{eq:submod-p}) is symmetric in $A$ and $B,$ requiring it when $%
p(A)>p(B)$ is the same as requiring it when $p(A)\neq p(B).$} $p(A)\neq p(B)$%
.
\end{corollary}

This is a \emph{restricted} version of submodularity, where inequality (\ref%
{eq:submod-p}) is \emph{not} required when $p(A)=p(B).$

\begin{proof}
If $p(A)>p(B)$ then there can be no $z$ satisfying (\ref{eq:C}), because for 
$C=A\backslash B$ it would require that $p(A)-p(A\cap B)\leq z(A\backslash
B)<p(A\cup B)-p(B)$, which contradicts (\ref{eq:submod-p}).
\end{proof}

\bigskip

\begin{corollary}
\label{c:pm-necessary}Let $\mu $ be a deterministic buyer-favorable
tie-consistent mechanism. If $\mu $ is monotonic then%
\begin{equation*}
p(A\cup \{i\})+p(A\cup J)\geq p(A\cup J\cup \{i\})+p(A)
\end{equation*}%
for all pairwise disjoint sets $A,\{i\},J\subseteq K$ for which%
\begin{equation*}
p(A\cup \{i\})>p(A\cup J).
\end{equation*}
\end{corollary}

\begin{proof}
When $A\backslash B$ is a singleton, say $A\backslash B=\{i\}$, condition (%
\ref{eq:C}) becomes $p(A)-p(A\backslash \{i\})\leq z_{i}<p(B\cup \{i\})-p(B)$%
, and so the nonexistence of $z$ is equivalent to $p(A)-p(A\backslash
\{i\})\geq p(B\cup \{i\})-p(B)$. Put $J:=B\backslash A$ and replace $A\cap B$
by $A.$
\end{proof}

\bigskip

For disjoint sets $A,I\subset K$, define the \emph{discrete derivative} of $%
p $ at $A$ in the \emph{direction} $I$ by%
\begin{equation*}
p^{\prime }(A;I):=p(A\cup I)-p(A);
\end{equation*}%
this is the \emph{marginal price} of the set of goods $I$, i.e., the change
in price due to adding $I$. With this notation, the condition of Corollary %
\ref{c:pm-necessary} becomes: 
\begin{equation*}
\text{if\ \ }p^{\prime }(A;\{i\})>p^{\prime }(A;J)\text{\ \ then\ \ }%
p^{\prime }(A;\{i\})\geq p^{\prime }(A\cup J;\{i\}),
\end{equation*}%
which is a \textquotedblleft decreasing marginal price" condition.
Similarly, putting $I:=A\backslash B$ and $J:=B\backslash A$, and replacing $%
A\cap B$ by $A$, the condition of Corollary \ref{c:pm-subm} becomes:%
\begin{equation}
\text{if\ \ }p^{\prime }(A;I)>p^{\prime }(A;J)\text{\ \ then\ \ }p^{\prime
}(A;I)\geq p^{\prime }(A\cup J;I).  \label{eq:AIJ}
\end{equation}%
Thus, a sufficient condition for monotonicity is \textquotedblleft (\ref%
{eq:AIJ}) for all pairwise disjoint $A,I,J,$" while a necessary condition
for monotonicity is \textquotedblleft (\ref{eq:AIJ}) for all pairwise
disjoint $A,\{i\},J$" (i.e., singleton sets $I)$.

Additional conditions for monotonicity are provided in Appendix \ref%
{s-a:det-mon}.

\section{Open Problems}

We list a number of issues that remain open, specifying in each case the
best that we know.

\begin{enumerate}
\item The main open problem is to find useful characterizations of the
monotonicity of mechanisms; see Section \ref{s:mon det} and Appendix \ref%
{s-a:det-mon} for the deterministic case, and the last paragraph in Appendix %
\ref{s-a:s-fav monot}. This should also help to address some of the other
open problems below.

\item The bound on the ratio between \textsc{MonRev} and \textsc{SRev}: it
is at most $k$ by Theorem \ref{th:monrev}, and at least $\Omega (\log k)$,
obtained by \textsc{BRev} (see (\ref{eq:brev=logk*srev})) and Example \ref%
{ex:harmonic}.

\item The bound on the ratio between \textsc{MonRev} and \textsc{AMonRev}:
it is at most $k$, as implied by Theorem \ref{th:monrev} (because \textsc{%
SRev}$\,\leq \,$\textsc{AMonRev}), and at least $\Omega (\log k)$ by Example %
\ref{ex:harmonic} (see Corollary \ref{c:amon vs mon}).

\item The bound on the ratio between \textsc{SymDRev} and \textsc{SRev}: it
is at most $O(\log ^{2}k)$ by Theorem \ref{th:symdrev}, and at least $\Omega
(\log k)$, obtained by \textsc{BRev} (see (\ref{eq:brev=logk*srev})).

\item The bound on the ratio between \textsc{MonRev} and \textsc{MonDRev},
the revenue from monotonic deterministic mechanisms: it is at most $k$, as
implied by Theorem \ref{th:monrev} (because \textsc{SRev}$\,\leq \,$\textsc{%
MonDRev}), and, of course, at least $1.$
\end{enumerate}

\appendix

\section{Appendices}

\subsection{Appendix: The Canonical Pricing Function\label{sus-a:canonical p}%
}

The canonical pricing function is nondecreasing, convex, and closed (see
Section \ref{susus:canonical p}); we show here that it is the \emph{unique }%
such pricing function. This generalizes the result of Proposition \ref%
{p:p0-det} in the deterministic case (where only the property of being
nondecreasing matters). Recall that $%
%TCIMACRO{\TeXButton{p_Q}{p_{\scriptscriptstyle Q}}}%
%BeginExpansion
p_{\scriptscriptstyle Q}%
%EndExpansion
:Q\rightarrow \mathbb{R}_{+}$ denotes the common restriction of all pricing
functions to $Q.$

\begin{proposition}
\label{p:p0}Let $\mu \ $be a mechanism, with range of allocations $%
Q\subseteq \lbrack 0,1]^{k}.$

\begin{description}
\item[(i)] The canonical pricing function $p_{0}$ of $\mu $ is the unique
pricing function of $\mu $ that is nondecreasing, convex, and closed, and is
defined as follows: for every $g\in \lbrack 0,1]^{k}$, let%
\begin{equation}
p_{1}(g)%
%TCIMACRO{\TeXButton{:=}{{\;:=\;}}}%
%BeginExpansion
{\;:=\;}%
%EndExpansion
\inf \sum_{j=1}^{J}\lambda ^{j}%
%TCIMACRO{\TeXButton{p_Q}{\pQ}}%
%BeginExpansion
\pQ%
%EndExpansion
(g^{j}),  \label{eq:p2}
\end{equation}%
where the infimum is taken over all convex combinations of elements of $Q$
that are no less than $g$---i.e., $\sum_{j=1}^{J}\lambda ^{j}g^{j}\geq g$,
where $\sum_{j=1}^{J}\lambda ^{j}=1$, $\lambda ^{j}\geq 0$ and $g^{j}\in Q$
for every $j=1,...,J$; then%
\begin{equation*}
p_{0}(g)=(\mathrm{cl\,}p_{1})(g):=\underline{\lim }_{h\rightarrow g}p_{1}(h).
\end{equation*}

\item[(ii)] If the set $Q$ is closed then the canonical pricing function $%
p_{0}$ of $\mu $ is the unique pricing function of $\mu $ that is
nondecreasing and convex; it is the function $p_{1}$ given by (\ref{eq:p2})
(i.e., in this case $p_{1}$ is closed, and $p_{0}=\mathrm{cl\,}p_{1}=p_{1}).$
\end{description}
\end{proposition}

\begin{proof}
(i) For each $g\in \lbrack 0,1]^{k}$ let%
\begin{equation*}
\phi _{g}(x)%
%TCIMACRO{\TeXButton{:=}{{\;:=\;}}}%
%BeginExpansion
{\;:=\;}%
%EndExpansion
\left\{ 
%TCIMACRO{\TeXButton{\arraystretch=0.5}{\renewcommand{\arraystretch}{0.5}}}%
%BeginExpansion
\renewcommand{\arraystretch}{0.5}%
%EndExpansion
\begin{tabular}{lll}
$g\cdot x-%
%TCIMACRO{\TeXButton{p_Q}{\pQ}}%
%BeginExpansion
\pQ%
%EndExpansion
(g),$ &  & if $x\geq 0,$ \\ 
&  &  \\ 
$\infty ,$ &  & otherwise.%
\end{tabular}%
%TCIMACRO{\TeXButton{\arraystretch=1}{\renewcommand{\arraystretch}{1}}}%
%BeginExpansion
\renewcommand{\arraystretch}{1}%
%EndExpansion
\right.
\end{equation*}%
Then $\phi _{g}$ is a convex function, and $b=\sup \{\phi _{g}:g\in Q\}$.
The Fenchel conjugate $\phi _{g}^{\ast }$ of $\phi _{g}$ is%
\begin{eqnarray*}
\phi _{g}^{\ast }(h) &=&\sup_{x}(h\cdot x-\phi _{g}(x))=\sup_{x\geq
0}(h\cdot x-g\cdot x+%
%TCIMACRO{\TeXButton{p_Q}{\pQ}}%
%BeginExpansion
\pQ%
%EndExpansion
(g)) \\
&=&\left\{ 
%TCIMACRO{\TeXButton{\arraystretch=0.5}{\renewcommand{\arraystretch}{0.5}}}%
%BeginExpansion
\renewcommand{\arraystretch}{0.5}%
%EndExpansion
\begin{tabular}{lll}
$%
%TCIMACRO{\TeXButton{p_Q}{\pQ}}%
%BeginExpansion
\pQ%
%EndExpansion
(g),$ &  & if $h\leq g,$ \\ 
&  &  \\ 
$\infty ,$ &  & otherwise.%
\end{tabular}%
%TCIMACRO{\TeXButton{\arraystretch=1}{\renewcommand{\arraystretch}{1}}}%
%BeginExpansion
\renewcommand{\arraystretch}{1}%
%EndExpansion
\right.
\end{eqnarray*}%
By Theorem 16.5 in Rockafellar (1970), the Fenchel conjugate $p_{0}$ of $b$
equals$\linebreak \mathrm{cl}\left( \mathrm{conv}\{\phi _{g}^{\ast }:g\in
Q\}\right) $. By definition, $\mathrm{conv}\{\phi _{g}^{\ast }:g\in Q\}(h)$
is the infimum of all convex combinations $\alpha =\sum_{j}\lambda ^{j}\phi
_{g^{j}}^{\ast }(h^{j})$, where $\sum_{j}\lambda ^{j}h^{j}=h$ and $g^{j}\in
Q $ for all $j$. Now the expression $\alpha $ is finite only if $h^{j}\leq
g^{j}$ (otherwise $\phi _{g^{j}}^{\ast }(h^{j})$ is infinite) for all $j$,
in which case we get $\alpha =\sum_{j}\lambda ^{j}%
%TCIMACRO{\TeXButton{p_Q}{\pQ}}%
%BeginExpansion
\pQ%
%EndExpansion
(g^{j})$ and $\sum_{j}\lambda ^{j}g^{j}\geq \sum_{j}\lambda ^{j}h^{j}=h$;
conversely, given $\sum_{j}\lambda ^{j}g^{j}\geq h$ we can always find $%
h^{j}\leq g^{j}$ for all $j$ such that $\sum_{j}\lambda ^{j}h^{j}=h$.
Therefore $\mathrm{conv}\{\phi _{g}^{\ast }:g\in Q\}(h)$ is the infimum of $%
\sum_{j}\lambda ^{j}%
%TCIMACRO{\TeXButton{p_Q}{\pQ}}%
%BeginExpansion
\pQ%
%EndExpansion
(g^{j})$ over all convex combinations $\sum_{j}\lambda ^{j}g^{j}\geq h$ with 
$g^{j}\in Q$ for all $j$, which is precisely $p_{1}(h)$, and so $p_{0}=%
\mathrm{cl\,}\left( \mathrm{conv}\{\phi _{g}^{\ast }:g\in Q\}\right) =%
\mathrm{cl\,}p_{1}.$

Uniqueness: let $p$ be any nondecreasing, convex, and closed pricing
function of $\mu $; since $p$ coincides with $%
%TCIMACRO{\TeXButton{p_Q}{\pQ}}%
%BeginExpansion
\pQ%
%EndExpansion
$ on $Q$, it follows that $p\leq p_{1}$ (because $p$ is nondecreasing and
convex), and thus $p\leq \mathrm{cl}\,p_{1}=p_{0}$ (because $p$ is closed).
Now $p\geq p_{0}$ (because $p_{0}$ is the minimal pricing function of $\mu )$%
, and so $p=p_{0}.$

(ii) When $Q\subseteq \lbrack 0,1]^{k}$ is a closed set (and thus compact)
the function $p_{1}$ is already closed. Indeed, let $g_{(n)}\rightarrow g$,
and for each $n$ take a convex combination $\sum_{j=1}^{J}\lambda
_{(n)}^{j}g_{(n)}^{j}\geq g_{(n)}$ (i.e., $\lambda _{(n)}^{j}\geq 0$ and $%
\sum_{j=1}^{J}\lambda _{(n)}^{j}=1)$ with $\sum_{j=1}^{J}\lambda _{(n)}^{j}%
%TCIMACRO{\TeXButton{p_Q}{\pQ}}%
%BeginExpansion
\pQ%
%EndExpansion
(g_{(n)}^{j})\leq p_{1}(g_{(n)})+1/n$, where all $g_{(n)}^{j}\in Q$ and $%
J\leq k+2$ (by Caratheodory's theorem, since $(h,%
%TCIMACRO{\TeXButton{p_Q}{\pQ}}%
%BeginExpansion
\pQ%
%EndExpansion
(h))\in \mathbb{R}^{k+1})$. The compactness of $Q$ implies that there is a
subsequence, without loss of generality the original sequence, such that $%
\lambda _{(n)}^{j}\rightarrow \lambda ^{j}$ and $g_{(n)}^{j}\rightarrow
g^{j} $ for every $j$. Since $Q$ is closed and $g_{(n)}^{j}\in Q$ we have $%
g^{j}\in Q$, and so $\underline{\lim }_{n}%
%TCIMACRO{\TeXButton{p_Q}{\pQ}}%
%BeginExpansion
\pQ%
%EndExpansion
(g_{(n)}^{j})\geq 
%TCIMACRO{\TeXButton{p_Q}{\pQ}}%
%BeginExpansion
\pQ%
%EndExpansion
(g^{j})$ (because $%
%TCIMACRO{\TeXButton{p_Q}{\pQ}}%
%BeginExpansion
\pQ%
%EndExpansion
=p_{0}$ on $Q$ and $p_{0}$ is lower semicontinuous); therefore in the limit
we get a convex combination $\sum_{j=1}^{J}\lambda ^{j}g_{j}\geq g$ (i.e., $%
\lambda ^{j}\geq 0$ and $\sum_{j=1}^{J}\lambda ^{j}=1)$ with%
\begin{equation*}
\sum_{j=1}^{J}\lambda ^{j}%
%TCIMACRO{\TeXButton{p_Q}{\pQ}}%
%BeginExpansion
\pQ%
%EndExpansion
(g^{j})\leq \underline{\lim }_{n}\sum_{j=1}^{J}\lambda _{(n)}^{j}%
%TCIMACRO{\TeXButton{p_Q}{\pQ}}%
%BeginExpansion
\pQ%
%EndExpansion
(g_{(n)}^{j})\leq \underline{\lim }_{n}p_{1}(g_{(n)}).
\end{equation*}%
By the definition of $p_{1}(g)$ we thus have $p_{1}(g)\leq $ $%
\sum_{j=1}^{J}\lambda ^{j}%
%TCIMACRO{\TeXButton{p_Q}{\pQ}}%
%BeginExpansion
\pQ%
%EndExpansion
(g^{_{j}})$, and so $p_{1}(g)\leq \underline{\lim }_{n}p_{1}(g_{(n)})$, as
claimed.
\end{proof}

\bigskip

A slightly different way of viewing this characterization of $p_{0}$ is
similar to the notion of the \textquotedblleft convexification" of a
function $f$, which is the maximal convex function that is $\leq f$ (such a
function always exists, as the supremum of all such functions is a convex
function $\leq f$). Proposition \ref{p:p0} (i) says that $p_{0}$ is the 
\emph{nondecreasing closed convexification} \emph{of} $%
%TCIMACRO{\TeXButton{p_Q}{\pQ}}%
%BeginExpansion
\pQ%
%EndExpansion
$, i.e., the maximal function\footnote{%
Maximal among \emph{all} functions, not only pricing functions.} that is
nondecreasing, closed, convex, and $\leq 
%TCIMACRO{\TeXButton{p_Q}{\pQ}}%
%BeginExpansion
\pQ%
%EndExpansion
$ on $Q$; again, this is the supremum of all such functions (because the
properties are preserved when we take the supremum). Since $%
%TCIMACRO{\TeXButton{p_Q}{\pQ}}%
%BeginExpansion
\pQ%
%EndExpansion
$ satisfies these properties on $Q$ (beacuse it is the restriction of $p_{0}$
to $Q$), we may also characterize $p_{0}$ as the \emph{maximal extension of} 
$%
%TCIMACRO{\TeXButton{p_Q}{\pQ}}%
%BeginExpansion
\pQ%
%EndExpansion
$ \emph{that is nondecreasing, convex, and closed}. The requirement reduces
to being a nondecreasing and convex function when $Q$ is a closed set (see
(ii)), and to being a nondecreasing function when we consider deterministic
pricing of a deterministic mechanism (see Proposition \ref{p:p0-det}).

\bigskip

The following examples show that both convexification and closure are indeed
needed.

\textbf{Example. }The need of \emph{convexification}. Let $k=2$, and
consider the separate selling of each good at price $1$. Thus, $%
Q=\{0,1\}^{2} $, and $%
%TCIMACRO{\TeXButton{p_Q}{\pQ}}%
%BeginExpansion
\pQ%
%EndExpansion
(0,0)=0$, $%
%TCIMACRO{\TeXButton{p_Q}{\pQ}}%
%BeginExpansion
\pQ%
%EndExpansion
(1,0)=%
%TCIMACRO{\TeXButton{p_Q}{\pQ}}%
%BeginExpansion
\pQ%
%EndExpansion
(0,1)=1$, $%
%TCIMACRO{\TeXButton{p_Q}{\pQ}}%
%BeginExpansion
\pQ%
%EndExpansion
(1,1)=2$. Let $g=(1/2,1/2)$; then the minimal nondecreasing pricing function 
$p_{2}$ satisfies $p_{2}(g)=2$ (since the only element of $Q$ that is $\geq
g $ is $(1,1))$. However, $p_{0}(g)=p_{1}(g)=1$ (use $g=1/2(1,0)+1/2(0,1)$).

\bigskip

\textbf{Example}. The need of \emph{closure}. Let $k=1$, and let $\mu =(q,s)$
be the mechanism whose menu consists of paying $2-2\delta $ for an
allocation of $1-\delta ^{2}$ of the good, for all $\delta \in (0,1]$. Then
it is straightforward to verify that every valuation $x\leq 1$ chooses $%
\delta =1$ (i.e., $(q(x),s(x))=(0,0))$, and every valuation $x\geq 1$
chooses $\delta =1/x$ (i.e., $(q(x),s(x))=(1-1/x^{2},2-2/x))$. The buyer
payoff function is thus $b(x)=0$ for $x\leq 1$ and $b(x)=x+1/x-2$ for $x\geq
1$; note that the limit option of paying $2$ for an allocation of $1$ (i.e., 
$\delta =0)$ is therefore never optimal: $x-2<b(x)$. Thus, $Q=[0,1)$ and $%
%TCIMACRO{\TeXButton{p_Q}{\pQ}}%
%BeginExpansion
\pQ%
%EndExpansion
(g)=2-2\sqrt{1-g}$ for every $g\in \lbrack 0,1)$, which implies that $%
p_{0}(1)=\lim_{g\rightarrow 1}%
%TCIMACRO{\TeXButton{p_Q}{\pQ}}%
%BeginExpansion
\pQ%
%EndExpansion
(g)=2$ (the inequality $\geq $ is because $p_{0}$ is nondecreasing, and the
inequality $\leq $ is because $p_{0}$ is closed). However, $p_{1}(1)=\infty $%
, because there is no $g^{\prime }\in \mathrm{conv\,}Q$ such that $g^{\prime
}\geq 1$.

To avoid infinite prices, we embed this into a two-good example: let $k=2$,
and let the menu consist of $((1-\delta ^{2},0),2-2\delta )$ for all $\delta
\in (0,1]$, together with $((1,1),3)$ (and so there is no need to include
the limit option of $((1,0),2)$, since for every $x\in \mathbb{R}_{+}^{2}$
with $x_{1}\geq 2$ the option $((1-1/x_{1}^{2},0),2-2/x_{1})$ yields a
strictly better payoff: $x_{1}+1/x_{1}-2>x_{1}-2$). Thus $Q=([0,1)\times
\{0\})\cup \{(1,1)\}$ and $%
%TCIMACRO{\TeXButton{p_Q}{\pQ}}%
%BeginExpansion
\pQ%
%EndExpansion
(g_{1},0)=2-2\sqrt{1-g_{1}}$ for all $g_{1}\in \lbrack 0,1)$, and $%
%TCIMACRO{\TeXButton{p_Q}{\pQ}}%
%BeginExpansion
\pQ%
%EndExpansion
(1,1)=3$, which yields $p_{0}(1,0)=2$ (because $p_{0}$ is nondecreasing and
closed), whereas $p_{1}(1,0)=3$ (because the only element of $\mathrm{conv}%
\,Q$ that is $\geq (1,0)$ is $(1,1)).$

\subsection{Appendix: Pricing Approximation\label{sus-a:chawla}}

We provide here the result of Chawla, Teng, and Tzamos (2022) on comparing
revenues by pricing functions, which we use in Sections \ref{sus:a-mon rev}
and \ref{sus:sym det general}, and Appendix \ref{sus-a:det-super}. Let $%
\mathcal{P}$ and $\mathcal{P}^{\prime }$ be two classes of $k$-good pricing
functions, defined on $G=\{0,1\}^{k}$ (in the deterministic case) or on $%
G=[0,1]^{k}$ (in the general case); the pricing functions need not be
canonical. The maximal revenues that are obtainable by mechanisms with
pricing functions in $\mathcal{P}$ and $\mathcal{P}^{\prime }$ are denoted
by $\mathcal{P}$-\textsc{Rev} and $\mathcal{P}^{\prime }$-\textsc{Rev},
repsectively. A set $Z$ is a \emph{cone} if $z\in Z$ implies $\alpha z\in Z$
for every scalar $\alpha \geq 0$.

\begin{theorem}[Chawla, Teng, and Tzamos]
\label{th:chawla}Let $\mathcal{P}^{\prime }$ be a cone of nondecreasing and
closed $k$-good pricing functions. Assume that there are constants $%
0<c_{1}<c_{2}<\infty $ such that for every $p\in \mathcal{P}$ there is $%
p^{\prime }\in \mathcal{P}^{\prime }$ satisfying%
\begin{equation}
c_{1}p^{\prime }(g)\leq p(g)\leq c_{2}p^{\prime }(g)  \label{eq:p-p'}
\end{equation}%
for every $g$ in $G$; then%
\begin{equation*}
\mathcal{P}\text{-\textsc{Rev}}(X)\leq 2\ln \left( 2\,\frac{c_{2}}{c_{1}}%
\right) \cdot \mathcal{P}^{\prime }\text{-\textsc{Rev}}(X)
\end{equation*}%
for every $k$-good random valuation $X$.
\end{theorem}

We have added the technical assumptions that the pricing functions are
nondecreasing and closed (conditions that are always satisfied by canonical
pricing functions, and also by the separable functions that we generate from
them). These assumptions are used to fill in some missing details in the
proof of Lemma 3.1 of Chawla, Teng, and Tzamos (2022). Specifically, they
yield (in their notation) the existence of a maximal $\lambda _{\alpha }$,
and allow the application of the envelope theorem (Milgrom and Segal 2002,
Theorem 2): the function $u_{v}$ is Lipschitz, and hence absolutely
continuous.

\subsection{Appendix: Monotonic Mechanisms\label{s-a:s-fav monot}}

In this appendix we show, first, that if a mechanism is monotonic, then so
are its seller-favorable, and thus tie-favorable, versions.

Let $\mu =(q,s)$ be a mechanism, and let $\tilde{\mu}=(\tilde{q},\tilde{s})$
be a seller-favorable version of $\mu $; thus, $\mu $ and $\tilde{\mu}$ have
the same buyer payoff function $b$, and $\tilde{s}(x)=b^{\prime }(x;x)-b(x)$
for every $x$. Denote by $\mathcal{D}\equiv \mathcal{D}_{b}\subseteq \mathbb{%
R}_{++}^{k}$ the set of points $x$ in the interior of $\mathbb{R}_{+}^{k}$
where the convex function $b$ is differentiable, i.e., $\partial
b(x)=\{\nabla b(x)\}$ (where $\nabla b(x)$ denotes the gradient of $b$ at $x$%
); the set $\mathcal{D}$ is dense in $\mathbb{R}_{++}^{k}$, and hence in $%
\mathbb{R}_{+}^{k}$, and its complement $\mathbb{R}_{+}^{k}\backslash 
\mathcal{D}$ has Lebesgue measure zero (by Theorem 25.5 in Rockafellar
1970). The mechanisms $\mu $ and $\tilde{\mu}$ may differ only on $\mathbb{R}%
_{+}^{k}\backslash \mathcal{D}$: at every $x\in \mathcal{D}$ we have $q(x)=%
\tilde{q}(x)=\nabla b(x)$ and $s(x)=\tilde{s}(x)=b^{\prime
}(x;x)-b(x)=\nabla b(x)\cdot x-b(x).$

\begin{proposition}
\label{p:s-fav monot}Let $\mu $ be a mechanism with buyer payoff function $b$%
, and let $\tilde{\mu}$ be a seller-favorable mechanism with the same buyer
payoff function $b$. If $\mu $ is monotonic then $\tilde{\mu}$ is monotonic.
\end{proposition}

\begin{proof}
First, we claim that for every $x\in \mathbb{R}_{+}^{k}$ there is a sequence 
$\{x_{n}\}_{n=1}^{\infty }\subset \mathcal{D}$ such that $x_{n}\rightarrow x$
and $\tilde{s}(x_{n})\rightarrow \tilde{s}(x)$. Indeed, for $x\neq \mathbf{0}
$ let $\{x_{n}\}_{n}$ be a sequence of points in $\mathcal{D}$ such that $%
x_{n}\rightarrow x$ from the direction $x$, i.e., $\left\Vert
x_{n}-x\right\Vert ^{-1}(x_{n}-x)\rightarrow \left\Vert x\right\Vert ^{-1}x$
(such a sequence exists since $\mathcal{D}$ is dense in $\mathbb{R}_{++}^{k}$%
: for example, take $x_{n}\in \mathcal{D}$ to be within $1/n^{2}$ of $%
(1+1/n)x$). By Theorem 24.6 in Rockafellar (1970), every limit point of the
sequence $\{\nabla b(x_{n})\}_{n}$ belongs to $\partial b(x)_{x}$, and so $%
\tilde{s}(x_{n})=\nabla b(x_{n})\cdot x_{n}-b(x_{n})\rightarrow b^{\prime
}(x;x)-b(x)=\tilde{s}(x)$. For $x=\mathbf{0}$, for any sequence $%
\{x_{n}\}_{n}$ in $\mathcal{D}$ with $x_{n}\rightarrow \mathbf{0}$ we have $%
\tilde{s}(x_{n})=\nabla b(x_{n})\cdot x_{n}-b(x_{n})\rightarrow 0-b(\mathbf{%
0)}=b^{\prime }(\mathbf{0};\mathbf{0})-b(\mathbf{0})=\tilde{s}(\mathbf{0)}$.

Second, we claim that if $\tilde{s}$ is nondecreasing on $\mathcal{D}$ then $%
\tilde{s}$ is nondecreasing on $\mathbb{R}_{+}^{k}$. Indeed, take $x\leq y$
to be two points in $\mathbb{R}_{+}^{k}$. Take $x_{n}\in \mathcal{D}$ such
that $x_{n}\rightarrow x$ and $\tilde{s}(x_{n})\rightarrow \tilde{s}(x)$;
put $y_{n}:=x_{n}+y-x\geq x_{n}$, and let $y_{n}^{\prime }\in \mathcal{D}$
be such that $y_{n}^{\prime }\geq y_{n}$ and $\left\Vert y_{n}^{\prime
}-y_{n}\right\Vert \leq 1/n$ (the existence of $y_{n}^{\prime }$ is by the
density of $\mathcal{D}$ in $\mathbb{R}_{+}^{k}$). Thus, $x_{n}\leq
y_{n}^{\prime }$ and $y_{n}^{\prime }\rightarrow y$, and so $\tilde{s}%
(x_{n})\leq \tilde{s}(y_{n}^{\prime })$ (because $\tilde{s}$ is
nondecreasing on $\mathcal{D}$, which contains both $x_{n}$ and $%
y_{n}^{\prime }$), and $\overline{\lim }_{n}\,\tilde{s}(y_{n}^{\prime })\leq 
\tilde{s}(y)$ (because $\tilde{s}$ is upper semicontinuous\footnote{%
A real function $f$ is \emph{upper semicontinuous} if $\overline{\lim }%
_{y\rightarrow x}f(y)\leq f(x)$ for every $x;$ equivalently, the set $%
\{x:f(x)\geq t\}$ is closed for every real $t.$}: use Theorem 24.5 in
Rockafellar 1970 for $b^{\prime }$); altogether, $\tilde{s}(x)=\lim_{n}%
\tilde{s}(x_{n})\leq \overline{\lim }_{n}\,\tilde{s}(y_{n}^{\prime })\leq 
\tilde{s}(y).$

This completes the proof: if $s$ is nondecreasing on $\mathbb{R}_{+}^{k}$
then $\tilde{s}$ is nondecreasing on $\mathcal{D}$ (where it coincides with $%
s$), and so it is nondecreasing on $\mathbb{R}_{+}^{k}.$
\end{proof}

\bigskip

Second, we provide a formulation of monotonicity in terms of the allocation
function $q$. Assume for simplicity that the function $b$ is twice
differentiable;\footnote{%
And hence differentiable, which implies that there is a unique mechanism
with buyer payoff function $b.$} then we have $\nabla s(x)=\nabla (\nabla
b(x)\cdot x-b(x))=\nabla ^{2}b(x)\,x$, where $\nabla ^{2}b(x)\,x$ is the
product of the $k\times k$ matrix $\nabla ^{2}b(x)$ (the Hessian of $b$ at $%
x $) and the column vector $x$. Consider $q(tx)=\nabla b(tx)$ as a function
of $t\geq 0$; its derivative is $dq(tx)/dt=\nabla ^{2}b(tx)\,x=(1/t)\nabla
s(tx) $. The monotonicity of the mechanism is equivalent to $\nabla s(x)\geq
0$ for all $x$, and so it is equivalent also to the function $q(tx)$ being
nondecreasing in $t$ for every $x$, i.e., to \emph{the allocation function }$%
q$\emph{\ being nondecreasing along any ray from the origin} (cf. allocation
monotonicity, where $q$ is nondecreasing everywhere).

\subsection{Appendix: Revenue of Monotonic Mechanisms\label{s-a:monrev}}

We generalize the result of Theorem \ref{th:monrev} (see Section \ref%
{s:mon-rev}). A \emph{bundling partition }of the goods is $\Pi
=\{K_{j}\}_{j\in J}$, where $\cup _{j\in J}K_{j}=K=\{1,...,k\}$ and the sets 
$K_{j}$ are disjoint (i.e., $K_{j}\cap K_{j^{\prime }}=\emptyset $ for $%
j\neq j^{\prime }$). The corresponding revenue $\Pi $-\textsc{Rev} is
defined by%
\begin{equation*}
\Pi \text{-\textsc{Rev}}(X):=\sum_{j\in J}\text{\textsc{Rev}}\left(
\sum_{i\in K_{j}}X_{i}\right) .
\end{equation*}%
This is the maximal revenue obtained by selling each bundle of goods $K_{j}$
separately, i.e. (by Myerson's single-good result (\ref{eq:one good})), by
setting a price $p_{j}$ for each bundle $K_{j}$. The partition into
singletons, $\Pi _{\text{\textsc{S}}}=\{\{1\},...,\{k\}\}$, yields the
separate revenue \textsc{SRev}, and the singleton partition, i.e., $\Pi _{%
\text{\textsc{B}}}=\{K\}$, yields the bundling revenue \textsc{BRev}. Since
all bundling-partition mechanisms are clearly monotonic, it follows that $%
\Pi $-\textsc{Rev}$(X)\leq \,$\textsc{MonRev}$(X)$ (cf. (\ref{eq:SBrev})).

Theorem \ref{th:monrev} easily generalizes to:

\begin{proposition}
\label{p:pi-rev}Let $X$ be a $k$-good random valuation. Then 
\begin{equation*}
\text{\textsc{MonRev}}(X)\leq k\cdot \min_{\Pi }\Pi \text{-\textsc{Rev}}(X),
\end{equation*}%
where the minimum is taken over all partitions $\Pi $ of the set of goods.
\end{proposition}

\begin{proof}
Let $\Pi =(K_{j})_{j\in J}$ be a partition of $K$. For each $j\in J$ put $%
Y_{j}:=\max_{i\in K_{j}}X_{i}$; then $\max_{i\in K}X_{i}=\max_{j\in J}Y_{j}$%
, and so%
\begin{eqnarray*}
\text{\textsc{Rev}}\left( \max_{i\in K}X_{i}\right) &=&\text{\textsc{Rev}}%
\left( \max_{j\in J}Y_{j}\right) \leq \sum_{j\in J}\text{\textsc{Rev}}\left(
Y_{j}\right) \\
&\leq &\sum_{j\in J}\text{\textsc{Rev}}\left( \sum_{i\in K_{j}}X_{i}\right)
=\Pi \text{-\textsc{Rev}}(X),
\end{eqnarray*}%
where we have used Proposition \ref{p:max} twice, in its \textsc{SRev}
version for the first inequality, and in its \textsc{BRev} version for the
second. Apply (\ref{eq:k-max}).
\end{proof}

\subsection{Appendix: Submodular Deterministic Mechanisms\label{sus-a:submod
p}}

A deterministic mechanism is submodular if its canonical deterministic
pricing function $p_{0}^{\text{\textsc{D}}}$ is submodular (see Section \ref%
{sus:s-modularity}); we show here that it suffices that \emph{some} pricing
function is submodular, and only on the range of allocations $Q.$

\begin{proposition}
\label{p:subm-p0}Let $\mu $ be a deterministic mechanism with range of
allocations $Q\subseteq 2^{K}$. If some deterministic pricing function $p$
of $\mu $ satisfies the submodularity inequality (\ref{eq:subm-sets}) for
every\footnote{%
Whether $A\cup B$ and $A\cap B$ are in $Q$ or not.} $A,B\in Q$, then the
canonical deterministic pricing function $p_{0}^{\text{\textsc{D}}}$ of $\mu 
$ satisfies (\ref{eq:subm-sets}) for every $A,B\subseteq K.$
\end{proposition}

\begin{proof}
Let $A,B$ be two subsets of $K$ with $p_{0}(A),p_{0}(B)<\infty $ (otherwise
the inequality below is trivial); by Proposition \ref{p:p0-det}, there are $%
A^{\prime },B^{\prime }\in Q$ with $A^{\prime }\supseteq A$ and $B^{\prime
}\supseteq B$ such that $p_{0}(A)=%
%TCIMACRO{\TeXButton{p_Q}{\pQ}}%
%BeginExpansion
\pQ%
%EndExpansion
(A^{\prime })=p(A^{\prime })$ and $p_{0}(B)=%
%TCIMACRO{\TeXButton{p_Q}{\pQ}}%
%BeginExpansion
\pQ%
%EndExpansion
(B^{\prime })=p(B^{\prime })$, and so%
\begin{eqnarray*}
p_{0}(A)+p_{0}(B) &=&p(A^{\prime })+p(B^{\prime })\geq p(A^{\prime }\cup
B^{\prime })+p(A^{\prime }\cap B^{\prime }) \\
&\geq &p_{0}(A^{\prime }\cup B^{\prime })+p_{0}(A^{\prime }\cap B^{\prime
})\geq p_{0}(A\cup B)+p_{0}(A\cap B),
\end{eqnarray*}%
where the first inequality is by the submodularity of $p$ on $Q$, the second
because $p\geq p_{0}$, and the third because $p_{0}$ is nondecreasing.
\end{proof}

\subsection{Appendix: Ultramodular Functions\label{sus-a:utramodular}}

The following result is used in the proof of part (ii) of Proposition \ref%
{p:AM2supM} in Section \ref{sus:a-mon characterization}. Marinacci and
Montrucchio (2005, Lemma 5.1 and Theorem 5.4) prove it for interior points;
for completeness, we provide here a short version of their proof that
applies to boundary points as well. Recall the notation $\nabla
^{+}f(x):=(f^{\prime }(x;e^{i}))_{i=1,...,k}$.

\begin{proposition}
\label{p:delta+}Let $f:\mathbb{R}_{+}^{k}\rightarrow \mathbb{R}$ be an
ultramodular function. Then the vector $\nabla ^{+}f(x)$ is the maximal
subgradient of $f$ at $x$ for every $x\in \mathbb{R}_{+}^{k}$, and the
function $x\longmapsto \nabla ^{+}f(x)$ is nondecreasing in $x$ (i.e., $%
\nabla ^{+}f(y)\geq \nabla ^{+}f(x)$ for every $y\geq x$).
\end{proposition}

\begin{proof}
(i) Let $g\in \partial f(x)$; then $f(x+\delta e^{i})-f(x)\geq \delta g_{i}$
for every $\delta >0$ by the subgradient inequality, and so, by dividing by $%
\delta $ and taking the limit as $\delta \rightarrow 0^{+}$, we get $%
f^{\prime }(x;e^{i})\geq g_{i}$. Thus,%
\begin{equation*}
\partial f(x)\subseteq \{g:g\leq \nabla ^{+}f(x)\}.
\end{equation*}

(ii) Let $\delta >0$; we have%
\begin{eqnarray*}
f(x+\delta e)-f(x) &=&[f(x_{1}+\delta ,...,x_{k}+\delta
)-f(x_{1},x_{2}+\delta ,...,x_{k}+\delta )] \\
&&+[f(x_{1},x_{2}+\delta ,...,x_{k}+\delta )-f(x_{1},x_{2},x_{3}+\delta
,...,x_{k}+\delta )] \\
&&+...+[f(x_{1},...,x_{k-1},x_{k}+\delta )-f(x_{1},...,x_{k})] \\
&\geq &\sum_{i=1}^{k}[f(x_{1},...,x_{i-1},x_{i}+\delta
,x_{i+1},...,x_{k})-f(x)],
\end{eqnarray*}%
where the inequality is by supermodularity (term by term, since $\delta >0)$%
. Dividing by $\delta $ and taking the limit as $\delta \rightarrow 0^{+}$
yields%
\begin{equation*}
f^{\prime }(x;e)\geq \sum_{i=1}^{k}f^{\prime }(x;e^{i})=e\cdot \nabla
^{+}f(x).
\end{equation*}

(iii) By (i) we have $f^{\prime }(x;e)=\max_{g\in \partial f(x)}e\cdot g\leq
e\cdot \nabla ^{+}f(x)$, which together with (ii) yields equality. Therefore 
$\nabla ^{+}f(x)$ must belong to the closed set $\partial f(x)$ (otherwise
the maximum would be strictly smaller), and is thus its maximal element.

(iv) Let $y\geq x$; then $f(y+\delta e^{i})-f(y)\geq f(x+\delta e^{i})-f(x)$
for every $\delta >0$ by ultramodularity, and so, by dividing by $\delta $
and taking the limit as $\delta \rightarrow 0^{+}$, we get%
\begin{equation*}
f^{\prime }(y;e^{i})\geq f^{\prime }(x;e^{i}).
\end{equation*}
\end{proof}

\subsection{Appendix: Separable Subadditivity without Allocation
Monotonicity \label{sus-a:sep-subadd}}

Here we show that the converse of Theorem \ref{th:am} (iii) in Section \ref%
{sus:a-mon characterization} does not hold for two or more goods: separably
subadditive pricing may not suffice for allocation monotonicity. This is
easy to see, already in the class of symmetric deterministic mechanisms
(Example \ref{ex:subadd-not-am-3}). However, it requires at least $3$ goods,
and so we provide another example for $2$ goods (Example \ref%
{ex:subadd-not-am-2}).

\begin{example}
\label{ex:subadd-not-am-3}Let $k\geq 3$, and let $p$ be the symmetric
deterministic pricing function $p(m)=m$ for every $m\neq 1$ and $p(1)=2$.
Then $p$ is separably subadditive, because $p(m)+p(n)-p(m+n)$ equals $0,1$,
or $2$ (according to how many of $m,n$ equal $1$). But $p$ is \emph{not}
submodular, because $p(2)+p(2)<p(3)+p(1)$ (e.g., $p(\{1,2\})+p(\{1,3\})<p(%
\{1,2,3\})+p(\{1\})$), and thus the corresponding deterministic mechanism is
not allocation monotonic by Theorem \ref{th:am} (iii).
\end{example}

\begin{example}
\label{ex:subadd-not-am-2}Let $k=2$. We first construct such a function $b$
on a bounded domain, specifically $D:=[0,1]^{2}$, and then show how to
extend it to all $\mathbb{R}_{+}^{2}$. Put $b(x):=[f(x)]_{+}$ for $x\in
\lbrack 0,1]^{2}$, where $f$ is the convex function $\
f(x_{1},x_{2}):=(1/3)(x_{1}^{2}+x_{2}^{2}+x_{1}+x_{2}-x_{1}x_{2}-2)$. It is
easy to see that $b=0$ on the convex set $D_{0}:=\{x\in \lbrack
0,1]^{2}:f(x)\leq 0\}$, whose boundary goes through the three points $(1,0),(%
\sqrt{3}-1,\sqrt{3}-1)$, and $(0,1)$, and the gradient of $b$ is $(0,0)$ for 
$x$ in the interior of $D_{0}$, and $(1/3)(2x_{1}+1-x_{2},2x_{2}+1-x_{1})$
for $x$ in the interior of $D\backslash D_{0}$, and so all subgradients of $%
b $ lie in $[0,1]^{2}$ (use $\nabla f(x)\geq 0$ when $f(x)\geq 0$); thus, $b$
is continuous, nondecreasing, nonexpansive, and convex on $D$. Moreover, it
is separably superadditive there, because $f(x_{1},0)$ and $f(0,x_{2})$ are $%
\leq 0$ for $0\leq x_{i}\leq 1$, and so $b(x_{1},0)+b(0,x_{2})=0+0\leq b(x)$%
. However, it is \emph{not} supermodular because, for interior points of $%
D\backslash D_{0}$ (in particular, for $x\ll (1,1)$ close to $(1,1)$) we
have $\partial ^{2}b/\partial x_{1}\partial x_{2}=-1/3<0.$

To extend $b$ to all $\mathbb{R}_{+}^{2}$ while keeping all the above
properties, we define $b(x):=b(\tilde{x})+[x_{1}-1]_{+}+[x_{2}-1]_{+}$,
where $\tilde{x}_{i}:=\min \{x_{i},1\}$ (cf. Example \ref{ex:AM not subm}
for a similar extension). Thus, there are four regions, $%
C_{00}:=D=[0,1]^{2},\;C_{10}:=[1,\infty )\times \lbrack
0,1],\;C_{01}:=[0,1]\times \lbrack 1,\infty )$, and $C_{11}:=[1,\infty )^{2}$
(the subscript gives the coordinates of the minimal point in each region),
where the function $b$ is, respectively, $b_{00}(x):=[f(x)]_{+},%
\;b_{10}(x):=(1/3)x_{2}^{2}+x_{1}-1,\;b_{01}(x):=(1/3)x_{1}^{2}+x_{2}-1$,
and $b_{11}(x):=x_{1}+x_{2}-5/3$. We will see below that $b$ is a convex
function. Since all its gradients lie in $[0,1]^{2}$, it is nondecreasing
and nonexpansive. It is also separably superadditive, because%
\begin{eqnarray*}
b(x)-b(x_{1},0)-b(0,x_{2}) &=&(b(\tilde{x})+[x_{1}-1]_{+}+[x_{2}-1]_{+}) \\
&&-(b(\tilde{x}_{1},0)+[x_{1}-1]_{+})-(b(0,\tilde{x}_{2})+[x_{2}-1]_{+}) \\
&=&b(\tilde{x})-b(\tilde{x}_{1},0)-b(0,\tilde{x}_{2})\geq 0
\end{eqnarray*}%
because $\tilde{x}\in D$. However, as we have already seen, $b$ is \emph{not}
supermodular (already on $D$).

It remains to show that $b$ is a convex function. Since this is equivalent
to convexity along any straight-line segment, and $b$ is convex in each
region, we have to show that at every boundary point $x^{\ast }$, for
example $x^{\ast }\in C_{00}\cap C_{10}$, as we cross from $C_{00}$ to $%
C_{10}$ in the direction $z$, we have $b_{00}^{\prime }(x^{\ast };z)\leq
b_{10}^{\prime }(x^{\ast };z)$, i.e., $(\nabla b_{10}(x^{\ast })-\nabla
b_{00}(x^{\ast }))\cdot z\geq 0$. Indeed: $x_{1}^{\ast }=1$ and $0\leq
x_{2}^{\ast }\leq 1$; since we go from $\{x_{1}<1\}$ to $\{x_{1}>1\}$, we
have $z_{1}\geq 0$ (with no restriction on $z_{2}$); thus $\nabla
b_{00}(x^{\ast })=(1/3)(2x_{1}^{\ast }+1-x_{2}^{\ast },2x_{2}^{\ast
}+1-x_{1}^{\ast })=(1-(1/3)x_{2}^{\ast },(2/3)x_{2}^{\ast })$ and $\nabla
b_{10}(x^{\ast })=(1,(2/3)x_{2}^{\ast })$, and the desired inequality holds%
\footnote{%
The condition for convexity is thus $\partial b_{00}/\partial x_{1}\leq
\partial b_{10}/\partial x_{1}$ and $\partial b_{00}/\partial x_{2}=\partial
b_{10}/\partial x_{2}$ at each $x^{\ast }\in C_{00}\cap C_{10}.$} (because $%
x_{2}^{\ast }\geq 0$ and $z_{1}\geq 0$). Another case: for $x^{\ast }\in
C_{01}\cap C_{11}$ and $z_{1}\geq 0$, we have $x_{1}^{\ast }=1$, $\nabla
b_{01}(x^{\ast })=(2/3,1)$, and $\nabla b_{11}(x^{\ast })=(1,1)$, and so
again $(\nabla b_{11}(x^{\ast })-\nabla b_{01}(x^{\ast }))\cdot z\geq 0$.
The other cases are similar.
\end{example}

\subsection{Appendix: Monotonicity of Deterministic Mechanisms\label%
{s-a:det-mon}}

Here we obtain further conditions for the monotonicity of deterministic
mechanisms (see Section \ref{s:mon det}).

We start with a sufficient condition that is weaker than that of Corollary %
\ref{c:pm-subm}, and is easily obtained from Theorem \ref{th:pm}.

\begin{corollary}
\label{c:detmon-suff2}Let $\mu $ be a deterministic buyer-favorable
tie-consistent mechanism, with a nondecreasing pricing function $%
p:2^{K}\rightarrow \mathbb{\mathbb{R}}_{+}$. A sufficient condition for $\mu 
$ to be monotonic is: if%
\begin{equation*}
p(A)>p(B)
\end{equation*}%
for some $A,B\subseteq K$, then there is $\emptyset \neq C\subseteq
A\backslash B$ such that 
\begin{equation}
p(A)-p(A\backslash C)\geq p(B\cup C)-p(B).  \label{eq:CC}
\end{equation}
\end{corollary}

\begin{proof}
Condition (\ref{eq:CC}) contradicts condition (\ref{eq:C}).
\end{proof}

\bigskip

Putting $I:=A\backslash B$, $J:=B\backslash A$, and replacing $A\cap B$ with 
$A$, this can rewritten as follows (cf. (\ref{eq:AIJ})): if%
\begin{equation*}
p^{\prime }(A;I)>p^{\prime }(A;J)
\end{equation*}%
then there is $\emptyset \neq \tilde{I}\subseteq I$ such that%
\begin{equation*}
p^{\prime }(A\cup I\backslash \tilde{I};\tilde{I})\geq p^{\prime }(A\cup J;%
\tilde{I}).
\end{equation*}

Proposition 5 in Hart and Reny (2015) shows that seller-favorable
deterministic symmetric mechanisms are monotonic; the same holds for
buyer-favorable (tie-consistent) mechanisms.

\begin{corollary}
Let $\mu $ be a deterministic symmetric buyer-favorable tie-consistent
mechanism. Then $\mu $ is monotonic.
\end{corollary}

\begin{proof}
If $p(A)>p(B)$ then $\left\vert A\right\vert >|B|$ (because $p$ is
nondecreasing), and so taking a set $C\subseteq A\backslash B$ with of size $%
\left\vert C\right\vert =\left\vert A\right\vert -\left\vert B\right\vert $
(such a $C$ exists since $\left\vert A\backslash B\right\vert \geq
\left\vert A\right\vert -\left\vert B\right\vert $) yields $%
p(A)-p(A\backslash C)=p(B\cup C)-p(B)$ (because $|A|=|B\cup C|$ and $%
|A\backslash C|=|B|$), which is (\ref{eq:CC}).
\end{proof}

\bigskip

Using a \textquotedblleft Theorem of the Alternative" yields another
characterization of monotonicity for deterministic mechanisms. The next
proposition gives the technical result; we denote by $\mathcal{P(I)}$ the
set of \emph{nonempty} subsets of a finite set $I.$

\begin{proposition}
\label{p:v,w}Let $I$ be a nonempty finite set, and let $v,w:\mathcal{P}%
(I)\rightarrow \mathbb{R}$. The following two statements are equivalent:

(I) There is no $z\in \mathbb{R}^{I}$ such that%
\begin{equation*}
v(C)\leq z(C)<w(C)\text{\ for all }C\in \mathcal{P}(I).
\end{equation*}

(II) There are $\lambda ,\mu \in \mathbb{R}_{+}^{\mathcal{P}(I)}\backslash
\{0\}$ such that%
\begin{eqnarray*}
\sum_{C\in \mathcal{P}(I)}\lambda _{C}\mathbf{1}_{C} &=&\sum_{C\in \mathcal{P%
}(I)}\mu _{C}\mathbf{1}_{C}\text{\qquad and} \\
\sum_{C\in \mathcal{P}\mathbf{(}I)}\lambda _{C}v(C) &\geq &\sum_{C\in 
\mathcal{P}\mathbf{(}I)}\mu _{C}w(C).
\end{eqnarray*}
\end{proposition}

\begin{proof}
Let $\mathcal{M}$ denote the $(2^{|I|}-1)\times |I|$ matrix whose rows are
the characteristic vectors $\mathbf{1}_{C}\in \mathbb{R}^{I}$ of all $C\in 
\mathcal{P}(I)$. Writing $v$ for the vector $(v(C))_{C\in \mathcal{P}(I)}$
and $w$ for the vector $(w(C))_{C\in \mathcal{P}(I)}$, (I) is equivalent to
the system of inequalities%
\begin{eqnarray*}
\mathcal{M}z-v\zeta &\geq &0 \\
-\mathcal{M}z+w\zeta &\gg &0 \\
\zeta &>&0
\end{eqnarray*}%
\emph{not} having a solution $z\in \mathbb{R}^{I},\;\zeta \in \mathbb{R}$
(indeed, for one direction take $\zeta =1$, and for the converse replace $z$
with $(1/\zeta )z)$. By Motzkin's Theorem of the Alternative (see, e.g.,
Mangasarian 1994), this is equivalent to the system%
\begin{eqnarray*}
\lambda ^{\top }\mathcal{M+\mu }^{\top }(-\mathcal{M}) &=&0 \\
\lambda ^{\top }(-v)+\mathcal{\mu }^{\top }w+\eta &=&0 \\
\lambda ,\mu ,\eta &\geq &0 \\
(\mu ,\eta ) &\neq &0
\end{eqnarray*}%
having a solution $\lambda ,\mu \in \mathbb{R}^{I},\eta \in \mathbb{R}$.
Now, if $\mu =0$ then $\lambda =0$ (because $\lambda ^{\top }\mathcal{M}=%
\mathcal{\mu }^{\top }\mathcal{M}=0\mathcal{\ }$implies $\lambda =0)$, and
then $\eta =0$ (because $\eta =\lambda ^{\top }v-\mathcal{\mu }^{\top }w)$,
contradicting $(\mu ,\eta )\neq 0$. Therefore $\mu \neq 0$, and so $\lambda
\neq 0$, and we get%
\begin{eqnarray*}
\lambda ^{\top }\mathcal{M} &=&\mathcal{\mu }^{\top }\mathcal{M} \\
\lambda ^{\top }v &\geq &\mathcal{\mu }^{\top }w,
\end{eqnarray*}%
which is (II).
\end{proof}

\bigskip

Applying this equivalence to the result of Theorem \ref{th:pm} yields

\begin{theorem}
\label{th:pm-alt}Let $\mu $ be a deterministic buyer-favorable
tie-consistent mechanism, with a nondecreasing pricing function $%
p:2^{K}\rightarrow \mathbb{R}_{+}$. A necessary and sufficient condition for 
$\mu $ to be monotonic is: if%
\begin{equation*}
p(A)>p(B)
\end{equation*}%
for some $A,B\subseteq K$, then there are $\lambda ,\mu \in \mathbb{R}_{+}^{%
\mathcal{P}\mathbf{(}A\backslash B)}\backslash \{0\}$ such that 
\begin{eqnarray*}
\sum_{C\in \mathcal{P}(A\backslash B)}\lambda _{C}\mathbf{1}_{C}
&=&\sum_{C\in \mathcal{P}(A\backslash B)}\mu _{C}\mathbf{1}_{C}\text{\qquad
and} \\
\sum_{C\in \mathcal{P}\mathbf{(}A\backslash B)}\lambda
_{C}(p(A)-p(A\backslash C)) &\geq &\sum_{C\in \mathcal{P}\mathbf{(}%
A\backslash B)}\mu _{C}(p(B\cup C)-p(B)).
\end{eqnarray*}
\end{theorem}

\subsection{Appendix: Revenue of Deterministic Mechanisms\label{s-a:det}}

In this appendix we generalize the line of proof of Section \ref{s:sym det}
from symmetric deterministic mechanisms (which are monotonic) to general
deterministic mechanisms (which need not be monotonic). This yields, in
particular, an improvement of an $O(\log k)$ factor to the result of
Corollary A.5 in Hart and Nisan (2019), on the ratio between the
deterministic and separate revenues. We consider first the supermodular
case, and then obtain the general result by pricing function approximation.

\subsubsection{The Supermodular Case\label{sus-a:det-super}}

A\ deterministic pricing function $p:2^{K}\rightarrow \lbrack 0,\infty ]$ is 
\emph{supermodular} if $p(A)+p(B)\leq p(A\cup B)+p(A\cap B)$ for every $%
A,B\subseteq K$; let \textsc{SupermodDRev} denote the maximal revenue
achievable by supermodular deterministic mechanisms.

\begin{proposition}
\label{p:supermod-det}Let $X$ be a $k$-good random valuation. Then%
\begin{equation*}
\text{\textsc{SupermodDRev}}(X)\leq \frac{1}{k}(2^{k}-1)\cdot \text{\textsc{%
SRev}}(X).
\end{equation*}
\end{proposition}

\begin{proof}
Let $\mu =(q,s)$ be a $k$-good deterministic mechanism with canonical
deterministic pricing function $p$ that is supermodular. Without loss of
generality assume that $p(A)$ is finite for all $A$; otherwise, replacing
each $p(B)=\infty $ by, say, $p^{\prime }(B)=|B|M$, where $M>\max
\{p(A):p(A)<\infty \}\geq \max \{s(x):x\in \mathbb{R}_{+}^{k}\}$, keeps the
pricing function nondecreasing and supermodular, and can only increase the
revenue (because any switch to a new price yields a higher payment).

All sets $A,B$ below should be understood as subsets of $K$, and we will
write $p_{A}$ instead of $p(A)$. Given a random allocation $X$, for each set 
$A$ let $\beta _{A}:=\mathbb{P}\left[ q(X)=A\right] $ be the probability
that the allocation is the set $A$; then%
\begin{equation*}
R(\mu ;X)=\sum_{A\neq \emptyset }\beta _{A}p_{A}
\end{equation*}%
(because $p_{\emptyset }=0$). Consider the expression%
\begin{equation*}
Z:=\sum_{A\mathcal{\neq \emptyset }}\frac{1}{|A|}\sum_{i\in
A}(p_{A}-p_{A\backslash \{i\}})\sum_{B\supseteq A}\beta _{B}.
\end{equation*}%
In the steps below we will prove the two inequalities%
\begin{equation}
R(\mu ;X)\leq Z\leq \frac{2^{k}-1}{k}\cdot \text{\textsc{SRev}}(X),
\label{eq:ZZ}
\end{equation}%
which yield the result.

\begin{description}
\item[(i)] For every set $A\neq \emptyset $ with $p(A)<\infty $ and every $%
i\in A$ we have%
\begin{equation*}
(p_{A}-p_{A\backslash \{i\}})\sum_{B\supseteq A}\beta _{B}\leq \text{\textsc{%
Rev}}(X_{i}).
\end{equation*}
\end{description}

\noindent \emph{Proof. }If $q(x)=A$ then $x(A)-p_{A}\geq x(A\backslash
\{i\})-p_{A\backslash \{i\}}$ by IC, and so $x_{i}\geq p_{A}-p_{A\backslash
\{i\}}$. Since $p_{B}-p_{B\backslash \{i\}}\geq p_{A}-p_{A\backslash \{i\}}$
for every $B\supseteq A$ by the supermodularity of $p$, we get%
\begin{equation*}
\mathbb{P}\left[ X_{i}\geq p_{A}-p_{A\backslash \{i\}}\right] \geq
\sum_{B\supseteq A}\mathbb{P}\left[ q(X)=B\right] =\sum_{B\supseteq A}\beta
_{B};
\end{equation*}%
multiplying by $p_{A}-p_{A\backslash \{i\}}$ completes the proof. $\square $

\begin{description}
\item[(ii)] 
\begin{equation*}
Z\leq \frac{2^{k}-1}{k}\cdot \text{\textsc{SRev}}(X).
\end{equation*}
\end{description}

\noindent \emph{Proof. }By (i) we have%
\begin{equation*}
Z\leq \sum_{A\neq \emptyset }\frac{1}{|A|}\sum_{i\in A}\text{\textsc{Rev}}%
(X_{i})=\sum_{i=1}^{k}\text{\textsc{Rev}}(X_{i})\sum_{A\ni i}\frac{1}{|A|}.
\end{equation*}%
For each $i\in K$ there are $\binom{k-1}{\ell -1}$ sets $A$ of size $\ell $
that contain $i$, and so 
\begin{equation*}
\sum_{A\ni i}\frac{1}{|A|}=\sum_{\ell =1}^{k}\frac{1}{\ell }\binom{k-1}{\ell
-1}=\sum_{\ell =1}^{k}\frac{1}{k}\binom{k}{\ell }=\frac{1}{k}\left(
2^{k}-1\right) .
\end{equation*}%
$\square $

\begin{description}
\item[(iii)] 
\begin{equation*}
Z=R\left( \mu ;X\right) +\sum_{B\neq \emptyset }\beta
_{B}\sum_{A\varsubsetneqq B}p_{A}\frac{2\left\vert A\right\vert
+1-\left\vert B\right\vert }{\left\vert A\right\vert +1}.
\end{equation*}
\end{description}

\noindent \emph{Proof. }Fix $B$, and consider all the terms of $Z$ that
include $\beta _{B}$. For each $A\subseteq B$ the term $\left\vert
A\right\vert ^{-1}p_{A}\beta _{B}$ appears $\left\vert A\right\vert $ times
(once for each $i\in A)$, yielding $p_{A}\beta _{B}$ in total; and for each $%
A\varsubsetneqq B$ the term $-(\left\vert A\right\vert +1)^{-1}p_{A}\beta
_{B}$ appears $\left\vert B\right\vert -\left\vert A\right\vert $ times
(once for each $A^{\prime }\subseteq B$ such that $A=A^{\prime }\backslash
\{i\}$, i.e., once for each $i\in B\backslash A)$, yielding $-(\left\vert
B\right\vert -\left\vert A\right\vert )(\left\vert A\right\vert
+1)^{-1}p_{A}\beta _{B}$ in total. Thus,%
\begin{eqnarray*}
Z &=&\sum_{B\neq \emptyset }\beta _{B}\sum_{A\subseteq B}p_{A}\left( 1-\frac{%
\left\vert B\right\vert -\left\vert A\right\vert }{\left\vert A\right\vert +1%
}\right) \\
&=&\sum_{B\neq \emptyset }\beta _{B}p_{B}+\sum_{B\neq \emptyset }\beta
_{B}\sum_{A\varsubsetneqq B}p_{A}\left( 1-\frac{\left\vert B\right\vert
-\left\vert A\right\vert }{\left\vert A\right\vert +1}\right) ,
\end{eqnarray*}%
where in the second line we have split the interior sum into $A=B$ and $%
A\varsubsetneqq B$; this completes the proof. $\square $

\begin{description}
\item[(iv)] For every $B\neq \emptyset $ we have%
\begin{equation*}
\sum_{A\varsubsetneqq B}p_{A}\frac{2\left\vert A\right\vert +1-\left\vert
B\right\vert }{\left\vert A\right\vert +1}\geq 0.
\end{equation*}
\end{description}

\noindent \emph{Proof. }Let $n:=|B|$; for each $m=1,...,n-1$ put%
\begin{eqnarray*}
\pi _{m} &%
%TCIMACRO{\TeXButton{:=}{{\;:=\;}}}%
%BeginExpansion
{\;:=\;}%
%EndExpansion
&\binom{n}{m}^{-1}\sum_{A\subseteq B:\left\vert A\right\vert =m}p_{A}\text{\
\ \ and} \\
\lambda _{m} &%
%TCIMACRO{\TeXButton{:=}{{\;:=\;}}}%
%BeginExpansion
{\;:=\;}%
%EndExpansion
&\frac{2m+1-n}{m+1}\binom{n}{m};
\end{eqnarray*}%
we need to show that%
\begin{equation*}
\sum_{m=1}^{n-1}\lambda _{m}\pi _{m}\geq 0.
\end{equation*}%
We will show that each term $\lambda _{m}\pi _{m}$ whose coefficient $%
\lambda _{m}$ is negative (i.e., when $m<(n-1)/2)$ is \textquotedblleft
covered" by the corresponding term $\lambda _{n-m-1}\pi _{n-m-1}$ (whose
coefficient is positive); i.e.,%
\begin{equation}
\lambda _{m}\pi _{m}+\lambda _{n-m-1}\pi _{n-m-1}\geq 0.
\label{eq:lambda-pi}
\end{equation}%
Indeed, for each $m<(n-1)/2$ we have%
\begin{equation*}
\left\vert \lambda _{m}\right\vert =\frac{n-2m-1}{m+1}\binom{n}{m}=\frac{%
n-2m-1}{n-m}\binom{n}{n-m-1}=\lambda _{n-m-1};
\end{equation*}%
together with $\pi _{m}\leq \pi _{n-1-m}$, which we will prove in (v) below,
this yields (\ref{eq:lambda-pi}). Summing over all $m<(n-1)/2$ completes the
proof. $\square $

\begin{description}
\item[(v)] Let $N$ be a set of size $n$, and for each $m=0,1,...,n$ let $\pi
_{m}$ be the average price of subsets of $N$ of size $m$, i.e.,%
\begin{equation*}
\pi _{m}:=\binom{n}{m}^{-1}\sum_{C\subseteq N:\left\vert C\right\vert
=m}p_{C}.
\end{equation*}%
Then $\pi _{m}$ is a nondecreasing function of $m.$
\end{description}

\noindent \emph{Proof. }Let $m<n$; we will show that $\pi _{m}\leq \pi
_{m+1} $. The function $p$ is nondecreasing,\footnote{%
This fact by itself does not yield the result immediately; recall Simpson's
paradox.} and so%
\begin{equation*}
\sum_{C}\sum_{D}p_{C}\leq \sum_{C}\sum_{D}p_{D},
\end{equation*}%
where both sums range over all pairs $(C,D)$ such that $\left\vert
C\right\vert =m$, $\left\vert D\right\vert =m+1$, and $C\subset D$. Each $%
p_{C}$ appears in the left-hand sum $n-m$ times (once for each $i\in
N\backslash C)$, and each $p_{D}$ appears in the right-hand sum $m+1$ times
(once for each $i\in D)$. Therefore%
\begin{equation*}
(n-m)\binom{n}{m}\pi _{m}\leq (m+1)\binom{n}{m+1}\pi _{m+1}.
\end{equation*}%
Since $(n-m)\binom{n}{m}=(m+1)\binom{n}{m+1}$, we get $\pi _{m}\leq \pi
_{m+1}$. $\square $

Results (ii), (iii), and (iv) yield (\ref{eq:ZZ}), and the proof is complete.
\end{proof}

\bigskip

\noindent \textbf{Remark. }The result of Proposition \ref{p:supermod-det} is
tight: the pricing functions in the proof of Proposition 7.1 in Hart and
Nisan (2019)---which are used in Proposition A.10 there to get the ratio 
\textsc{DRev}$(X)/$\textsc{SRev}$(X)$ close to $(2^{k}-1)/k$---are
supermodular. Indeed, the sequence $t_{n}$ constructed there increases
exponentially, i.e., $t_{n}\leq \varepsilon \cdot t_{n+1}$ for all $n$,
where $\varepsilon >0$ is small; hence, by taking $\varepsilon \leq 1/2$, in
the nontrivial case where $A,B\subsetneqq A\cup B$ we get $p(A),p(B)\leq
(1/2)p(A\cup B)$, and so $p(A)+p(B)\leq p(A\cup B)\leq p(A\cup B)+p(A\cap B)$%
.

\subsubsection{The General Case\label{sus-a:det-general}}

We approximate any deterministic pricing function by a supermodular one.

\begin{proposition}
\label{p:det-p-approx}For every deterministic pricing function $p$ on $k$
goods there exists a supermodular pricing function $p^{\prime }$ such that%
\begin{equation*}
p(A)\leq p^{\prime }(A)\leq 2^{k-1}p(A)
\end{equation*}%
for every $A\subseteq K.$
\end{proposition}

\begin{proof}
The minimal supermodular function $p^{\prime }$ that majorizes $p$ (i.e., $%
p^{\prime }\geq p$) is constructed inductively, starting with $p^{\prime
}(\emptyset ):=0$ and $p^{\prime }(\{i\}):=p(\{i\})$ for every $i\in K$, by
letting%
\begin{equation*}
p^{\prime }(A)%
%TCIMACRO{\TeXButton{:=}{{\;:=\;}}}%
%BeginExpansion
{\;:=\;}%
%EndExpansion
\max \{p(A),\tilde{p}(A)\}
\end{equation*}%
where%
\begin{equation*}
\tilde{p}(A)%
%TCIMACRO{\TeXButton{:=}{{\;:=\;}}}%
%BeginExpansion
{\;:=\;}%
%EndExpansion
\max_{i,j\in A,\,i\neq j}\{p^{\prime }(A\backslash \{i\})+p^{\prime
}(A\backslash \{j\})-p^{\prime }(A\backslash \{i,j\})\}
\end{equation*}%
for every $A\subseteq K$ with $\left\vert A\right\vert \geq 2$ (indeed, the
function $p^{\prime }$ is clearly supermodular, and $p^{\prime }\geq p$; it
is straightforward to show inductively that any supermodular $p^{\prime
\prime }$ such that $p^{\prime \prime }\geq p$ must satisfy $p^{\prime
\prime }\geq p^{\prime }$).

We claim that 
\begin{equation}
p^{\prime }(A)\leq 2^{\left\vert A\right\vert -1}p(A)  \label{eq:2^(k-1)}
\end{equation}%
for all $A\neq \emptyset $. We prove this by induction. For singleton sets $%
A=\{i\}$ we have equality. Take $A$ with $\left\vert A\right\vert \geq 2$,
and let the maximum in the definition of $\tilde{p}(A)$ be attained at a
certain pair $i\neq j$; then%
\begin{eqnarray*}
\tilde{p}(A) &\leq &p^{\prime }(A\backslash \{i\})+p^{\prime }(A\backslash
\{j\})\leq 2^{\left\vert A\right\vert -2}p(A\backslash \{i\})+2^{\left\vert
A\right\vert -2}p(A\backslash \{j\}) \\
&\leq &2^{\left\vert A\right\vert -2}p(A)+2^{\left\vert A\right\vert
-2}p(A)=2^{\left\vert A\right\vert -1}p(A)
\end{eqnarray*}%
(where we have used: $p\geq 0$; (\ref{eq:2^(k-1)}) for the smaller sets $%
A\backslash \{i\}$ and $A\backslash \{j\}$; and $p$ being nondecreasing).
This yields (\ref{eq:2^(k-1)}) for $A.$
\end{proof}

\bigskip

The bound of $2^{O(k)}$ is tight.

\textbf{Example.} Let $k$ be even, say $k=2m$, and partition the $k$ goods
into the $m$ disjoint pairs $K_{1}=\{1,2\},K_{2}=\{3,4\},...,K_{m}=\{2m-1,2m%
\}$. Let $p(A)=1$ if $A$ contains at least one element of each $K_{\ell }$
(i.e., $\left\vert A\cap K_{\ell }\right\vert \geq 1$ for all $\ell
=1,...,m) $, and $0$ otherwise. We claim that the minimal supermodular $%
p^{\prime }$ that majorizes $p$ is given by%
\begin{equation}
p^{\prime }(A):=\prod\nolimits_{\ell =1}^{m}\left\vert A\cap K_{\ell
}\right\vert ,  \label{eq:Al}
\end{equation}%
and so, in particular,%
\begin{equation*}
p^{\prime }(K)=2^{m}=2^{k/2}p(K).
\end{equation*}

To show this, let $p^{\prime \prime }$ be the minimal supermodular function
that majorizes $p$. We will show that $p^{\prime \prime }\geq p^{\prime }$;
since $p^{\prime }\geq p$ and $p^{\prime }$ is supermodular (as is easy to
verify), it follows that $p^{\prime \prime }=p^{\prime }.$

For every set $A\subseteq K$ put $a_{\ell }:=\left\vert A\cap K_{\ell
}\right\vert $, and let $a:=(a_{1},...,a_{m})\in \{0,1,2\}^{m}$ be the
\textquotedblleft profile" of $A$. The pricing functions $p$, and thus $%
p^{\prime \prime }$ as well as $p^{\prime }$, depend only on the profile $a$%
; we will abuse notation and write $p(a)$ instead of $p(A)$, and similarly
for $p^{\prime }$ and $p^{\prime \prime }.$

In the domain where at least one coordinate of $a$ is $0$ the function $p$
vanishes, and so it is supermodular, which yields $p^{\prime \prime
}(a)=0=p^{\prime }(a).$

In the remaining domain, where $a\geq (1,...,1)$, the function $p$ is the
constant function $1$. Starting with $p^{\prime \prime
}(1,...,1)=p(1,...,1)=1$, we proceed by induction on $n:=\left\vert \{\ell
:a_{\ell }=2\}\right\vert $, the number of coordinates in $a$ that equal $2$%
. The supermodularity condition yields%
\begin{equation*}
p^{\prime \prime }(2,\ast )+p^{\prime \prime }(0,\ast )\geq 2p^{\prime
\prime }(1,\ast )
\end{equation*}%
for every completion $\ast $ (for example: $p^{\prime \prime }(B\cup
\{1,2\})+p^{\prime \prime }(B)\geq p^{\prime \prime }(B\cup \{1\})+p^{\prime
\prime }(B\cup \{2\})$ for every $B\subseteq \{3,...,2m\})$. Since $%
p^{\prime \prime }(0,\ast )=p(0,\ast )=0$, we get $p^{\prime \prime }(2,\ast
)\geq 2p^{\prime \prime }(1,\ast )$, and thus, by induction, $p^{\prime
\prime }(a)\geq 2^{n}p^{\prime \prime }(1,...,1)=2^{n}=p^{\prime }(a)$
(recall that $n$ is the number of $2$'s in $a)$. Thus $p^{\prime \prime
}\geq p^{\prime }$ as claimed.

\bigskip

Corollary A.5 of Hart and Nisan (2019) yields \textsc{DRev}$(X)\leq
O(2^{k}\log k)\cdot \,$\textsc{SRev}$(X)$; we improve this by a factor of $%
\log k.$

\begin{proposition}
\label{p:det vs srev}Let $X$ be a $k$-good random allocation. Then%
\begin{equation*}
\text{\textsc{DRev}}(X)\leq \ln 4\cdot (2^{k}-1)\cdot \text{\textsc{SRev}}%
(X).
\end{equation*}
\end{proposition}

\begin{proof}
Proposition \ref{p:det-p-approx} yields, by the result of Chawla, Teng, and
Tzamos (2022) (Theorem \ref{th:chawla}), an approximation factor of $2\ln
(2\cdot 2^{k-1})=\ln 4\cdot k$, namely,%
\begin{equation*}
\text{\textsc{DRev}}(X)\leq \ln 4\cdot k\cdot \text{\textsc{SupermodDRev}}(X)
\end{equation*}%
for all $X$. Combine this with the result of Proposition \ref{p:supermod-det}%
.
\end{proof}

\end{document}